\documentclass[11pt,reqno]{amsart}

\usepackage{mathrsfs}
\usepackage{amsmath}
\usepackage{amssymb}
\usepackage{amsfonts}
\usepackage{amsthm}
\usepackage{a4wide}
\usepackage{graphicx}
\usepackage{color}
\usepackage[sans]{dsfont}
\usepackage{linearA}
\usepackage{pgfkeys}
\usepackage{longtable}
\usepackage{pdflscape}
\usepackage{float}
\usepackage{cancel}
\usepackage{dsfont}
\usepackage{setspace}
\usepackage{array}
\usepackage{hyperref}
\hypersetup{colorlinks=true,allcolors=[rgb]{0,0,0.6}}
\usepackage[english]{babel}
\usepackage[latin1]{inputenc}
\usepackage[T1]{fontenc}
\usepackage{tikz,pgflibraryshapes}
\usepackage{enumitem}
\usepackage{eufrak}
\usepackage[square,numbers,sort&compress]{natbib}
\usepackage[shadow,textwidth=22mm,textsize=tiny]{todonotes}

\newtheorem{theorem}{Theorem}[section]
\newtheorem{lemma}[theorem]{Lemma}
\newtheorem{proposition}[theorem]{Proposition}
\newtheorem{corollary}[theorem]{Corollary}
\newtheorem{definition}{Definition}

\newtheorem{remark}[theorem]{Remark}
\newtheorem{assumption}[theorem]{Hypothesis}

\numberwithin{equation}{section}
\DeclareMathOperator{\slim}{s-lim}
\DeclareMathOperator{\wlim}{w-lim}

\author[J. Faupin]{J{\'e}r{\'e}my Faupin}
\address[J. Faupin]{Institut Elie Cartan de Lorraine \\
Universit{\'e} de Lorraine, 
57045 Metz Cedex 1, France}
\email{jeremy.faupin@univ-lorraine.fr}

\author[J. Fr\"ohlich]{J\"urg Fr\"ohlich}
\address[J. Fr{\"o}hlich]{Institut f{\"u}r Theoretische Physik, ETH H{\"o}nggerberg, CH-8093 Z{\"u}rich, Switzerland}
\email{juerg@phys.ethz.ch}

\begin{document}
\bibliographystyle{abbrv} \title[Dissipative Scattering Theory]{Asymptotic completeness in dissipative scattering theory}

\begin{abstract}
We consider an abstract pseudo-Hamiltonian for the nuclear optical model, given by a dissipative operator of the form $H = H_V - i C^* C$, where $H_V = H_0 + V$ is self-adjoint and $C$ is a bounded operator. We study the wave operators associated to $H$ and $H_0$. We prove that they are asymptotically complete if and only if $H$ does not have spectral singularities on the real axis. For Schr{\"o}dinger operators, the spectral singularities correspond to real resonances.
\end{abstract}

\maketitle

\section{Introduction}

In this paper we study the quantum-mechanical scattering theory for dissipative quantum systems. A typical example is a neutron interacting with a nucleus. When a neutron is targeted onto a complex nucleus, it may, after interacting with it, be elastically scattered off the nucleus or be absorbed by the nucleus, leading to the formation of a compound nucleus. The concept of a compound nucleus was introduced by Bohr \cite{Bo36_01}.

In \cite{FePoWe54_01}, Feshbach, Porter and Weisskopf proposed a model describing the interaction of a neutron with a nucleus, allowing for the description of both elastic scattering and the formation of a compound nucleus. The force exerted by the nucleus on the neutron is modeled by a phenomenological potential of the form $V - i W$, where $V$, $W$ are real-valued and $W \ge 0$. The nucleus is supposed to be localized in space, which corresponds to the assumption that $V$ and $W$ are compactly supported or decay rapidly at infinity. On $L^2( \mathbb{R}^3 )$, the pseudo-Hamiltonian for the neutron is given by
\begin{equation}\label{eq:intro1}
H = - \Delta + V - i W. 
\end{equation}
In the following, a linear operator $H$ is called a pseudo-Hamiltonian if $-iH$ generates a strongly continuous contractive semigroup $\{ e^{ - i t H } \}_{Êt \ge 0 }$. For any initial state $u_0$, with $\| u_0 \| = 1$, the map $t \mapsto \| e^{ - i t H }Êu_0 \|$ is decreasing on $[ 0 , \infty )$, and the quantity
\begin{equation}\label{eq:defpabs}
 p_{ \mathrm{abs}Ê} := 1 - \lim_{ t \to \infty } \big \| e^{ - i t H } u_0 \big \|  
\end{equation}
gives the probability of absorption of the neutron by the nucleus, i.e., the probability of formation of a compound nucleus. The probability that the neutron, initially in the state $u_0$, eventually escapes from the nucleus is given by $p_{ \mathrm{scat}Ê} :=\lim_{ t \to \infty } \| e^{ - i t H } u_0 \|$, and in the case where this probability is strictly positive, one expects that there exists an (unnormalized) scattering state $u_+$ such that $\|Êu_+ \|Ê= p_{ \mathrm{scat}Ê}$ and
\begin{equation}
\lim_{ t \to \infty } \big \| e^{ - i t H } u_0 - e^{ i t \Delta } u_+ \big \| = 0 . \label{eq:intro_ex_scatt}
\end{equation}

This model is referred to as the nuclear optical model, the term optical being used in reference to the phenomenon in optics of refraction and absorption of light waves by a medium. The model is empirical in that the form of the potentials $V$ and $W$ are determined by optimizing the fit to experimental data. Usually, $V$ and $W$ are decomposed into a sum of terms corresponding to the form of the expected interaction potentials in different regions of physical space, and sometimes a spin-orbit interaction term is included. We refer to e.g. \cite{Ho71_01} or \cite{Fe92_01} for a thorough description. A large range of observed scattering data can then be predicted by the model to a high precision.

Since the explicit expression of the pseudo-Hamiltonian rests on experimental scattering data, it is desirable to develop the full scattering theory of a class of models, in order to justify their use from a theoretical point of view. In this paper, we consider an abstract pseudo-Hamiltonian generalizing \eqref{eq:intro1}, of the form
\begin{equation}\label{eq:exprH}
H := H_0 + V -  i C^* C .
\end{equation}
Under natural assumptions, \eqref{eq:exprH} defines a dissipative operator acting on a Hilbert space, generating a strongly continuous semigroup of contractions. Our hypotheses on $H_0$, $V$ and $C$ will be formulated in such a way that they can be verified in the particular case where $H$ is given by \eqref{eq:intro1}.

Mathematical scattering theory for dissipative operators on Hilbert spaces has been considered by many authors. We mention here works, related to ours, by Martin \cite{Martin}, Davies \cite{Davies4,Davies1} and Neidhardt \cite{Ne85_01}, for general abstract results, Mochizuki \cite{Mo68_01} and Simon \cite{Simon2}, for Schr{\"o}dinger operators of the form \eqref{eq:intro1}, and by Kato \cite{Kato1}, Wang and Zhu \cite{WangZhu}, and Falconi, Schubnel and the authors \cite{FaFaFrSc17_01} for ``weak coupling'' results. The existence of the wave operators associated to $H$ and $H_0$ is established under various conditions. But proving their asymptotic completeness is a much more difficult problem which, to our knowledge, is solved only in some particular cases; (see \cite{Kato1,WangZhu,FaFaFrSc17_01} for weak coupling results, and, e.g., Stepin \cite{St04_01}, for some models in one dimension). We will recall the definition of the wave operators and the notion of asymptotic completeness in the next section.

Scattering theory for dissipative operators on Hilbert spaces also has important applications in the scattering theory of Lindblad master equations \cite{Davies2,FaFaFrSc17_01}. If one considers a particle interacting with a dynamical target and takes a trace over the degrees of freedom of the target, it is known that, in the kinetic limit, the reduced effective dynamics of the particle is given by a quantum dynamical semigroup generated by a Lindbladian. Scattering theory for Lindblad master equations provides an alternative approach to studying the phenomenon of capture. For quantum dynamical semigroups, the probability of particle capture is given by the difference between $1$ and the trace of a certain wave operator $\Omega$ applied to the initial state of the particle, \cite{Davies2}. The definition of $\Omega$ and the proof of its existence rest on the scattering properties of a dissipative operator of the form \eqref{eq:exprH}. We will outline the consequences of our results for the scattering theory of Lindblad master equations in Section \ref{sec:Lindblad}.

Summary of main results: Under suitable assumptions on the abstract pseudo-Hamiltonian \eqref{eq:exprH}, we prove that the space of initial states for which the probability of absorption $p_{Ê\mathrm{abs}Ê}$ in \eqref{eq:defpabs} is equal to $1$ coincides with the subspace spanned by the generalized eigenvectors of $H$ corresponding to non-real eigenvalues. For any initial state $u_0$ orthogonal to all the generalized eigenstates of $H$, we show that there exists a scattering state $u_+ \neq 0$ satisfying \eqref{eq:intro_ex_scatt}. Using these results, we prove that asymptotic completeness holds if and only if $H$ does not have ``spectral singularities'' on the real axis. Asymptotic completeness implies that the restriction of $H$ to the orthogonal complement of the subspace spanned by the generalized eigenvectors of the adjoint operator $H^*$ is similar to $H_0$. Our definition of a spectral singularity is related to that of J. Schwartz \cite{Sc60_01} and corresponds to a real resonance in the case of Schr{\"o}dinger operators of the form \eqref{eq:intro1}. 

In the next section we describe the model that we consider and we state our results in precise form.

\section{Hypotheses and statement of the main results}

\subsection{The model}

Let $\mathcal{H}$ be a complex separable Hilbert space. On $\mathcal{H}$, we consider the operator \eqref{eq:exprH}, where $H_0$ is self-adjoint and bounded from below, $V$ is symmetric and $C \in \mathcal{L}( \mathcal{H} )$. Without loss of generality, we suppose that $H_0 \ge 0$. Moreover, we assume that $V$ and $C^* C$ are relatively compact with respect to $H_0$ so that, in particular,
\begin{equation*}
H_V := H_0 + V,
\end{equation*}
is self-adjoint on $\mathcal{H}$, with domain $\mathcal{D}( H_V ) = \mathcal{D}( H_0 )$, and $H$ is a closed maximal dissipative operator with domain $\mathcal{D}( H ) = \mathcal{D}( H_0 )$.

That $H$ is dissipative follows from the observation that 
\begin{equation*}
\mathrm{Im}( \langle u , H u \rangle ) = - \| C u \|^2 \le 0 ,
\end{equation*}
for all $u \in \mathcal{D}( H )$. This implies (see e.g. \cite{EnNa20_01} or \cite{Da07_01}) that the spectrum of $H$ is contained in the lower half-plane, $\{ z \in \mathbb{C} , \mathrm{Im}( z ) \le 0 \}$, and that $ - i H$ is the generator of a strongly continuous one-parameter semigroup of contractions $\{ e^{ - i t H } \}_{ t \ge 0 }$. In fact, since $H$ is a perturbation of the self-adjoint operator $H_V$ by the bounded operator $- i C^* C $, $- i H$ generates a group $\{ e^{ - i t H }Ê\}_{Êt \in \mathbb{R} }$ satisfying 
\begin{equation*}
\big \|Êe^{ - i t H }Ê\big \| \le 1 , \, \, t \ge 0 , \qquad \qquad \big \| e^{ - i t H }Ê\big \| \le e^{ \| C^* C \|Ê| t |Ê} , \, \, t \le 0 ,
\end{equation*}
(see \cite{EnNa20_01} or \cite{Da07_01}).

Let $\sigma( H )$ denote the spectrum of $H$. Because $V$ and $C^*C$ are relatively compact perturbations of $H_0$, the essential spectrum of $H$, denoted by $\sigma_{Ê\mathrm{ess}Ê}( H )$, coincides with the essential spectrum of $H_0$; (see Section \ref{subsec:spectrum} for the definition of the essential spectrum of a closed operator). Moreover $\sigma( H ) \setminus \sigma_{Ê\mathrm{ess} }( H )$ consists of an at most countable number of eigenvalues of finite algebraic multiplicities, that can only accumulate at points of $\sigma_{ \mathrm{ess} }( H )$. See Figure \ref{fig1}.

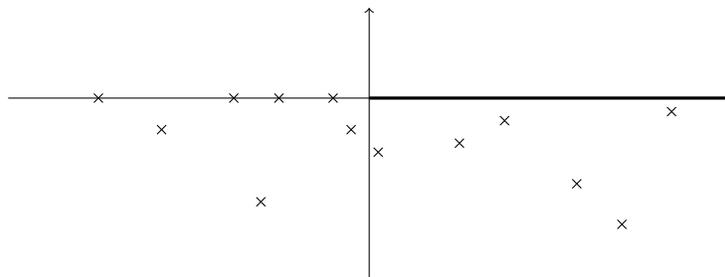
\begin{figure}[H] 
\begin{center}
\begin{tikzpicture}[scale=0.6, every node/.style={scale=0.9}]

   \draw[->](-4,2) -- (12,2);      
      \draw[->](4,-2) -- (4,4);      
  \draw[-, very thick] (4,2)--(11.95,2);
       
    \draw[-] (1.9,1.9)--(2.1,2.1);
    \draw[-] (1.9,2.1)--(2.1,1.9);

    \draw[-] (-1.9,1.9)--(-2.1,2.1);
    \draw[-] (-1.9,2.1)--(-2.1,1.9);

    \draw[-] (0.9,1.9)--(1.1,2.1);
    \draw[-] (0.9,2.1)--(1.1,1.9);

    \draw[-] (3.1,1.9)--(3.3,2.1);
    \draw[-] (3.1,2.1)--(3.3,1.9);

    \draw[-] (5.9,0.9)--(6.1,1.1);
    \draw[-] (5.9,1.1)--(6.1,0.9);

    \draw[-] (6.9,1.4)--(7.1,1.6);
    \draw[-] (6.9,1.6)--(7.1,1.4);

    \draw[-] (10.6,1.6)--(10.8,1.8);
    \draw[-] (10.6,1.8)--(10.8,1.6);

    \draw[-] (9.5,-0.9)--(9.7,-0.7);
    \draw[-] (9.5,-0.7)--(9.7,-0.9);

    \draw[-] (8.5,0)--(8.7,0.2);
    \draw[-] (8.5,0.2)--(8.7,0);

    \draw[-] (4.1,0.9)--(4.3,0.7);
    \draw[-] (4.1,0.7)--(4.3,0.9);

    \draw[-] (3.5,1.2)--(3.7,1.4);
    \draw[-] (3.5,1.4)--(3.7,1.2);

    \draw[-] (-0.5,1.2)--(-0.7,1.4);
    \draw[-] (-0.5,1.4)--(-0.7,1.2);

    \draw[-] (1.5,-0.2)--(1.7,-0.4);
    \draw[-] (1.5,-0.4)--(1.7,-0.2);

        \end{tikzpicture}
\caption{ \footnotesize  \textbf{Form of the spectrum of $H$.} The essential spectrum coincides with that of $H_0$ and is contained in $[ 0 , \infty )$. The eigenvalues on the real axis are negative, of finite algebraic multiplicities and associated to eigenvectors belonging to $ \mathcal{H}_{Ê\mathrm{b} }( H )$ (see \eqref{eq:defHb}). The eigenvalues with strictly negative imaginary parts have finite algebraic multiplicities and are associated to generalized eigenvectors belonging to $\mathcal{H}_{ \mathrm{p} }( H )$ (see \eqref{eq:defHp}). }\label{fig1}
\end{center}
\end{figure}

Before stating our main results, we introduce some notations (Section \ref{subsec:subspaces}) and our main hypotheses (Section \ref{subsec:hypoth}).

\subsection{Spectral subspaces}\label{subsec:subspaces}
The space of bound states, 
\begin{equation}\label{eq:defHb}
\mathcal{H}_{ \mathrm{b} }( H ) := \mathrm{Span} \big \{ u \in \mathcal{D}( H ) , \, \exists \lambda \in \mathbb{R} , \, H u = \lambda u \big \} ,
\end{equation}
is the vector space spanned by the set of eigenvectors of $H$ corresponding to real eigenvalues. Note that $\mathcal{H}_{ \mathrm{b} }( H )$ is usually defined as the closure of $\mathrm{Span} \big \{ u \in \mathcal{D}( H ) , \, \exists \lambda \in \mathbb{R} , \, H u = \lambda u \big \}$ \cite{Davies1}. But it will be observed in Section \ref{subsec:spectrum} that this vector space is actually closed under our assumptions.

For $\lambda \in \sigma( H ) \setminus \sigma_{ \mathrm{ess} }( H )$, we denote by 
\begin{equation}\label{eq:defPilambda}
\Pi_\lambda := \frac{1}{ 2 i \pi } \int_\gamma ( z \mathrm{Id} - H )^{-1} dz ,
\end{equation}
the usual Riesz projection, where $\gamma$ is a circle oriented counterclockwise and centered at $\lambda$, of sufficiently small radius (so that $\lambda$ is the only point of the spectrum of $H$ contained in the interior of $\gamma$). The algebraic multiplicity of $\lambda$ is $\mathrm{dim} \, \mathrm{Ran}( \Pi_\lambda )$. Since $H$ is not self-adjoint, its restriction to $\mathrm{Ran}( \Pi_\lambda )$ may have a nontrivial Jordan form, and $\mathrm{Ran}( \Pi_\lambda )$ is in general spanned by generalized eigenvectors of $H$ associated to $\lambda$, i.e., by vectors $u \in \mathcal{D}( H^k )$ such that $( H - \lambda )^k u = 0$ for some $1 \le k \le \mathrm{dim} \, \mathrm{Ran}( \Pi_\lambda )$. We set
\begin{equation}\label{eq:defHp}
\mathcal{H}_{ \mathrm{p} }( H ) := \mathrm{Span} \, \big \{ u \in \mathrm{Ran}( \Pi_\lambda ) , \, \lambda \in \sigma( H ) , \, \mathrm{Im} \, \lambda < 0  \big \} .
\end{equation}
If $H$ has only a finite number of eigenvalues with strictly negative imaginary parts the vector space $\mathcal{H}_{ \mathrm{p} }( H )$ is closed. Moreover, defining the ``dissipative space'' $\mathcal{H}_{ \mathrm{d} }( H )$ by
\begin{align}\label{eq:defHd}
& \mathcal{H}_{ \mathrm{d} }( H ) := \big \{ u \in \mathcal{H} , \lim_{ t \to \infty } \|Êe^{ - i t H }Êu \| = 0 \big \} ,
\end{align}
we have the obvious inclusion
\begin{equation*}
\mathcal{H}_{Ê\mathrm{p} }( H ) \subseteq \mathcal{H}_{ \mathrm{d} }( H ).
\end{equation*}
It is easy to verify that the vector space $\mathcal{H}_{ \mathrm{d} }( H )$ is closed.

An important role in our analysis will be played by the adjoint operator
\begin{equation}
H^* = H_0 + V + i C^* C = H_V + i C^* C.
\end{equation}
Note that $\lambda \in \sigma ( H ^* )$ if and only if $\bar \lambda \in \sigma ( H )$, and that $i H^*$ generates the contraction semigroup $\{ e^{ i t H^* } \}_{Êt \ge 0 }$. The spaces $\mathcal{H}_{Ê\mathrm{b} }( H^* )$, $\mathcal{H}_{Ê\mathrm{p} }( H^* )$ and $\mathcal{H}_{Ê\mathrm{d} }( H^* )$ are defined in the same was as for $H$, namely
\begin{align}
& \mathcal{H}_{ \mathrm{b} }( H^* ) := \mathrm{Span} \big \{ u \in \mathcal{D}( H ) , \, \exists \lambda \in \mathbb{R} , \, H^* u = \lambda u \big \} , \label{eq:defH*b} \\
& \mathcal{H}_{ \mathrm{p} }( H^* ) := \mathrm{Span} \, \big \{ u \in \mathrm{Ran}( \Pi^*_\lambda ) , \, \lambda \in \sigma( H^* ) , \, \mathrm{Im} \, \lambda > 0  \big \} ,\label{eq:defH^*p} \\
& \mathcal{H}_{ \mathrm{d} }( H^* ) := \big \{ u \in \mathcal{H} , \lim_{ t \to \infty } \|Êe^{ i t H^* }Êu \| = 0 \big \}. \label{eq:defH*d}
\end{align}
In \eqref{eq:defH^*p}, $\Pi^*_\lambda$ stands for the Riesz projection associated to $\lambda$ for $H^*$, i.e.
\begin{equation}
\Pi^*_\lambda := \frac{1}{ 2 i \pi } \int_\gamma ( z \mathrm{Id} - H^* )^{-1} dz , \label{eq:pilambda*}
\end{equation}
where $\gamma$ is any circle oriented counterclockwise and centered at $\lambda$, of sufficiently small radius.

Further properties of the subspaces introduced in this section are discussed in Section \ref{sec:prelim}.

\subsection{Hypotheses}\label{subsec:hypoth}

Our first hypothesis concerns the spectra of the self-adjoint operators $H_0$ and $H_V$.
\begin{assumption}[Spectra of $H_0$ and $H_V$]\label{V-1}
The spectrum of $H_0$ is purely absolutely continuous, the singular continuous spectrum of $H_V$ is empty, $H_V$ has at most finitely many eigenvalues of finite multiplicity, and each eigenvalue of $H_V$ is strictly negative.
\end{assumption}
The assumptions that the number of eigenvalues of $H_V$ is finite and that $H_V$ has no embedded eigenvalues are mostly made for the purpose of simplicity of exposition. It is likely that these assumptions can be relaxed.

It will sometimes be convenient to add the following assumption:
\begin{assumption}[Eigenvalues of $H$]\label{V2}
The number of non-real eigenvalues of $H$ is finite.
\end{assumption}

Our next hypothesis concerns the wave operators for the \emph{self-adjoint} operators $H_V$ and $H_0$.
\begin{assumption}[Wave operators for $H_V$ and $H_0$] \label{V0}
The wave operators 
\begin{equation*}
W_\pm( H_V , H_0 ) := \underset{t\to \pm \infty }{\slim} e^{ i t H_V } e^{ - i t H_0 } \quad \text{ and } \quad W_{Ê\pm } ( H_0 , H_V ) := \underset{t\to \pm \infty }{\slim} e^{ i t H_0 } e^{ - i t H_V } \Pi_{ \mathrm{ac} }( H_V ) 
\end{equation*}
exist and are asymptotically complete, i.e.,
\begin{align*}
& \mathrm{Ran} ( W_\pm( H_V , H_0 ) ) = \mathcal{H}_{ \mathrm{ac} }( H_V ) = \mathcal{H}_{ \mathrm{pp} }( H_V )^\perp , \\
& \mathrm{Ran} ( W_\pm( H_0 , H_V ) ) = \mathcal{H} .
\end{align*}
Here $\mathcal{H}_{ \mathrm{ac} }( H_V )$ and $\mathcal{H}_{ \mathrm{pp} }( H_V )$ denote the absolutely continuous and pure-point spectral subspaces of $H_V$, respectively, and $\Pi_{ \mathrm{ac} }( H_V )$ denotes the orthogonal projection onto $\mathcal{H}_{ \mathrm{ac}Ê}( H_V )$.
\end{assumption}
In our next assumption we require that $C$ be relatively smooth with respect to $H_V$, in the sense of Kato \cite{Kato1}.
\begin{assumption}[Relative smoothness of $C$ with respect to $H_V$]\label{V1}
There exists a constant $\mathrm{c}_V > 0$, such that
\begin{align}\label{eq:ZV}
\int_{ \mathbb{R} } \big \|ÊC e^{ - i t H_V } \Pi_{ \mathrm{ac} }( H_V ) u \big \|^2 dt \le \mathrm{c}_V^2 \| \Pi_{ \mathrm{ac} }( H_V ) u \|^2 ,
\end{align}
for all $u \in \mathcal{H}$.
\end{assumption}
We recall that \eqref{eq:ZV} is equivalent to
\begin{align}
\int_{ \mathbb{R} } \Big ( \big \|ÊC \big ( H_V - ( \lambda + i 0^+ ) \big )^{-1} u \big \|^2 + \big \|ÊC \big ( H_V - ( \lambda - i 0^+ ) \big )^{-1} u \big \|^2 \Big ) d \lambda \le 2 \pi \mathrm{c}_V^2 \| u \|^2 , \label{eq:equivV1}
\end{align}
for all $u \in \mathrm{Ran} ( \Pi_{ \mathrm{ac} }( H_V ) )$. See \cite{Kato1}.

Before stating our last assumption, we introduce the following definition.
\begin{definition}\label{def:spec-sing}
Let $\lambda \in [ 0 , \infty )$. We say that $\lambda$ is a regular spectral point of $H$ if there exists a compact interval $K_\lambda \subset \mathbb{R}$ whose interior contains $\lambda$ and such that the limit
\begin{equation*}
C \big ( H - ( \mu - i 0^+ ) \big )^{-1} C^* := \lim_{ \varepsilon \downarrow 0 } C \big ( H - ( \mu - i \varepsilon ) \big )^{-1} C^* 
\end{equation*}
exists uniformly in $\mu \in K_\lambda$ in the norm topology of $\mathcal{L}( \mathcal{H} )$. If $\lambda$ is not a regular spectral point of $H$, we say that $\lambda$ is a ``\emph{spectral singularity}'' of $H$.
\end{definition}
It should be noted that only the limit of the resolvent as the spectral parameter approaches $[ 0 , \infty )$ from below is considered in the previous definition. Due to the fact that $H$ is dissipative, the limit of the resolvent on $[ 0 , \infty )$ from above is, in a sense that will be made precise in Section \ref{sec:main}, automatically well-defined. It is implicitly assumed in Definition \ref{def:spec-sing} that the resolvents $( H - ( \mu - i \varepsilon ) )^{-1}$ are well-defined for $\varepsilon > 0$ small enough. In particular, $\lambda$ cannot be an accumulation point of eigenvalues of $H$ located in $\lambda - i ( 0 , \infty )$. Our definition of a spectral singularity of $H$ is related to that of \cite{Sc60_01} and to the notion of spectral projections for non-self-adjoint operators \cite{Du58_01}. More details will be given in Section \ref{sec:main}.

In our last assumption we suppose that $H$ has only finitely many spectral singularities and that each spectral singularity is of ``finite order''.
\begin{assumption}[Spectral singularities of $H$]\label{V3}
$H$ has a finite number of spectral singularities $\{ \lambda_1 , \dots , \lambda_n \} \subset [ 0 , \infty )$ and, for each spectral singularity $\lambda_j \in [ 0 , \infty )$, there exist an integer $\nu_j > 0$ and a compact interval $K_{\lambda_j}$, whose interior contains $\lambda_j$, such that the limit
\begin{equation*}
\lim_{ \varepsilon \downarrow 0 }  | \mu - \lambda_j |^{\nu_j} \big \|ÊC \big ( H - ( \mu - i \varepsilon ) \big )^{-1} C^* \big \| 
\end{equation*}
exists uniformly in $\mu \in K_{\lambda_j}$ in the norm topology of $\mathcal{L}( \mathcal{H} )$. Moreover there exists $m > 0$ such that
\begin{equation}\label{eq:sup>m}
\sup_{Ê\mu \ge m , \, \varepsilon > 0 } \big \| C \big ( H - ( \mu - i \varepsilon ) \big )^{-1} C^* \big \| < \infty.
\end{equation}
\end{assumption}

As explained in the introduction, the main application we have in mind concerns Schr{\"o}dinger operators of the form \eqref{eq:intro1}, where $H_0 = - \Delta$ on $L^2( \mathbb{R}^3 )$, and $V : \mathbb{R}^3 \to \mathbb{R}$, $C : \mathbb{R}^3 \to \mathbb{C}$ are potentials. In this case, conditions on $V$ and $C$ that imply Hypotheses \ref{V-1}--\ref{V3} are known; (see Section \ref{sec:schr}). In particular for bounded, compactly supported potentials $V$ and $C$, spectral singularities, in the sense of Definition \ref{def:spec-sing}, correspond to real resonances associated to resonant states with incoming Sommerfeld radiation condition.

\subsection{Main results}\label{subsec:main_res}
Our first result shows that the subspace of ``dissipative states'' $\mathcal{H}_{ \mathrm{d} }( H ) = \{ u \in \mathcal{H} , \| e^{ - i t H }Êu \|Ê\to 0 , t \to \infty \}$ coincides with the subspace $\mathcal{H}_{ \mathrm{p} }( H )$ spanned by the generalized eigenstates of $H$ corresponding to non-real eigenvalues.
\begin{theorem}\label{thm:Hp-Hd}
Suppose that Hypotheses \ref{V-1}--\ref{V3} hold. Then
\begin{equation*}
\mathcal{H}_{ \mathrm{d} }( H ) = \mathcal{H}_{ \mathrm{p} }( H ).
\end{equation*}
\end{theorem}
The problem of finding conditions implying that $\mathcal{H}_{ \mathrm{d} }( H ) = \mathcal{H}_{ \mathrm{p} }( H )$ is quoted as open in \cite{Davies1}. For small perturbations, such a result follows from similarity of $H$ and $H_0$ (see \cite{Kato1}) implying that $\mathcal{H}_{ \mathrm{d} }( H ) = \mathcal{H}_{ \mathrm{p} }( H ) = \{ 0 \}$; but, to our knowledge, Theorem \ref{thm:Hp-Hd} is new, given our assumptions.

For the nuclear optical model, Theorem \ref{thm:Hp-Hd} implies that, unless the initial state is a linear combination of generalized eigenstates corresponding to non-real eigenvalues of $H$, the probability that the neutron eventually escapes from the nucleus is always strictly positive.

Assuming that $H$ has no spectral singularities, we have the following result.
\begin{theorem}\label{thm:nonblowup}
Suppose that Hypotheses \ref{V-1}--\ref{V3} hold and that $H$ has no spectral singularities in $[ 0 , \infty )$. Then there exist $m_1>0$ and $m_2>0$ such that, for all $u \in \mathcal{H}_{ \mathrm{p} }( H^* )^\perp$,
\begin{equation*}
m_1 \| u \| \le \big \| e^{ - i t H }Êu \big \|Ê\le m_2 \|Êu \|Ê, \quad t \in \mathbb{R}.
\end{equation*}
\end{theorem}
The second inequality of the last equation shows that the solution of the Schr{\"o}dinger equation
$$
\left \{
\begin{array}{l}
 i \partial_t u_t = H u_t \\
 u_0 \in \mathcal{H}_{ \mathrm{p} }( H^* )^\perp ,
\end{array}
\right.
$$
cannot blow up, as $t\to - \infty$, and that the norm of $u_t$ is controlled by the norm of the initial state $u_0$. For Schr{\"o}dinger operators with a complex potential, a related result has been established in \cite{Go10_01}.

Our next results concern the scattering theory for $H$ and $H_0$. The wave operator $W_-( H , H_0 )$ is defined by
\begin{equation*}
W_-( H , H_0 ) := \underset{t\to \infty }{\slim}  \, e^{ - i t H } e^{ i t H_0 } .
\end{equation*}
It will be recalled in Section \ref{subsec:waveop} that $W_-( H , H_0 )$ exists under our assumptions. One of our main concerns is to study the vector space $\mathrm{Ran}( W_-( H , H_0 ) )$. This is a central issue of dissipative scattering theory and it is also a crucial input in the scattering theory of Lindblad master equations, as mentioned in the introduction.

Roughly speaking, the following two theorems will show that $W_-( H , H_0 )$ is asymptotically complete if and only if $H$ has no spectral singularities in $[ 0 , \infty )$.
\begin{theorem}\label{thm:AC}
Suppose that Hypotheses \ref{V-1}--\ref{V3} hold and that $H$ has no spectral singularities in $[ 0 , \infty )$. Then $W_-( H , H_0 )$ is asymptotically complete, in the sense that
\begin{equation*}
\mathrm{Ran} ( W_-( H , H_0 ) ) = \big ( \mathcal{H}_{ \mathrm{b} }( H ) \oplus \mathcal{H}_{ \mathrm{p} }( H^* ) \big )^\perp.
\end{equation*}
In particular, the restriction of $H$ to $( \mathcal{H}_{ \mathrm{b} }( H ) \oplus \mathcal{H}_{ \mathrm{p} }( H^* ) )^\perp$ is similar to $H_0$.
\end{theorem}
The notion of asymptotic completeness of the wave operators in our context will be discussed in Section \ref{subsec:compl}. It is proven in \cite{FaFaFrSc17_01} that, if Hypotheses \ref{V-1}, \ref{V0} and \ref{V1} hold with $\mathrm{c}_V < 2$, (see Eq. \eqref{eq:ZV}), then 
$
\mathrm{Ran}( W_-( H , H_0 ) ) = ( \mathcal{H}_{ \mathrm{b} }( H ) \oplus \mathcal{H}_{ \mathrm{d} }( H^* ) )^\perp .
$
Theorem \ref{thm:AC} improves this result in that it does not require any smallness condition on the constant $\mathrm{c}_V$, and shows that $\mathrm{Ran} ( W_-( H , H_0 ) )$ in fact coincides with the orthogonal complement of the vector space spanned by the generalized eigenvectors of $H^*$.

Theorem \ref{thm:AC} has further important consequences that we list in the following corollary. The wave operator $W_+( H_0 , H )$ is defined by
\begin{equation*}
W_+( H_0 , H ) := \underset{t\to \infty }{\slim}  \, e^{ i t H_0 } e^{ - i t H } \Pi_{ \mathrm{b} }( H )^\perp,
\end{equation*}
where $\Pi_{ \mathrm{b} }( H )$ denotes the orthogonal projection onto $\mathcal{H}_{ \mathrm{b} }( H )$ and $\Pi_{ \mathrm{b} }( H )^\perp := \mathrm{Id} - \Pi_{ \mathrm{b} }( H )$. The scattering operator $S( H , H_0 )$ is defined by
\begin{equation*}
S( H , H_0 ) := W_+( H_0 , H ) W_-( H , H_0 ).
\end{equation*}
It will be shown in Section \ref{sec:prelim} that $W_+ ( H_0 , H )$ and $S( H , H_0 )$ are well-defined under our assumptions.
\begin{corollary}\label{cor:W+}
Suppose that Hypotheses \ref{V-1}--\ref{V3} hold and that $H$ has no spectral singularities in $[ 0 , \infty )$. Then $W_+( H_0 , H ) : \mathcal{H} \to \mathcal{H}$ is surjective and
\begin{equation}
\mathrm{Ker} ( W_+( H_0 , H ) ) = \mathcal{H}_{ \mathrm{b} }( H ) \oplus \mathcal{H}_{ \mathrm{p} }( H ). \label{eq:Ker_W+}
\end{equation}
Moreover, $S( H, H_0 ): \mathcal{H} \to \mathcal{H}$ is bijective.
\end{corollary}
The equivalence between the closedness of $\mathrm{Ran}( W_-( H , H_0 ) )$ and the invertibility of $S( H , H_0 )$ was already observed in \cite{Davies1}. The surjectivity of $W_+( H_0 , H )$ is an obvious consequence of the invertibility of $S( H , H_0 )$. For the nuclear optical model, \eqref{eq:Ker_W+} implies that, for any initial state $u_0$ orthogonal to all the generalized eigenstates of $H$, there exists a scattering state $u_+ \neq 0$ satisfying \eqref{eq:intro_ex_scatt}.

Finally we will show that the assumption that $H$ has no spectral singularities is essentially necessary for asymptotic completeness, in the sense that if there exists a non-regular spectral point of $H$ with sufficiently singular behavior, then $W_-( H , H_0 )$ is not asymptotically complete.
\begin{theorem}\label{thm:nonAC}
Suppose that Hypotheses \ref{V-1}--\ref{V3} hold. Assume that there exist an interval $J \subset [ 0 , \infty )$ and a vector $u \in \mathcal{H}$ such that
\begin{align}
\lim_{ \varepsilon \downarrow 0 } \int_{ÊJÊ}Ê\big \|ÊC ( H - ( \lambda - i \varepsilon ) )^{-1}ÊC^* u \big \|^2 d \lambda = \infty. \label{eq:spectrsingenough}
\end{align}
Then $ W_-( H , H_0 ) $ is \emph{not} asymptotically complete: 
\begin{equation*}
\mathrm{Ran} ( W_-( H , H_0 ) ) \subsetneq \big ( \mathcal{H}_{ \mathrm{b} }( H ) \oplus \mathcal{H}_{ \mathrm{p} }( H^* ) )^\perp.
\end{equation*}
\end{theorem}
As already mentioned before, for Schr{\"o}dinger operators of the form \eqref{eq:intro1}, spectral singularities correspond to real resonances. In other words, in that case, a spectral singularity is a pole in $\mathbb{R}$ of the meromorphic extension of $\lambda \mapsto ( H - \lambda^2 )^{-1}$ from the upper half-plane to $\mathbb{C}$, where $( H - \lambda^2 )^{-1}$ is understood as a map from $L^2_{ \mathrm{c}Ê}( \mathbb{R}^3 ) := \{ u \in L^2( \mathbb{R}^3 ) , u \text{ is compactly supported} \}$ to $L^2_{Ê\mathrm{loc} }( \mathbb{R}^3 ) := \{ u : \mathbb{R}^3 \to \mathbb{C}, u \in L^2( K ) \text{ for all compact set } K \subset \mathbb{R}^3 \}$. In particular condition \eqref{eq:spectrsingenough} is always satisfied for any such singularity; (see Section \ref{sec:schr}). Theorem \ref{thm:AC} and Theorem \ref{thm:nonAC} then imply that $W_- ( H , H_0 )$ is asymptotically complete if and only if $H$ does not have real resonances. We will come back to applications of our results to Schr{\"o}dinger operators in Section \ref{sec:schr}.

\subsection{Ideas of the proof}\label{subsec:ideas}
The first step of the proof of Theorem \ref{thm:Hp-Hd} consists in showing that
\begin{equation}\label{eq:1ststep}
\mathrm{Ran}( W_-( H , H_0 ) ) = S( H ) \cap \mathcal{H}_{ \mathrm{b} }( H )^\perp,
\end{equation}
where
\begin{equation*}
S( H ) := \Big \{ u \in \mathcal{H} , \, \sup_{ t \ge 0 } \big \|Êe^{ i t H }Êu \big \|Ê< \infty \Big \} .
\end{equation*}
It is easy to verify that $\mathrm{Ran}( W_-( H , H_0 ) ) \subset S( H ) \cap \mathcal{H}_{ \mathrm{b} }( H )^\perp$. The converse inclusion will be established thanks to the property proven in \cite{Davies1} that $\mathcal{H}_{ \mathrm{b} }( H )^\perp = \mathcal{H}_{ \mathrm{ac}Ê}( H )$, where $\mathcal{H}_{ \mathrm{ac}Ê}( H )$ is the absolutely continuous spectral subspace of $H$ defined in a suitable way. The definition of $\mathcal{H}_{ \mathrm{ac}Ê}( H )$ will be recalled in Section \ref{subsec:absolute}.

The second step will use in an essential way the notion of a spectral projection for non-self-adjoint operators \cite{Du58_01,Sc60_01}, defined by
\begin{equation}
E_H( I ) := \underset{ \varepsilon \downarrow 0 }{\wlim} \frac{1}{ 2 i \pi } \int_I \big ( ( H - ( \lambda + i \varepsilon ) )^{-1} - ( H - ( \lambda - i \varepsilon ) )^{-1} \big ) d \lambda , \label{eq:defprojEH(I)}
\end{equation}
where $I \subset [ 0 , \infty )$ is a closed interval. We mention that such spectral projections were already used in a stationary approach to scattering theory, for differential operators in \cite{Mo67_01,Mo68_01}, and in an abstract setting in \cite{Go70_01,Go71_01,Hu71_01}. We will recall that $E_H ( I )$ is a well-defined projection if $H$ does not have spectral singularities in $I$. Its adjoint is given by $E_{H^*}( I )$. If $H$ does not have spectral singularities in $I$, we will show that 
\begin{equation*}
\mathrm{Ran} ( E_{H^*} ( I ) ) \subset \mathrm{Ran} ( W_+( H^* , H_0 ) ) ,
\end{equation*}
where $W_+( H^* , H_0 ) = \slim e^{ i t H^* } e^{ - i t H_0 }$, $t \to \infty$. Taking the orthogonal complements, we will deduce that
\begin{align}
\mathcal{H}_{Ê\mathrm{b} }( H ) \oplus \mathcal{H}_{Ê\mathrm{d} }( H ) = \mathrm{Ran} ( W_+( H^* , H_0 ) )^\perp \subset \bigcap_{I \subset [ 0 , \infty ) } \mathrm{Ker}( E_H( I ) ) ,  \label{eq:llaa1}
\end{align}
where the intersection runs over all closed intervals $I \subset [ 0 , \infty )$ such that $I$ does not contain any spectral singularities of $H$. The equality in this statement is not difficult to verify.

In the last step of our proof of Theorem \ref{thm:Hp-Hd}, we will establish that
\begin{equation}
 \bigcap_{I \subset [ 0 , \infty ) } \mathrm{Ker}( E_H( I ) ) \subset \mathcal{H}_{ \mathrm{b} }( H ) \oplus \mathcal{H}_{ \mathrm{p} }( H ) . \label{eq:llaa2}
\end{equation}
The inclusions in \eqref{eq:llaa1} and \eqref{eq:llaa2} will then show that $\mathcal{H}_{ \mathrm{d} }( H ) = \mathcal{H}_{ \mathrm{p} }( H )$. Technically the verification of \eqref{eq:llaa2} is the most involved part of the proof. The main ingredient will be a spectral decomposition formula suitably modified to take into account the spectral singularities $\{ \lambda_j \}_{j=1}^n$ of $H$. It will be of the form
\begin{align}
&\prod_{ j = 1 }^n ( ( H - i )^{-1} - \mu_j )^{\nu_j}Ê= \prod_{ j = 1 }^n ( ( H - i )^{-1} - \mu_j )^{\nu_j} \Pi_{ \mathrm{pp} } \notag  \\
& + \underset{\varepsilon \downarrow 0}{\wlim} \frac{1}{2i\pi} \int_0^\infty  \prod_{ j = 1 }^n \big ( ( \lambda - i )^{-1} - \mu_j \big )^{\nu_j} \big ( \big ( H - ( \lambda + i \varepsilon )Ê\big )^{-1} - \big ( H - ( \lambda - i \varepsilon )Ê\big )^{-1} \big )  d \lambda , \label{eq:spectral_decomp}
\end{align}
where $\mu_j = ( \lambda_j - i )^{-1}$ and $\Pi_{ \mathrm{pp} }$ is the sum of all Riesz projections of $H$ corresponding to isolated eigenvalues. Equation \eqref{eq:spectral_decomp} may be seen as a generalization of the well-known spectral decomposition formula for self-adjoint operators to dissipative operators with finitely many spectral singularities. In particular, if $H$ does not have spectral singularities, \eqref{eq:spectral_decomp} reduces to
\begin{equation}
\mathrm{Id}Ê= \Pi_{ \mathrm{pp} } + \underset{\varepsilon \downarrow 0}{\wlim} \frac{1}{2i\pi} \int_0^\infty \big ( \big ( H - ( \lambda + i \varepsilon )Ê\big )^{-1} - \big ( H - ( \lambda - i \varepsilon )Ê\big )^{-1} \big ) d \lambda.\label{eq:spectral_decomp_2}
\end{equation}

Once Theorem \ref{thm:Hp-Hd} is established, we will proceed to prove Theorems \ref{thm:AC} and \ref{thm:nonAC} as follows: Using Parseval's theorem, we will justify that, for all $\varepsilon > 0$ and for all $u \in ( \mathcal{H}_{ \mathrm{b} }( H ) \oplus \mathcal{H}_{ \mathrm{p} }( H^* ) )^\perp$,
\begin{align}
\int_0^\infty e^{ - s \varepsilon } \big \|ÊC e^{ i s H }Êu \big \|^2 ds &= \frac{1}{ 2 \pi } \int_{Ê\mathbb{R}Ê}Ê\big \|ÊC ( H - ( \lambda - i \varepsilon ) )^{-1}Êu \big \|^2 d \lambda . \label{eq:parseval_intro}
\end{align}
If $H$ does not have spectral singularities, using the resolvent equation and Hypothesis \ref{V1}, we will show that the right side of this equation remains bounded, as $\varepsilon \to 0$. This will prove that any $u \in ( \mathcal{H}_{ \mathrm{b} }( H ) \oplus \mathcal{H}_{ \mathrm{p} }( H^* ) )^\perp$ belongs to $S( H )$ and hence to $\mathrm{Ran}( W_-( H , H_0 ) )$, by \eqref{eq:1ststep}. Thus Theorem \ref{thm:AC} will follow. If $H$ does have a spectral singularity, and if \eqref{eq:spectrsingenough} holds, we will construct a vector $u \in ( \mathcal{H}_{ \mathrm{b} }( H ) \oplus \mathcal{H}_{ \mathrm{p} }( H^* ) )^\perp$ such that the limit of \eqref{eq:parseval_intro}, as $\varepsilon \to 0$, is infinite. Again, by \eqref{eq:1ststep}, this will prove that $u \notin \mathrm{Ran}( W_-( H , H_0 ) )$, and thus establish the statement of Theorem \ref{thm:nonAC}.

Finally, Theorem \ref{thm:nonblowup} will be a consequence of Theorem \ref{thm:AC}.

\subsection{Organization of the paper}

The remainder of the paper is organized as follows. Section \ref{sec:prelim} describes preliminaries, recalling several known results that are used in the sequel. In Section \ref{sec:abs-range}, we prove \eqref{eq:1ststep}. Our main results are established in Section \ref{sec:main}. In Section \ref{sec:schr}, we verify our abstract hypotheses for dissipative Schr{\"o}dinger operators and, in Section \ref{sec:Lindblad}, we sketch consequences of our results for the scattering theory of Lindblad master equations. For the convenience of the reader, the proofs of several auxiliary results are presented in appendices.

\section{Preliminaries}\label{sec:prelim}

In this section, we gather some basic facts concerning the spectral and scattering theories for $H$. Several results described here can be found in the literature, sometimes under slightly different assumptions (see \cite{Davies4,Davies1,Ex85_01,FaFaFrSc17_01,Martin,Simon2}). For completeness, proofs of the main results of this section are recalled or elaborated upon in Appendices \ref{app:spectrum} and \ref{app:existencewave}.

\subsection{The spectrum of $H$}\label{subsec:spectrum}

We begin by recalling a few spectral properties of the operator $H$. The first one concerns the subspaces $\mathcal{H}_{ \mathrm{b} }( H )$ and $\mathcal{H}_{ \mathrm{b} }( H^* )$ defined in \eqref{eq:defHb} and \eqref{eq:defH*b}.
\begin{lemma}\label{lm:Hb}
Suppose that $H = H_0 + V - i C^* C$, with $H_0$ self-adjoint, $V$ symmetric and relatively compact w.r.t. $H_0$, and $C \in \mathcal{L}( \mathcal{H} )$. Then
\begin{equation*}
\mathcal{H}_{ \mathrm{b} }( H ) = \mathcal{H}_{ \mathrm{b} }( H^* ) = \mathrm{Span} \big \{ u \in \mathcal{D}( H ) , \, \exists \lambda \in \mathbb{R} , \, H u = \lambda u \big \} \subset \mathcal{H}_{ \mathrm{pp} }( H_V ) \cap \mathrm{Ker}( C ).
\end{equation*}
\end{lemma}
Lemma \ref{lm:Hb} is a consequence of \cite[Theorem 9.1]{Simon2} or \cite[Lemma 1]{Davies1}. We give a short proof in Appendix \ref{app:spectrum}. Lemma \ref{lm:Hb} implies that, if $\mathcal{H}_{ \mathrm{pp} }( H_V )$ is finite dimensional then so is $\mathcal{H}_{ \mathrm{b} }( H )$.

For a closed operator, there are in general several definitions of essential spectrum; (see e.g. \cite{EdEv87_01}). Fortunately, these different definitions coincide in our situation. It is convenient to define the essential spectrum of $H$ as
\begin{equation*}
\sigma_{ \mathrm{ess} }( H ) := \mathbb{C} \setminus \rho_{ \mathrm{ess} }( H ) ,
\end{equation*}
where
\begin{align*}
\rho_{ \mathrm{ess} }( H ) := \big \{ z \in \mathbb{C} , \, &\mathrm{Ran}( H - z \mathrm{Id} ) \text{ is closed}, \\
& \mathrm{dim} \, \mathrm{Ker}( H - z \mathrm{Id} ) < \infty \ \text{ or } \ \mathrm{codim} \, \mathrm{Ran} ( H - z \mathrm{Id} ) < \infty \big \}.
\end{align*}
We then have that (see \cite[Section IV.5.6]{Ka66_01} and \cite[proof of Proposition B.2]{Frank}):
\begin{lemma}\label{lm:essspec}
Suppose that $H = H_0 + V - i C^* C$ with $H_0$ self-adjoint, $V$ symmetric and relatively compact with respect to $H_0$, and $C \in \mathcal{L}( \mathcal{H} )$ relatively compact with respect to $H_0$. Then
\begin{equation*}
\sigma_{\mathrm{ess}}( H ) = \sigma_{ \mathrm{ess} }( H_V ) = \sigma_{ \mathrm{ess} }( H_0 )  .
\end{equation*}
Moreover
\begin{align*}
\sigma( H ) \setminus \sigma_{ \mathrm{ess} }( H ) = \big \{ z \in \mathbb{C} , \, &\mathrm{Ran}( H - z \mathrm{Id} ) \text{ is closed}, \\
& 0 < \mathrm{dim} \, \mathrm{Ker}( H - z \mathrm{Id} )  = \mathrm{codim} \, \mathrm{Ran} ( H - z \mathrm{Id} ) < \infty \big \} ,
\end{align*}
and $\sigma( H ) \setminus \sigma_{ \mathrm{ess} }( H )$ is at most countable and consists of eigenvalues of finite algebraic multiplicities.
\end{lemma}
Recall that, for $\lambda \in \sigma( H ) \setminus \sigma_{ \mathrm{ess} }( H )$, $\Pi_\lambda$ denotes the usual Riesz projection for $H$ associated to $\lambda$ (see \eqref{eq:defPilambda}). The range of $ \Pi_\lambda $ is spanned by linear combinations of generalized eigenstates corresponding to $\lambda$. Let
\begin{equation*}
\Pi_{ \mathrm{pp} } := \sum_{ \lambda } \Pi_\lambda ,
\end{equation*}
where the sum runs over all $\lambda \in \sigma( H ) \setminus \sigma_{ \mathrm{ess} }( H )$. Then
\begin{equation}
\mathrm{Ran} ( \Pi_{ \mathrm{pp} } ) = \mathcal{H}_{ \mathrm{b} }( H ) \oplus \mathcal{H}_{ \mathrm{p} }( H ), \quad \mathrm{Ker} ( \Pi_{ \mathrm{pp} } ) = ( \mathcal{H}_{ \mathrm{b} }( H ) \oplus \mathcal{H}_{ \mathrm{p} }( H^* ) )^\perp , \label{eq:rankerpipp}
\end{equation}
where $\mathcal{H}_{ \mathrm{p} }( H )$ and $\mathcal{H}_{ \mathrm{p} }( H^* )$ are the vector spaces spanned by the generalized eigenvectors of $H$ and $H^*$ corresponding to non-real eigenvalues (see \eqref{eq:defHp} and \eqref{eq:defH^*p}). Moreover, it is easy to verify that
\begin{equation*}
\mathcal{H}_{ \mathrm{p} }( H ) \subset \mathcal{H}_{ \mathrm{d} }( H ) \subset \mathcal{H}_{ \mathrm{b} }( H )^\perp, \quad \mathcal{H}_{ \mathrm{p} }( H^* ) \subset \mathcal{H}_{ \mathrm{d} }( H^* ) \subset \mathcal{H}_{ \mathrm{b} }( H^* )^\perp ,
\end{equation*}
where $\mathcal{H}_{ \mathrm{d} }( H )$ and $\mathcal{H}_{ \mathrm{d} }( H^* )$ are defined in \eqref{eq:defHd} and \eqref{eq:defH*d}.

To conclude this section, we remark that the only possible generalized eigenvectors corresponding to a real eigenvalue are eigenvectors in the usual sense, as expressed in the following easy lemma; (see Appendix \ref{app:spectrum} for the proof).
\begin{lemma}\label{lm:real-generalized}
Suppose that $H = H_0 + V - i C^* C$, with $H_0$ self-adjoint, $V$ symmetric and relatively compact with respect to $H_0$, and $C \in \mathcal{L}( \mathcal{H} )$. Let $\lambda \in \mathbb{R}$ be an eigenvalue of $H$ and suppose that $u \in \mathcal{H}$ satisfies $u \in \mathcal{D}( H^k )$, for some $k \in \mathbb{N}$,  $( H - \lambda )^k u = 0$ and $( H - \lambda )^{ k - 1 } u \neq 0$. Then $k=1$, i.e., $u \neq 0$ and $H u = \lambda u$.
\end{lemma}

\subsection{The wave operators $W_-( H , H_0 )$ and $W_+( H^* , H_0 )$}\label{subsec:waveop}
We recall conditions implying the existence of the wave operators
\begin{equation*}
W_-( H , H_0 ) := \underset{t\to \infty }{\slim}  \, e^{ - i t H } e^{ i t H_0 } , \quad W_+( H^* , H_0 ) := \underset{t\to \infty }{\slim}  \, e^{ i t H^* } e^{ - i t H_0 } .
\end{equation*}
We remark that
\begin{align*}
2 \int_0^t \big \langle u , e^{ i s H^* } C^*ÊC e^{ - i s H } u \rangle ds = - \int_0^t \partial_s \big \langle u , e^{ i s H^* } e^{ - i s H } u \rangle ds = \| u \|^2 - \big \|Êe^{ - i t H }Êu \big \|^2 ,
\end{align*}
and hence
\begin{align}\label{eq:alpha}
\int_0^\infty \big \|ÊC e^{ - i t H } u \big \|^2 dt \le \frac12 \| u \|^2 ,
\end{align}
for all $u \in \mathcal{H}$. The same argument shows that
\begin{align}\label{eq:alpha2}
\int_0^\infty \big \|ÊC e^{ i t H^* } u \big \|^2 dt \le \frac12 \| u \|^2 ,
\end{align}
for all $u \in \mathcal{H}$. Equations \eqref{eq:alpha} and \eqref{eq:alpha2}, combined with Hypotheses \ref{V-1}, \ref{V0} and \ref{V1}, imply the existence of the wave operators $W_-( H , H_0 )$ and $W_+( H^* , H_0 )$. More precisely, recalling that the domains of $H$, $H^*$ and $H_0$ coincide, we have the following result.
\begin{proposition}\label{prop:existence_W-}
Suppose that Hypotheses \ref{V-1}, \ref{V0} and \ref{V1} hold. Then the wave operators $W_-( H , H_0 )$ and $W_+( H^* , H_0 )$ exist and are injective contractions. Moreover, for all $t \in \mathbb{R}$, 
  \begin{align}
 & e^{ - i t H } W_-( H , H_0 ) = W_-( H , H_0 ) e^{ - i t H_0 }, \quad  e^{ - i t H^* } W_+( H^* , H_0 ) = W_+( H^* , H_0 ) e^{ - i t H_0 } . \label{eq:inter-1}
  \end{align}
In particular, $W_-( H , H_0 ) \mathcal{D}( H_0 ) \subset \mathcal{D}( H_0 )$, $W_+( H^* , H_0 ) \mathcal{D}( H_0 ) \subset \mathcal{D}( H_0 )$, and
  \begin{align}
& H W_-( H , H_0) u = W_-( H , H_0 ) H_0 u , \quad H^* W_+( H^* , H_0 ) u = W_+( H^* , H_0 ) H_0 u , \label{eq:inter-2}
  \end{align}
  for all $u \in \mathcal{D}( H_0 )$.
\end{proposition}
The existence of $W_-( H , H_0 )$ and $W_+( H^* , H_0 )$ follows from a standard Cook argument. Contractivity is a consequence of contractivity of $\{ e^{ - i t H }Ê\}_{ t \ge 0 }$, $\{ e^{ i t H^* }Ê\}_{ t \ge 0 }$ and unitarity of $\{ e^{ - i t H_0 }Ê\}_{ t \in \mathbb{R} }$. To render our analysis self-contained, we recall a proof of Proposition \ref{prop:existence_W-} in Appendix \ref{app:existencewave}.

As mentioned in the previous section, one of our main concerns consists in studying the vectors spaces $\mathrm{Ran}( W_-( H , H_0 ) )$ and $\mathrm{Ran}( W_+( H^* , H_0 ) )$. As stated in the following proposition, the closures of the ranges of $W_- (H , H_0)$ and $W_+( H^* , H_0 )$ are known without requiring any assumption beyond Hypotheses \ref{V-1}, \ref{V0} and \ref{V1}.
\begin{proposition}\label{prop:randense_W-}
Suppose that Hypotheses \ref{V-1}, \ref{V0} and \ref{V1} hold. Then 
\begin{equation*}
\overline{Ê\mathrm{Ran}( W_-( H , H_0 ) ) } = \big ( \mathcal{H}_{ \mathrm{b} }( H ) \oplus \mathcal{H}_{ \mathrm{d} }( H^* ) \big )^\perp, \quad \overline{\mathrm{Ran}( W_+( H^* , H_0 ) )} = \big ( \mathcal{H}_{ \mathrm{b} }( H ) \oplus \mathcal{H}_{ \mathrm{d} }( H ) \big )^\perp.
\end{equation*}
\end{proposition}
See Appendix \ref{app:existencewave} for the proof of Proposition \ref{prop:randense_W-}.

\subsection{The absolutely continuous subspace}\label{subsec:absolute}

Following \cite{Davies4,Davies1}, we define the ``absolutely continuous subspace'' of $H$, $\mathcal{H}_{ \mathrm{ac} }( H )$, as follows: Let
\begin{align}
& M( H ) := \Big \{ u \in \mathcal{H} , \exists \mathrm{c}_u > 0 ,  \forall v \in \mathcal{H} , \int_0^\infty \big |Ê\langle e^{ - i t H }Êu , v \rangle \big |^2 dt \le \mathrm{c}_u \| v \|^2 \Big \} . \label{eq:defM(H)}
\end{align}
Then $\mathcal{H}_{ \mathrm{ac} }( H ) := \overline{ M ( H ) }$ is the closure of $M( H )$ in $\mathcal{H}$. When $H$ is self-adjoint, this definition coincides with the usual one  based on the nature of the spectral measures of $H$. Likewise, we set $\mathcal{H}_{ \mathrm{ac} }( H^* ) := \overline{ M ( H^* ) }$, where
\begin{equation}
M( H^* ) := \Big \{ u \in \mathcal{H} , \exists \mathrm{c}^*_u > 0 ,  \forall v \in \mathcal{H} , \int_0^\infty \big |Ê\langle e^{ i t H^* }Êu , v \rangle \big |^2 dt \le \mathrm{c}^*_u \| v \|^2 \Big \}. \label{eq:defM(H*)} 
\end{equation}

Another characterization of the absolutely continuous subspace of $H$ follows from the theory of unitary dilations of one-parameter contraction semigroups; (see e.g. \cite{NF} for the theory of unitary dilations, and \cite{Ne85_01} for further relations with dissipative scattering theory). There is a unique orthogonal decomposition
\begin{equation*}
\mathcal{H} = \mathcal{H}_{ \mathrm{u} } \oplus \mathcal{H}_{ \mathrm{c.n.u.} } ,
\end{equation*}
where $\mathcal{H}_{ \mathrm{u} }$ and $\mathcal{H}_{ \mathrm{c.n.u.} }$ are Hilbert spaces invariant under the action of $\{ e^{ - i t H } \}_{ t \ge 0 }$, such that $\{ e^{ - i t H } \}_{ t \ge 0 }$ is unitary on $\mathcal{H}_{ \mathrm{u} }$, and completely non-unitary on $\mathcal{H}_{ \mathrm{c.n.u.} }$ (meaning that $\{ e^{ - i t H } \}_{ t \ge 0 }$ is not unitary on any of the nontrivial subspaces of $\mathcal{H}_{ \mathrm{c.n.u.} }$). It is shown in \cite{Davies1} (see also \cite{Ex85_01}) that, for all $t \ge 0$ and $u \in \mathcal{H}_{ \mathrm{u} }$, we have that $e^{ - i t H } u = e^{ - i t H_V } u$, which, assuming Hypotheses \ref{V-1}, \ref{V0} and \ref{V1}, implies that (see \cite[Proof of Theorem 5]{Davies1})
\begin{equation*}
\mathcal{H}_{ \mathrm{ac} }(H) = \mathcal{H}_{ \mathrm{ac} }( H_V ) \oplus \mathcal{H}_{ \mathrm{c.n.u.} } .
\end{equation*}
In turn, this equality yields (see \cite[Theorem 5]{Davies1})
\begin{equation}
\mathcal{H}_{ \mathrm{ac} }(H) = \mathcal{H}_{ \mathrm{b} }( H )^\perp , \label{eq:HacHbperp}
\end{equation}
where $\mathcal{H}_{ \mathrm{b} }( H )$ is defined in \eqref{eq:defHb}. Likewise, $\mathcal{H}_{ \mathrm{ac} }(H^*) = \mathcal{H}_{ \mathrm{b} }( H^* )^\perp$, and therefore 
\begin{equation}
\mathcal{H}_{ \mathrm{ac} }(H) = \mathcal{H}_{ \mathrm{ac} }(H^*) , \label{eq:equal_acspec}
\end{equation}
since $\mathcal{H}_{ \mathrm{b} }( H ) = \mathcal{H}_{ \mathrm{b} }( H^* )$, by Lemma \ref{lm:Hb}.

Assuming that $W_-( H , H_0 )$ exists, and using the intertwining property \eqref{eq:inter-1}, it is not difficult to verify that 
\begin{equation}
\mathrm{Ran}( W_-( H , H_0 ) ) \subset M( H ) \subset \mathcal{H}_{ \mathrm{ac} }( H ), \quad \mathrm{Ran}( W_+( H^* , H_0 ) ) \subset M( H^* ) \subset \mathcal{H}_{ \mathrm{ac} }( H ) , \label{eq:1stincl}
\end{equation}
see the proof of Theorem \ref{prop:range1}, below, for details. We  further discuss these inclusions in Section \ref{sec:abs-range}.

\subsection{The wave operators $W_+( H_0 , H )$, $W_-( H_0 , H^* )$ and the scattering operators}
We consider now the wave operators
\begin{equation*}
W_+( H_0 , H ) := \underset{t\to \infty }{\slim}  \, e^{ i t H } e^{ - i t H } \Pi_{ \mathrm{ac} }( H ) , \quad W_-( H_0 , H^* ) := \underset{t\to \infty }{\slim}  \, e^{ - i t H_0 } e^{ i t H^* } \Pi_{ \mathrm{ac}Ê}( H^* ) ,
\end{equation*}
where $\Pi_{ \mathrm{ac} }( H )$, $\Pi_{ \mathrm{ac} }( H^* )$ are the orthogonal projections onto $\mathcal{H}_{ \mathrm{ac} }( H )$ and $\mathcal{H}_{ \mathrm{ac} }( H^* )$, respectively. We have the following result, whose proof uses the Kato smoothness estimate \eqref{eq:ZV}, an auxiliary technical lemma that will be recalled in Section \ref{sec:abs-range}, Lemma \ref{lm:compwn}, and arguments similar to those in the proof of Proposition \ref{prop:existence_W-}. We refer the reader to Appendix \ref{app:existencewave} for details.
\begin{proposition}\label{prop:existence_W+}
Suppose that Hypotheses \ref{V-1}, \ref{V0} and \ref{V1} hold. Then the wave operators $W_+( H_0 , H )$ and $W_-( H_0 , H^* )$ exist and are contractions. These operators have dense ranges, and their kernels are given by
\begin{equation}\label{eq:ker_W+}
\mathrm{Ker}( W_+( H_0 , H ) ) = \mathcal{H}_{ \mathrm{b} }( H ) \oplus \mathcal{H}_{ \mathrm{d} }( H ) , \quad \mathrm{Ker}( W_-( H_0 , H^* ) ) = \mathcal{H}_{ \mathrm{b} }( H ) \oplus \mathcal{H}_{ \mathrm{d} }( H^* ).
\end{equation} 
Moreover, for all $t \in \mathbb{R}$, 
  \begin{align}
 & e^{ - i t H_0 } W_+( H_0 , H ) = W_+( H_0 , H ) e^{ - i t H }, \quad  e^{ - i t H_0 } W_-( H_0 , H^* ) = W_-( H_0 , H^* ) e^{ - i t H^* } . \label{eq:inter-3}
  \end{align}
In particular, $W_+( H_0 , H ) \mathcal{D}( H_0 ) \subset \mathcal{D}( H_0 )$, $W_-( H_0 , H^* ) \mathcal{D}( H_0 ) \subset \mathcal{D}( H_0 )$, and
  \begin{align}
& H_0 W_+( H_0 , H) u = W_+( H_0 , H ) H u , \quad H_0 W_-( H_0 , H^* ) u = W_-( H_0 , H^* ) H^* u , \label{eq:inter-4}
  \end{align}
  for all $u \in \mathcal{D}( H_0 )$.
\end{proposition}

We mention that the existence of $W_+( H_0 , H )$ and $W_-( H_0 , H^* )$ can be proven under various different conditions, using the Kato-Birman theory of trace-class perturbations (see \cite{Davies4}), the Enss' method (see \cite{Simon2,Pe89_01,Ka03_01}), or, as in our situation, Kato's theory of smooth perturbations (see \cite{Kato1} for weak coupling results and \cite{Mo76_01,Ka02_01} for assumptions comparable to ours, in the particular case where $V = 0$).

In dissipative scattering theory, the scattering operators are defined by
\begin{equation}
S( H , H_0 ) := W_+( H_0 , H ) W_-( H , H_0 ) , \quad S( H^* , H_0 ) = W_- ( H_0 , H^* ) W_+ ( H^* , H_0 ). \label{eq:defscattop}
\end{equation}
In view of Proposition \ref{prop:existence_W-}, Proposition \ref{prop:existence_W+} and \eqref{eq:1stincl}, the scattering operators can be shown to exist.
\begin{proposition}\label{prop:existencescatt}
Suppose that Hypotheses \ref{V-1}, \ref{V0} and \ref{V1} hold. Then the scattering operators $S( H , H_0 )$ and $S( H^* , H_0 )$ exist and are contractions. Moreover,
\begin{equation*}
S( H , H_0 )^* = S( H^* , H_0 ).
\end{equation*}
\end{proposition}
The following result gives a necessary and sufficient condition for the invertibility of the scattering operators. It is a consequence of \cite[Theorem 7]{Davies1} and Proposition \ref{prop:existence_W+}. We remark that the proof of \cite[Theorem 7]{Davies1} uses the assumption that there exists a conjugation operator commuting with $H_0$, $V$ and $C^*C$. We use a slightly differently reasoning process. Details of the proof can be found in Appendix \ref{app:existencewave}.
\begin{proposition}\label{prop:invertscatt}
Suppose that Hypotheses \ref{V-1}, \ref{V0} and \ref{V1} hold. Then the following conditions are equivalent:
\begin{itemize}
\item[(i)] The scattering operators $S( H , H_0 )$ and $S( H^* , H_0 )$ are bijective on $\mathcal{H}$.
\item[(ii)] The range of the wave operators $W_-( H , H_0 )$ and $W_+( H^* , H_0 )$ are given by
\begin{align*}
&Ê\mathrm{Ran}( W_-( H , H_0 ) ) = \big ( \mathcal{H}_{ \mathrm{b} }( H ) \oplus \mathcal{H}_{ \mathrm{d} }( H^* ) \big )^\perp, \\
& \mathrm{Ran}( W_+( H^* , H_0 ) ) = \big ( \mathcal{H}_{ \mathrm{b} }( H ) \oplus \mathcal{H}_{ \mathrm{d} }( H ) \big )^\perp.
\end{align*}
\end{itemize}
\end{proposition}
Of course, since 
\begin{equation*}
\mathrm{Ran}( S ( H , H_0 ) ) \subset \mathrm{Ran}( W_+ (H_0 , H) ), \quad \mathrm{Ran}( S ( H^* , H_0 ) ) \subset \mathrm{Ran}( W_- (H_0 , H^*) ),
\end{equation*}
we see that if the equivalent conditions of the previous proposition are satisfied, then the wave operators $W_+( H_0 , H )$ and $W_- ( H_0 , H^* )$ are surjective.

Once it is known that the scattering operators exist, one can define the scattering matrices in the usual way (see e.g. \cite{Ya92_01}). Indeed, considering for instance $S( H , H_0 )$, we observe that, by the intertwining properties \eqref{eq:inter-2} and \eqref{eq:inter-4}, the operator $S( H , H_0 )$ commutes with $H_0$, and there is therefore a direct integral decomposition of the form
\begin{equation*}
\mathcal{H} = \int^\oplus_{Ê\sigma( H_0 ) }Ê\mathcal{H}( \lambda ) dÊ\lambda ,
\end{equation*}
with the property that
\begin{equation*}
S( H , H_0 ) = \int^\oplus_{ \sigma( H_0 ) } S( \lambda ) d \lambda.
\end{equation*}
In other words, $S( H , H_0 )$ acts as multiplication by the operator-valued function $ \lambda \mapsto S( \lambda )$. The equalities in the two equations above are interpreted as unitary equivalence. The operator $S( \lambda ) : \mathcal{H}( \lambda ) \to \mathcal{H}( \lambda )$ (for a.e. $\lambda \in \sigma( H_0 )$) is called the scattering matrix. We emphasize that, in contrast to the unitary case, $S( \lambda )$ may not be invertible at some  points $\lambda_0 \in \sigma ( H_0 )$.

\subsection{Asymptotic completeness of the wave operators}\label{subsec:compl}
To conclude this preliminary section, we discuss the notion of asymptotic completeness of the wave operators. 
\begin{definition}\label{def:compl}
The wave operators $W_-( H , H_0 )$ and $W_+( H^* , H_0 )$ are said to be asymptotically complete if 
\begin{equation*}
Ê\mathrm{Ran}( W_-( H , H_0 ) ) = \big ( \mathcal{H}_{ \mathrm{b} }( H ) \oplus \mathcal{H}_{ \mathrm{p} }( H^* ) \big )^\perp, \quad \mathrm{Ran}( W_+( H^* , H_0 ) ) = \big ( \mathcal{H}_{ \mathrm{b} }( H ) \oplus \mathcal{H}_{ \mathrm{p} }( H ) \big )^\perp .
\end{equation*}
\end{definition}
Slightly different definitions of asymptotic completeness appear in the literature (see \cite{Martin,Davies4}) but, in many examples, all definitions coincide.  In view of the usual definition of asymptotic completeness in unitary scattering theory (see e.g. \cite{RS-III}), Definition \ref{def:compl} is natural in our situation.

Assuming that Hypotheses \ref{V-1}, \ref{V0} and \ref{V1} hold, Proposition \ref{prop:existence_W-} shows that, if $W_-( H , H_0 )$ is asymptotically complete, then $W_-( H , H_0 ) : \mathcal{H} \to ( \mathcal{H}_{ \mathrm{b} }( H ) \oplus \mathcal{H}_{ \mathrm{p} }( H^* ) )^\perp$ is bijective. By the intertwining property \eqref{eq:inter-2}, this implies that the restriction of $H$ to $( \mathcal{H}_{ \mathrm{b} }( H ) \oplus \mathcal{H}_{ \mathrm{p} }( H^* ) )^\perp$ is similar to $H_0$.

Proposition \ref{prop:randense_W-} and the fact that $\mathcal{H}_{ \mathrm{p} }( H ) \subset \mathcal{H}_{ \mathrm{d} }( H ) \subset \mathcal{H}_{ \mathrm{b} }( H )^\perp$ imply the following proposition.
\begin{proposition}
Suppose that Hypotheses \ref{V-1}, \ref{V0} and \ref{V1} hold. Then the wave operators $W_-( H , H_0 )$ and $W_+( H^* , H_0 )$ are asymptotically complete if and only if the following two conditions are satisfied:
\begin{align*}
& \mathrm{(a)} \ \ \mathrm{Ran}( W_-( H , H_0 ) ) \ \text{ and }Ê\ \mathrm{Ran}( W_+( H^* , H ) ) \ \text{ are closed.} \hspace{4cm} \\
& \mathrm{(b)} \ \ \mathcal{H}_{ \mathrm{p} }( H ) = \mathcal{H}_{ \mathrm{d}Ê}( H ).
\end{align*}
\end{proposition}

\section{Proof of \eqref{eq:1ststep}}\label{sec:abs-range}

In this section we discuss the relation between $\mathrm{Ran}( W_-( H , H_0 ) )$ (assuming the wave operators exists) and the absolutely continuous subspace $\mathcal{H}_{ \mathrm{ac} }( H )$ (see Section \ref{subsec:absolute}). In particular we clarify the inclusion $\mathrm{Ran}( W_-( H , H_0 ) ) \subset \mathcal{H}_{ \mathrm{ac} }( H )$ by establishing that  $\mathrm{Ran}( W_-( H , H_0 ) ) = S( H ) \cap \mathcal{H}_{ \mathrm{ac} }( H )$. We recall that 
\begin{equation}
S( H ) = \Big \{ u \in \mathcal{H} , \, \sup_{ t \ge 0 } \big \|Êe^{ i t H }Êu \big \|Ê< \infty \Big \} . \label{eq:S(H)}
\end{equation}
Similarly,
\begin{equation*}
S( H^* ) := \Big \{ u \in \mathcal{H} , \, \sup_{ t \ge 0 } \big \|Êe^{ - i t H^* }Êu \big \|Ê< \infty \Big \} .
\end{equation*}
We note that the identities
\begin{align*}
\big \|Êe^{ i t H } u \big \|^2 = \| u \|^2 + 2 \int_0^t \big \|ÊC e^{ i s H }Êu \big \|^2 ds , \quad \big \|Êe^{ - i t H^* } u \big \|^2 = \| u \|^2 + 2  \int_0^t \big \|ÊC e^{ - i s H^* }Êu \big \|^2 ds ,
\end{align*}
which are valid for all $t \ge 0$ and all $u \in \mathcal{H}$, imply that $S( H )$ and $S( H^* )$ can equivalently be defined by
\begin{equation*}
S( H ) = \Big \{ u \in \mathcal{H} , \, \int_0^\infty \big \|ÊC e^{ i s H }Êu \big \|^2 dsÊ< \infty \Big \} , \quad S( H^* ) = \Big \{ u \in \mathcal{H} , \, \int_0^\infty \big \|ÊC e^{ - i s H^* }Êu \big \|^2 dsÊ< \infty \Big \}.
\end{equation*}
The spaces $M(H)$ and $M( H^* )$ have been defined in \eqref{eq:defM(H)} and \eqref{eq:defM(H*)} and have the property that $\mathcal{H}_{ \mathrm{ac} }( H ) = \overline{ M( H ) }$ and $\mathcal{H}_{ \mathrm{ac} }( H^* ) = \overline{ M( H^* ) }$. Moreover, $\mathcal{H}_{ \mathrm{ac} }( H ) = \mathcal{H}_{ \mathrm{ac}Ê}( H^* )$ (see \eqref{eq:equal_acspec}). The main result of this section is the following theorem.
\begin{theorem}\label{prop:range1}
Suppose that Hypotheses \ref{V-1}, \ref{V0} and \ref{V1} hold. Then
\begin{align*}
& \mathrm{Ran}( W_-( H , H_0 ) ) = S( H ) \cap \mathcal{H}_{ \mathrm{ac} }( H ) = S( H ) \cap M( H ) , \\
& \mathrm{Ran}( W_+( H^* , H_0 ) ) = S( H^* ) \cap \mathcal{H}_{ \mathrm{ac} }( H ) = S( H^* ) \cap M( H^* ) .
\end{align*}
\end{theorem}

The proof of Theorem \ref{prop:range1} is based on the following two lemmas. The first one is well-known for self-adjoint operators. Its proof has been extended to generators of strongly continuous one-parameter contraction semigroups in \cite[Lemma 5.1]{Davies4}.
\begin{lemma}\label{lm:compwn}
Let $-iT$ be the generator of a strongly continuous one-parameter contraction semigroup in a separable Hilbert space $\mathcal{H}$. Then, for all $u \in \mathcal{H}_{ \mathrm{ac} }( T ) = \overline{M( T )}$ (where $M(T)$ is defined as in \eqref{eq:defM(H)}),
\begin{equation}\label{eq:limitAC}
\lim_{ t \to \infty } \langle e^{ - i t T }Êu , v \rangle = \lim_{ t \to \infty } \big \|ÊK e^{ - i t T }Êu \big \|Ê= 0,
\end{equation}
for all $v \in \mathcal{H}$ and all compact operators $K$ on $\mathcal{H}$.
\end{lemma}
\begin{lemma}\label{lm:LM1}
Suppose that Hypotheses \ref{V-1}, \ref{V0} and \ref{V1} hold. Let $u \in S(H)$. Then, for all $v \in \mathcal{H}_{ \mathrm{ac} }( H )$,
\begin{equation}
\big \langle v , e^{ i t H } u \big \rangle \to 0 , \quad t \to \infty . \label{eq:a2}
\end{equation}
Likewise, for all $u \in S(H^*)$ and $v \in \mathcal{H}_{ \mathrm{ac} }( H )$,
\begin{equation}
\big \langle v , e^{ - i t H^* } u \big \rangle \to 0 , \quad t \to \infty . \label{eq:a221}
\end{equation}
In particular, if $\mathcal{H}_{ \mathrm{b} }( H ) = \{ 0 \}$, see \eqref{eq:defHb}, then \eqref{eq:a2}--\eqref{eq:a221} hold for all $v \in \mathcal{H}$, and, for all compact operators $K$ on $\mathcal{H}$, we have that
\begin{equation}
\big \| K e^{ i t H } u \big \| \to 0 , \quad \big \| K e^{ - i t H^* } u \big \| \to 0, \quad t \to \infty . \label{eq:a3}
\end{equation}
\end{lemma}
\begin{proof}
We prove \eqref{eq:a2}; the proof of \eqref{eq:a221} is similar. Let $u \in S( H )$,  $v \in \mathcal{H}_{ \mathrm{ac} }( H )$ and $\varepsilon > 0$. Since, by \eqref{eq:equal_acspec}, $\mathcal{H}_{ \mathrm{ac} }( H ) = \mathcal{H}_{ \mathrm{ac}Ê}( H^* ) = \overline{ M( H^* ) }$, there exists $v_\varepsilon \in M( H^* )$ such that $\|Êv - v_\varepsilon \|Ê\le \varepsilon$. Hence, for all $t \ge 0$,
\begin{equation}
\Big | \big \langle v , e^{ i t H } u \big \rangle - \big \langle v_\varepsilon , e^{ i t H } u \big \rangle \Big | \le m_u \varepsilon , \label{eq:a1}
\end{equation}
where $m_u := \sup_{ t \ge 0 } \| e^{ i t H }Êu \|$.

Let $g(t) := \langle v_\varepsilon , e^{ i t H } u \rangle$. We want to show that $g(t) \to 0$, as $t \to \infty$. We claim that the function $g$ is uniformly continuous on $\mathbb{R}$. Indeed, 
\begin{equation*}
\big |Êg( t + t' ) - g( t ) \big |Ê= \big | \big \langle ( e^{ - i t' H^* }Ê- \mathrm{Id}Ê) v_\varepsilon , e^{ i t H }Êu \big \rangle \big |Ê\le m_u \big \|Ê( e^{ - i t' H^* }Ê- \mathrm{Id}Ê) v_\varepsilon \big \|  ,
\end{equation*}
and we have that $\|Ê( e^{ - i t' H^* }Ê- \mathrm{Id}Ê) v_\varepsilon \| \to 0$, as $t' \to 0$. This shows that $g$ is uniformly continuous on $\mathbb{R}$. Next, we claim that $g$ is square integrable on $\mathbb{R}$. Indeed, let $t_0 > 0$. We have that
\begin{align*}
\int_{-\infty}^{t_0} | g(t) |^2 dt &= \int_{-\infty}^{t_0} \big|\big \langle v_\varepsilon , e^{ i t H } u \big \rangle \big |^2 dt \\
& = \int_{-\infty}^{0} \big |\big \langle v_\varepsilon , e^{ i ( t + t_0 ) H } u \big \rangle \big |^2 dt \\
 &= \int_{-\infty}^{0} \big | \big \langle e^{ - i t H^* } v_\varepsilon , e^{ i t_0 H } u \big \rangle \big |^2 dt \le c^*_{ v_\varepsilon } \| e^{ i t_0 H } u \| \le c^*_{ v_\varepsilon } m_u.
\end{align*}
Here we have used that $v_\varepsilon \in M( H^* )$, $c^*_{ v_\varepsilon }$ is defined in \eqref{eq:defM(H*)}, and $m_u := \sup_{ t \ge 0 } \| e^{ i t H }Êu \|$, as above. This shows that $g$ is square integrable on $\mathbb{R}$, and therefore, since $g$ is uniformly continuous, $g(t) \to 0$, as $t \to \infty$. Together with \eqref{eq:a1}, this concludes the proof of \eqref{eq:a2}.

To prove \eqref{eq:a3}, it suffices to observe that, if $\mathcal{H}_{ \mathrm{b} }( H ) = \{ 0 \}$ then $\mathcal{H}_{ \mathrm{ac}Ê}( H ) = \mathcal{H}$ by \eqref{eq:HacHbperp}, and hence \eqref{eq:a2} implies that
\begin{equation*}
\big \| K e^{ i t H } u \big \| \to 0 , \quad t \to \infty ,
\end{equation*}
for any finite-rank operator $K$. The result for compact operators then follows by approximation, using that $\| e^{ i t H } u \|$ is uniformly bounded in $t$, because $u \in S( H )$. The same holds for $K e^{ - i t H^* } u$, instead of $K e^{ i t H } u$.
\end{proof}
We are now ready to prove Theorem \ref{prop:range1}.
\begin{proof}[Proof of Theorem \ref{prop:range1}]
We only prove the statement for $W_-( H , H_0 )$. The proof for $W_+( H^* , H_0 )$ is identical. Since $M(H) \subset \mathcal{H}_{ \mathrm{ac} }( H )$, it suffices to prove that
\begin{equation*}
S( H ) \cap \mathcal{H}_{ \mathrm{ac} }( H ) \subset  \mathrm{Ran}( W_-( H , H_0 ) ) \subset S( H ) \cap M( H ).
\end{equation*}
The inclusion $\mathrm{Ran}( W_-( H , H_0 ) ) \subset S( H ) \cap M( H )$ follows from the intertwining property \eqref{eq:inter-1}. Indeed, let $u = W_-( H , H_0 ) v \in \mathrm{Ran}( W_-( H , H_0 ) )$. Then
\begin{equation*}
\big \|Êe^{Êi t H }ÊW_-( H , H_0 ) v \big \| = \big \|ÊW_-( H , H_0 ) e^{Êi t H_0 } v \big \| \le \| v \| ,
\end{equation*}
for all $t \ge 0$, where we have used that $W_-( H , H_0 )$ is a contraction and $e^{ i t H_0 }$ is unitary. Therefore $u \in S( H )$. Furthermore, for all $w \in \mathcal{H}$,
\begin{align*}
\int_0^\infty \big |Ê\langle e^{ - i t H }Êu , w \rangle \big |^2 dt &= \int_0^\infty \big |Ê\langle e^{ - i t H_0 }Êv , W_-( H , H_0 )^* w \rangle \big |^2 dt \\
& \le \mathrm{c}_v \big \|ÊW_-( H , H_0 )^* w \big \|^2 \le \mathrm{c}_v \| w \|^2 ,
\end{align*}
for some constant $\mathrm{c}_v > 0$, where we have used that $H_0$ is a self-adjoint operator with purely absolutely continuous spectrum, by Hypothesis \ref{V-1}, and that $W_-( H , H_0 )^*$ is a contraction by Proposition \ref{prop:existence_W-}. Hence $u \in M( H )$.

To prove that $S( H ) \cap \mathcal{H}_{ \mathrm{ac}Ê}( H ) \subset \mathrm{Ran}( W_-( H , H_0 ) )$, we consider a vector $u \in S( H ) \cap \mathcal{H}_{ \mathrm{ac}Ê}( H )$. We decompose
\begin{equation}
u = e^{ - i t H }Ê\Pi_{ \mathrm{pp} }( H_V ) e^{ i t H }Êu + e^{ - i t H }Ê\Pi_{ \mathrm{ac} }( H_V ) e^{ i t H }Êu , \label{eq:a7}
\end{equation}
where, we recall, $\Pi_{ \mathrm{pp} }( H_V )$ and $\Pi_{ \mathrm{ac} }( H_V )$ are the orthogonal projections onto the pure point and absolutely continuous spectral subspaces of $H_V$.  We claim that
\begin{equation}
\big \|Êe^{ - i t H }Ê\Pi_{ \mathrm{pp} }( H_V ) e^{ i t H }Êu \big \| \to 0 , \quad t \to \infty. \label{eq:a8}
\end{equation}
Indeed, since $e^{ - i t H }$ is a contraction for $t \ge 0$, we have that 
\begin{equation*}
\big \| e^{ - i t H }Ê\Pi_{ \mathrm{pp} }( H_V ) e^{ i t H }Êu \big \| \le \big \|Ê\Pi_{ \mathrm{pp} }( H_V ) e^{ i t H }Êu \big \|.
\end{equation*}
To prove that $\|Ê\Pi_{ \mathrm{pp} }( H_V ) e^{ i t H }Êu \| \to 0$, as $t \to \infty$, since $\mathrm{Ran} ( \Pi_{ \mathrm{pp} }( H_V ) )$ is finite dimensional according to Hypothesis \ref{V-1},  it suffices to verify that 
\begin{equation*}
\big \langle w , e^{ i t H }Êu \big \rangle \to 0,
\end{equation*}
as $t \to \infty$, for all $w \in \mathrm{Ran} ( \Pi_{ \mathrm{pp} }( H_V ) )$. We decompose $w = w_{\mathrm{b}} + w_{ \mathrm{ac} }$, with $w_{ \mathrm{b} } \in \mathcal{H}_{ \mathrm{b} }( H )$ and $w_{ \mathrm{ac}Ê} \in \mathcal{H}_{ \mathrm{b} }( H )^\perp$. Since $w_{ \mathrm{b} } \in \mathcal{H}_{ \mathrm{b} }( H ) = \mathcal{H}_{ \mathrm{b} }( H^* )$ (see Lemma \ref{lm:Hb}), we can write 
\begin{equation*}
\big \langle w_{ \mathrm{b} } , e^{ i t H }Êu \big \rangle = \sum_j \alpha_j e^{ i t \beta_j }Ê\langle w_j , u \rangle,
\end{equation*}
where the sum is finite, $w_j \in \mathcal{H}_{ \mathrm{b} }( H )$ are eigenvectors of $H^*$ corresponding to real eigenvalues $\beta_j$, and $\alpha_j \in \mathbb{C}$. Therefore 
\begin{equation*}
\big \langle w_{ \mathrm{b} } , e^{ i t H }Êu \big \rangle = 0,
\end{equation*}
since $u \in \mathcal{H}_{ \mathrm{ac} }( H )$ and since $\mathcal{H}_{ \mathrm{ac} }( H ) = \mathcal{H}_{ \mathrm{b} }( H )^\perp$ (see \eqref{eq:HacHbperp}). Moreover, since $w_{ \mathrm{ac}Ê} \in  \mathcal{H}_{ \mathrm{ac} }( H )$ and $u \in S( H )$, Lemma \ref{lm:LM1} shows that 
\begin{equation*}
\big \langle w_{ \mathrm{ac}Ê} , e^{ i t H }Êu \big \rangle \to 0,
\end{equation*}
as $t \to \infty$. Hence \eqref{eq:a8} is proven.

Next, we compute the limit of the second term in the right side of \eqref{eq:a7}, $e^{ - i t H }Ê\Pi_{ \mathrm{ac} }( H_V ) e^{ i t H }Êu$, as $t \to \infty$. To this end, we write
\begin{align*}
e^{ - i t H }Ê\Pi_{ \mathrm{ac} }( H_V ) e^{ i t H }Êu = e^{ - i t H } e^{ i t H_V }Ê\Pi_{ \mathrm{ac} }( H_V ) e^{ - i t H_V }Êe^{ i t H }Êu.
\end{align*}
Applying Cook's argument, using \eqref{eq:alpha} and Hypothesis \ref{V1} (exactly as in the proof of Proposition \ref{prop:existence_W-}; see Appendix \ref{app:existencewave}), it follows that 
\begin{equation*}
W_-( H , H_V ) := \underset{t\to \infty }{\slim} \, e^{ - i t H } e^{ i t H_V }Ê\Pi_{ \mathrm{ac} }( H_V ) 
\end{equation*}
exists. Likewise, since $u \in S( H )$, we have that $\int_0^\infty \big \|ÊC e^{ i s H }Êu \big \|^2 dsÊ< \infty$ and hence, by the same argument,
\begin{equation*}
\lim_{t \to \infty} \Pi_{ \mathrm{ac} }( H_V ) e^{ - i t H_V }Êe^{ i t H }Êu =: \Pi_{ \mathrm{ac} }( H_V ) W_- (H_V , H) u 
\end{equation*}
exists. The last two equations, combined with the fact that $e^{ - i t H } e^{ i t H_V }$ is uniformly bounded in $t \ge 0$ since $e^{-itH}$ is a contraction and $e^{itH_V}$ is unitary, imply that
\begin{align}
\lim_{ t \to + \infty } e^{ - i t H } e^{ i t H_V }Ê\Pi_{ \mathrm{ac} }( H_V ) e^{ - i t H_V }Êe^{ i t H }Êu = W_-( H , H_V ) \Pi_{ \mathrm{ac} }( H_V ) W_- (H_V , H) u  . \label{eq:jzpo1}
\end{align}

Equations \eqref{eq:a7}, \eqref{eq:a8} and \eqref{eq:jzpo1} yield
\begin{align*}
Êu = W_-( H , H_V ) \Pi_{ \mathrm{ac} }( H_V ) W_- (H_V , H) u  . 
\end{align*}
Since $W_-( H_V , H_0 )$ is bijective from $\mathcal{H}$ to $\mathcal{H}_{ \mathrm{ac}Ê}( H_V ) = \mathrm{Ran}( \Pi_{ \mathrm{ac}Ê}( H_V ) )$ by Hypothesis \ref{V0}, there exists $\tilde u \in \mathcal{H}$ such that $\Pi_{ \mathrm{ac} }( H_V ) W_- (H_V , H) u = W_-( H_V , H_0 ) \tilde u$. Applying the ``chain rule'' 
\begin{equation*}
u = W_-( H , H_V ) W_-( H_V , H_0 ) \tilde u = W_-( H , H_0 ) \tilde u,
\end{equation*}
we conclude that  $u \in \mathrm{Ran} ( W_-( H , H_0 ) )$. This completes the proof.
\end{proof}

\section{Proof of the main results}\label{sec:main}

In this section, we prove our main theorems. In Section \ref{subsec:spectral-proj}, we establish various properties of the spectral projections $E_H(I)$ (see \eqref{eq:defprojEH(I)}) and we show that $\mathrm{Ran} ( E_H( I ) ) \subset \mathrm{Ran}( W_-( H , H_0 ) )$. Next, in Section \ref{subsec:dissipative} we prove Theorem \ref{thm:Hp-Hd}, in Section \ref{subsec:asymp} we prove Theorem \ref{thm:AC}, Corollary \ref{cor:W+} and Theorem \ref{thm:nonAC},  and in Section \ref{subsec:nonblowup} we prove Theorem \ref{thm:nonblowup}. In Section \ref{subsec:localwave}, we state and prove some related results concerning local wave operators.

\subsection{Spectral projections}\label{subsec:spectral-proj}
Recall from Definition \ref{def:spec-sing} in Section \ref{subsec:hypoth} that $\lambda \in [ 0 , \infty )$ is said to be a regular spectral point of $H$ if there exists a closed interval $K_\lambda$, whose interior contains $\lambda$, such that the limit
\begin{equation*}
C \big ( H - ( \mu - i 0^+ ) \big )^{-1} C^* := \lim_{ \varepsilon \downarrow 0 } C \big ( H - ( \mu - i \varepsilon ) \big )^{-1} C^* 
\end{equation*}
exists uniformly in $\mu \in K_\lambda$ in the norm topology of $\mathcal{L}( \mathcal{H} )$.
A spectral singularity of $H$ is then a point $\lambda \in [ 0 , \infty )$ which is not a regular spectral point. If $\lambda$ is a regular spectral point of $H$ then, taking adjoints, we see that it is also a regular spectral point of $H^*$ in the sense that
\begin{equation*}
C \big ( H^* - ( \mu + i 0^+ ) \big )^{-1} C^* := \lim_{ \varepsilon \downarrow 0 } C \big ( H^* - ( \mu + i \varepsilon ) \big )^{-1} C^* 
\end{equation*}
exists uniformly in $\mu \in K_\lambda$ in the norm topology of $\mathcal{L}( \mathcal{H} )$.

As already mentioned in Section \ref{subsec:ideas}, if $I \subset [ 0 , \infty )$ is a closed interval that does not contain any spectral singularities, it is possible to define a ``spectral projection'' for $H$ in $I$ by mimicking Stone's formula,
\begin{equation}\label{eq:spec-proj}
E_H( I ) := \underset{ \varepsilon \downarrow 0 }{\wlim} \frac{1}{ 2 i \pi } \int_I \big ( ( H - ( \lambda + i \varepsilon ) )^{-1} - ( H - ( \lambda - i \varepsilon ) )^{-1} \big ) d \lambda.
\end{equation}
Proposition \ref{prop:spectr-sing} below shows that $E_H( I )$ is indeed a well-defined (but generally not an orthogonal) projection under our conditions. Such spectral projections for non-self-adjoint operators have been considered in \cite{Du58_01,Sc60_01}; see also \cite{DuSc71_01} for a textbook presentation of the related theory of spectral operators. In \cite{Sc60_01}, a spectral singularity corresponds to an exceptional point $\lambda_0$  outside of which the ``spectral resolution'' $I \mapsto E_H( I )$ determined by \eqref{eq:spec-proj} is countably additive and uniformly bounded. In our context, this is a weaker requirement than that of Definition \ref{def:spec-sing} in Section \ref{subsec:hypoth}.

The proof of the next proposition invokes known arguments. We provide details in Appendix \ref{app:proj}.
\begin{proposition}\label{prop:spectr-sing}
Suppose that Hypotheses \ref{V-1}, \ref{V0} and \ref{V1} hold. Let $I \subset [ 0 , \infty )$ be a closed interval. If $I$ does not contain any spectral singularities, the weak limit in \eqref{eq:spec-proj} exists in $\mathcal{L}( \mathcal{H} )$ and satisfies
\begin{equation}\label{eq:proj-inters-int}
E_H( I_1 ) E_H( I_2 ) = E_H( I_1 \cap I_2 ) , 
\end{equation}
for any closed intervals $I_1, I_2 \subset I$, with the convention that $E_H( \emptyset ) = 0$. In particular, $E_H( I )$ is a projection. Its adjoint is given by
\begin{equation}\label{eq:spec-proj-adj}
E_H(I)^* = E_{H^*}( I ) := \underset{ \varepsilon \downarrow 0 }{\wlim} \frac{1}{ 2 i \pi } \int_I \big ( ( H^* - ( \lambda + i \varepsilon ) )^{-1} - ( H^* - ( \lambda - i \varepsilon ) )^{-1} \big ) d \lambda.
\end{equation}
Moreover, for all $t \in \mathbb{R}$,
\begin{equation}\label{eq:spec-proj-exp}
e^{ i t H }ÊE_H( I ) = \underset{ \varepsilon \downarrow 0 }{\wlim} \frac{1}{ 2 i \pi } \int_I e^{ i t \lambda } \big ( ( H - ( \lambda + i \varepsilon ) )^{-1} - ( H - ( \lambda - i \varepsilon ) )^{-1} \big ) d \lambda ,
\end{equation}
and the family of operators $\{ e^{ i t H } E_H ( I ) \}_{ t \in \mathbb{R}Ê} \subset \mathcal{L}( \mathcal{H} )$ is uniformly bounded in $t \in \mathbb{R}$.
\end{proposition}
In the statement of Proposition \ref{prop:spectr-sing}, we have only listed properties that will be used in the proof of our main results. For further general properties concerning the spectral projections $E_H( I )$, we refer to \cite{Du58_01} and \cite{DuSc71_01}. The main property that will be used is given in the next theorem.
\begin{theorem}\label{prop:proj-ran}
Suppose that Hypotheses \ref{V-1}, \ref{V0} and \ref{V1} hold. Let $I \subset [ 0 , \infty )$ be a closed interval containing no spectral singularities of $H$. Then
\begin{equation*}
\mathrm{Ran} ( E_H( I ) ) \subset \mathrm{Ran}( W_-( H , H_0 ) ), \quad \mathrm{Ran} ( E_{H^*}( I ) ) \subset \mathrm{Ran}( W_-( H^* , H_0 ) ).
\end{equation*}
\end{theorem}
\begin{proof}
We prove the first inclusion; the proof of the second one is identical. According to Theorem \ref{prop:range1}, it suffices to prove that $\mathrm{Ran}( E_H ( I ) ) \subset S( H ) \cap \mathcal{H}_{ \mathrm{ac} }( H )$, where $S( H )$ is defined in \eqref{eq:S(H)} and $\mathcal{H}_{ \mathrm{ac}Ê}( H )$ is the absolutely continuous spectral subspace of $H$ defined in Section \ref{subsec:absolute}.

Let $u \in \mathrm{Ran}( E_H ( I ) )$. We first show that $u \in \mathcal{H}_{ \mathrm{ac} }( H ) = \mathcal{H}_{ \mathrm{b} }( H )^\perp = \mathcal{H}_{ \mathrm{b} }( H^* )^\perp$ (see \eqref{eq:HacHbperp} and Lemma \ref{lm:Hb}). By Lemma \ref{lm:Hb} and Hypothesis \ref{V-1}, we know that $H^*$ has at most finitely many real eigenvalues, and that all these eigenvalues are strictly negative.  If $v$ is an eigenstate associated to an eigenvalue $\lambda_0 < 0$ of $H^*$, we have that
\begin{align*}
\langle v , E_H( I ) u \rangle & = \lim_{ \varepsilon \downarrow 0 } \frac{1}{ 2 i \pi } \int_I \big \langle v , \big ( ( H - ( \lambda + i \varepsilon ) )^{-1} - ( H - ( \lambda - i \varepsilon ) )^{-1} \big ) u \big \rangle d \lambda \\
&= \lim_{ \varepsilon \downarrow 0 } \frac{1}{ 2 i \pi } \int_I \big \langle v , \big ( ( \lambda_0 - ( \lambda + i \varepsilon ) )^{-1} - ( \lambda_0 - ( \lambda - i \varepsilon ) )^{-1} \big ) u \big \rangle d \lambda = 0 ,
\end{align*}
because $\lambda_0 < 0$ and $I \subset [ 0 , \infty )$. This implies that $u \in \mathcal{H}_{ \mathrm{b} }( H^* )^\perp$.

The fact that $u \in S( H )$ follows from \eqref{eq:spec-proj-exp}.
\end{proof}

\subsection{The ``dissipative space''}\label{subsec:dissipative}

In this section we prove Theorem \ref{thm:Hp-Hd}. We recall that the vector spaces $\mathcal{H}_{ \mathrm{p} }( H )$ and $\mathcal{H}_{ \mathrm{d}Ê}( H )$ are defined by $\mathcal{H}_{ \mathrm{p} }( H ) = \mathrm{Span} \{ u \in \mathrm{Ran}( \Pi_\lambda ) , \lambda \in \sigma( H ) , \mathrm{Im} \, \lambda < 0  \}$, where $\Pi_\lambda$ denotes the Riesz projection \eqref{eq:defPilambda} associated to an eigenvalue $\lambda$, and $\mathcal{H}_{ \mathrm{d} }( H ) = \{ u \in \mathcal{H} ,  \| e^{ - i t H }Êu \|Ê\to 0 , t \to \infty \}$. The corresponding spaces for $H^*$ are defined similarly (see \eqref{eq:defH^*p}--\eqref{eq:defH*d}). Using the results obtained in the preceding section, we prove the following theorem.

\begin{theorem}\label{prop:proj-ran2}
Suppose that Hypotheses \ref{V-1}--\ref{V3} hold. Then
\begin{equation*}
\mathcal{H}_{ \mathrm{d} }( H ) = \mathcal{H}_{ \mathrm{p} }( H ), \quad \mathcal{H}_{ \mathrm{d} }( H^* ) = \mathcal{H}_{ \mathrm{p} }( H^* ).
\end{equation*}
\end{theorem}
\begin{proof}
We prove that $\mathcal{H}_{ \mathrm{d} }( H ) = \mathcal{H}_{ \mathrm{p} }( H )$, the other equality follows in the same way. Applying Proposition \ref{prop:randense_W-} and Theorem \ref{prop:proj-ran}, we deduce that
\begin{align}
\mathcal{H}_{Ê\mathrm{b} }( H ) \oplus \mathcal{H}_{Ê\mathrm{d} }( H ) = \mathrm{Ran} ( W_+( H^* , H_0 ) )^\perp \subset \bigcap_{I \subset [ 0 , \infty ) } \mathrm{Ran}( E_{H^*}( I ) )^\perp , \label{eq:hzp1}
\end{align}
where $E_{H^*}( I )$ is defined in \eqref{eq:spec-proj-adj} and where the intersection runs over all closed intervals $I \subset [ 0 , \infty )$ with the property that $I$ does not contain any spectral singularities of $H$. By Proposition \ref{prop:spectr-sing}, $\mathrm{Ran}( E_{H^*}( I ) )^\perp = \mathrm{Ker}( E_{H}( I ) )$. Let
\begin{equation}
\mathcal{K} := \bigcap_{I \subset [ 0 , \infty )} \mathrm{Ran}( E_{H^*}( I ) )^\perp =  \bigcap_{I \subset [ 0 , \infty )} \mathrm{Ker}( E_{H}( I ) ) , \label{eq:def_K}
\end{equation}
where, again, it is understood implicitly that the intersection runs over all closed intervals $I \subset [ 0 , \infty )$ with the property that $I$ does not contain any spectral singularities of $H$. The set $\mathcal{K}$ is a closed vector space and we claim that 
\begin{equation}
\mathcal{K} \subset \mathcal{H}_{Ê\mathrm{b} }( H ) \oplus \mathcal{H}_{ \mathrm{p} }( H ). \label{eq:incl_K}
\end{equation}
Some of the arguments used to prove this claim can be found in \cite{Sc60_01}. It is convenient to work with a bounded operator instead of the unbounded $H$. Let
\begin{equation*}
R := ( H - i )^{-1}.
\end{equation*}
Obviously, $R \in \mathcal{L}( \mathcal{H} )$ since the spectrum of $H$ is contained in $\{ z \in \mathbb{C} , \mathrm{Im}( z ) \le 0 \}$, and it follows from Proposition \ref{prop:spectr-sing} that $R$ commutes with $E_H( I )$, which yields that $ R \,  \mathcal{K} \subset \mathcal{K}$. We will show that $\mathcal{K}$ is contained in the vector space spanned by the generalized eigenstates of $R$. Since the generalized eigenstates of $R$ coincide with the generalized eigenstates of $H$, this will prove that $\mathcal{K} \subset \mathcal{H}_{Ê\mathrm{b} }( H ) \oplus \mathcal{H}_{ \mathrm{p} }( H )$.

For all $\mu \in \mathbb{C}\setminus \{ 0 \}$ such that $\mu^{-1} + i \notin \sigma( H )$, the resolvent equation implies that
\begin{equation}\label{eq:resolv_R*}
( R - \mu )^{-1} = - \frac{1}{ \mu } \Big ( \mathrm{Id} + \frac{1}{ \mu } \big ( H - ( \frac{1}{ \mu } + i ) \big )^{-1} \Big ).
\end{equation}
Let $\Sigma_{ \mathrm{e}Ê} := \{ \mu \in \mathbb{C} \setminus \{ 0 \} , \mu^{-1}Ê+ i \in [ 0 , \infty ) \} \cup \{Ê0 \} = \{Ê( \lambda - i )^{-1} , \lambda \in [ 0 , \infty ) \} \cup \{ 0 \}$. One verifies that
\begin{equation*}
\Sigma_{ \mathrm{e} } = \mathcal{C} \cap \{ z \in \mathbb{C} , \, \mathrm{Re}(z) \ge 0 \} ,
\end{equation*}
where $\mathcal{C} = \mathcal{C} ( \frac{i}{2} , \frac{1}{2} )$ denotes the circle centered at $i/2$ and of radius $1/2$. Moreover, by \eqref{eq:resolv_R*}, $\Sigma_{ \mathrm{e}Ê}$ is contained in the essential spectrum of $R$. In fact, by the  spectral mapping theorem (see, e.g., Lemma 2 in Section XIII.4 of \cite{ReedSimon4}), we have that
\begin{align*}
\sigma( R ) &= \Sigma_{ \mathrm{e}Ê} \cup \big \{ \mu \in \mathbb{C} \setminus \{ 0 \} , \, \mu^{-1}Ê+ i \text{ is an eigenvalue of }  H \big \} \\
&= \Sigma_{ \mathrm{e}Ê} \cup \big \{ ( \lambda - i )^{-1}Ê, \, \lambda \text{ is an eigenvalue of }  H \big \}.
\end{align*}
We define a contour whose interior contains $\Sigma_{Ê\mathrm{e} }$ but no other points of the spectrum of $R$. Let $-\mathrm{e}_0 < 0$ be the largest real eigenvalue of $H$. For $\varepsilon > 0$ small enough, let 
\begin{equation*}
\Gamma_\varepsilon := \Gamma_{1,\varepsilon} \cup \Gamma_{2,\varepsilon} \cup \Gamma_{3,\varepsilon} \cup \Gamma_{4,\varepsilon}
\end{equation*}
be the curve oriented counterclockwise, defined by
\begin{align}
& \Gamma_{1,\varepsilon} := \Big \{ ( \lambda - i - i \varepsilon )^{-1} , \, - \frac12 \mathrm{e}_0 \le \lambda \le \gamma_1( \varepsilon )  \Big \} , \quad \gamma_1( \varepsilon ) := ( \varepsilon^{-2}Ê- ( 1 + \varepsilon )^2 )^{\frac12} , \notag \\
& \Gamma_{2,\varepsilon} = \Big \{ \varepsilon e^{ i \theta } , \, \theta_1( \varepsilon ) \le \theta \le \theta_3( \varepsilon ) \Big \} , \quad \varepsilon e^{ i \theta_1 (\varepsilon ) } = ( \gamma_1( \varepsilon ) - i - i \varepsilon )^{-1} , \quad \theta_1( \varepsilon )  \in ( 0 , \frac{\pi}{2 } ) , \notag \\
& \phantom{ \Gamma_{2,\varepsilon} = \Big \{ \varepsilon e^{ i \theta } , \, \theta_1( \varepsilon ) \le \theta \le \theta_3( \varepsilon ) \Big \} , \quad } \,  \varepsilon e^{ i \theta_3 (\varepsilon ) } = ( \gamma_3( \varepsilon ) - i + i \varepsilon )^{-1}, \quad \theta_3 ( \varepsilon ) \in ( \frac{ 3 \pi }{ 2 } , 2 \pi ) ,  \notag \\
& \Gamma_{3,\varepsilon} = \Big \{ ( \lambda - i + i \varepsilon )^{-1} , \, - \frac12 \mathrm{e}_0 \le \lambda \le \gamma_3 ( \varepsilon )  \Big \} , \quad \gamma_3( \varepsilon ) := ( \varepsilon^{-2}Ê- ( 1 - \varepsilon )^2 )^{\frac12} , \notag \\
& \Gamma_{4,\varepsilon} = \Big \{ ( - \frac12 \mathrm{e}_0 - i + i x )^{-1}  , \,  - \varepsilon \le x \le \varepsilon \Big \}. \label{eq:def-Gammaeps}
\end{align}
See Figure \ref{fig2}.

\begin{figure}[H] 
\begin{center}
\begin{tikzpicture}[scale=0.6, every node/.style={scale=0.9}]

   \draw[->](-4,2) -- (12,2);      
      \draw[->](4,0) -- (4,8);

    \draw[very thick] (4,2) arc (-90:90:2cm);
    
      \draw (4.2,2.3) arc (-82:100:1.6cm);

      \draw (4.2,1.7) arc (-86:98:2.4cm);

      \draw[-] (3.7,5.46) -- (3.7,6.47);

      \draw (4.2,2.3) arc (60:304:0.35cm);

    \draw[-] (4.2,4.2)--(4.4,4.4);
    \draw[-] (4.2,4.4)--(4.4,4.2);   

	\draw[-] (4.8,4.8)--(5,5);
    \draw[-] (4.8,5)--(5,4.8);   
    
	\draw[-] (3.3,5.8)--(3.5,6);
    \draw[-] (3.3,6)--(3.5,5.8);       
    
	\draw[-] (2.8,5.6)--(3,5.8);
    \draw[-] (2.8,5.8)--(3,5.6);           

	\draw[-] (2,4.6)--(2.2,4.8);
    \draw[-] (2,4.8)--(2.2,4.6);           
    
    	\draw[-] (4.1,5.2)--(4.3,5.4);
    \draw[-] (4.1,5.4)--(4.3,5.2);           
    
	\draw[-] (3.6,4.8)--(3.8,5);
    \draw[-] (3.6,5)--(3.8,4.8);

	\draw[-] (4.2,4.2)--(4.4,4.4);
    \draw[-] (4.2,4.4)--(4.4,4.2);   

	\draw[-] (3,5)--(3.2,5.2);
    \draw[-] (3,5.2)--(3.2,5);           

	\draw[-] (4.5,4.6)--(4.7,4.8);
    \draw[-] (4.5,4.8)--(4.7,4.6);           

	\draw[-] (5,4.4)--(5.2,4.6);
    \draw[-] (5,4.6)--(5.2,4.4);

	\draw[->](6.43,4) -- (6.43,4);
	
	\draw[->](5.58,3.9) -- (5.58,3.89);
	
	\draw[->](3.68,2) -- (3.68,1.99);
	
	\draw[->](3.7,6.1) -- (3.7,6.09);
	
    \draw (7.1,3.4) node[above] {\small$\Gamma_{3,\varepsilon}$};	
    \draw (5.1,3.4) node[above] {\small$\Gamma_{1,\varepsilon}$};	
    \draw (3.2,1.9) node[above] {\small$\Gamma_{2,\varepsilon}$};	
    \draw (3.2,6) node[above] {\small$\Gamma_{4,\varepsilon}$};

        \end{tikzpicture}
\caption{ \footnotesize  \textbf{The contour $\Gamma_{\varepsilon}$.} The thick arc represents the essential spectrum $\Sigma_{\mathrm{e}}$ of $R = ( H - i )^{-1}$ and is contained in the circle $\mathcal{C} ( \frac{i}{2} , \frac{1}{2} )$. The eigenvalues of $R$ are contained in the closed disc of radius $\frac12$ centered at $\frac{i}{2}$. }\label{fig2}
\end{center}
\end{figure}
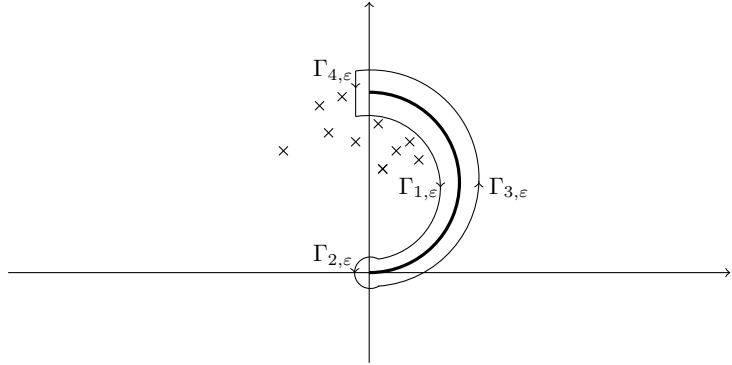

Using Riesz projections, we have that 
\begin{equation}\label{eq:spec-non-self}
\mathrm{Id}Ê= \Pi_{Ê\mathrm{pp} }( R ) + \frac{1}{2i\pi} \oint_{\Gamma_\varepsilon} ( \mu - R  )^{-1} d \mu ,
\end{equation}
for $\varepsilon > 0$ small enough, where $\Pi_{ \mathrm{pp} }( R ) = \sum_\lambda \Pi_\lambda( R )$ denotes the sum of the Riesz projections, defined as in \eqref{eq:defPilambda}, associated to all the isolated eigenvalues of $R$ corresponding to generalized eigenstates. Since the generalized eigenstates of $R$ coincide with those of $H$, we have that
\begin{equation}
\mathrm{Ran} ( \Pi_{ \mathrm{pp} }( R ) ) = \mathcal{H}_{ \mathrm{b} }( H ) \oplus \mathcal{H}_{ \mathrm{p} }( H ). \label{Ran_PIpp}
\end{equation}
Therefore, to prove \eqref{eq:incl_K}, using \eqref{eq:spec-non-self} and \eqref{Ran_PIpp}, it suffices to show that 
\begin{equation}
\lim_{ \varepsilon \downarrow 0 } \oint_{\Gamma_\varepsilon} ( \mu - R  )^{-1}  u \, d \mu = 0 , \label{eq:limit_gammaeps}
\end{equation}
for any $u \in \mathcal{K}$. Due to the presence of ``spectral singularities'' in the essential spectrum $\Sigma_{ \mathrm{e} }$ of $R$, we cannot compute the limit in \eqref{eq:limit_gammaeps} directly. Hence, before taking the limit, we compose both sides of \eqref{eq:spec-non-self} by an operator that regularizes the resolvent $( \mu - R )^{-1}$ in a neighborhood of the spectral singularities.

Let $\{ \lambda_1 , \dots \lambda_n \} \subset [ 0 , \infty )$ denote the set of spectral singularities of $H$, see Hypothesis \ref{V3}, and let $\mu_j = ( \lambda_j - i )^{-1}$, $j = 1 , \dots , n$, be the corresponding ``spectral singularities'' of $R$. Composing \eqref{eq:spec-non-self} by $R^4 \prod_{ j = 1 }^n ( R - \mu_j )^{\nu_j}$ gives
\begin{align*}
R^4 \prod_{ j = 1 }^n ( R - \mu_j )^{\nu_j}Ê=& R^4 \prod_{ j = 1 }^n ( R - \mu_j )^{\nu_j} \Pi_{ \mathrm{pp} }( R ) - \frac{1}{2i\pi} \oint_{\Gamma_\varepsilon} R^4 \prod_{ j = 1 }^n ( R - \mu_j )^{\nu_j} ( R - \mu Ê)^{-1} d \mu .
\end{align*}
By Cauchy's formula, this can be rewritten as
\begin{align}
R^4 \prod_{ j = 1 }^n ( R - \mu_j )^{\nu_j}Ê=& R^4 \prod_{ j = 1 }^n ( R - \mu_j )^{\nu_j} \Pi_{ \mathrm{pp} }( R ) -  \frac{1}{ 2 i \pi } \oint_{\Gamma_\varepsilon} \mu^4 \prod_{ j = 1 }^n ( \mu - \mu_j )^{\nu_j} ( R - \mu Ê)^{-1} d \mu . \label{eq:cauchy}
\end{align}
The product $\mu^4 \prod_{ j = 1 }^n ( \mu - \mu_j )^{\nu_j}$ regularizes the resolvent $( R - \mu )^{-1}$ in neighborhoods of $0$ and of the spectral singularities $\mu_j$. In Appendix \ref{app:proj}, we prove that 
\begin{align}
& \underset{\varepsilon \downarrow 0}{\wlim} \oint_{\Gamma_{\varepsilon}} \mu^4 \prod_{ j = 1 }^n ( \mu - \mu_j )^{\nu_j} ( R - \mu Ê)^{-1}  d \mu \notag \\
& = \underset{\varepsilon \downarrow 0}{\wlim} \int_{ 0 }^{ \infty } ( \lambda - i )^{-4} \prod_{ j = 1 }^n \big ( ( \lambda - i )^{-1} - \mu_j \big )^{\nu_j}  \notag \\
& \qquad \qquad \qquad \big [ \big ( H - ( \lambda - i \varepsilon )Ê\big )^{-1} - \big ( H - ( \lambda + i \varepsilon )Ê\big )^{-1} \big ]  d \lambda , \label{eq:weak-Gamma}
\end{align}
and that the weak limit in the right side of \eqref{eq:weak-Gamma} exists in $\mathcal{L}( \mathcal{H} )$. Moreover, there exists a constant $\mathrm{c} > 0$ such that, for all $v , w \in \mathcal{H}$,
\begin{align}
\lim_{ \varepsilon \downarrow 0 } \int_{ 0 }^{ \infty } & \Big | ( \lambda - i )^{-4} \prod_{ j = 1 }^n \big ( ( \lambda - i )^{-1} - \mu_j \big )^{\nu_j} \notag \\
& \big \langle v , \big [ \big ( H - ( \lambda - i \varepsilon )Ê\big )^{-1} - \big ( H - ( \lambda + i \varepsilon )Ê\big )^{-1} \big ] w \big \rangle \Big |  d \lambda \le \mathrm{c} \| v \| \|w \| , \label{eq:unifbound_tildeE}
\end{align}
see Appendix \ref{app:proj}. As in \eqref{eq:spec-proj}, we set
\begin{align*}
\tilde{E}_{ H }( I ) := \underset{\varepsilon \downarrow 0}{\wlim} \frac{1}{2i\pi} \int_I & ( \lambda - i )^{-4} \prod_{ j = 1 }^n \big ( ( \lambda - i )^{-1} - \mu_j \big )^{\nu_j} \\
& \big [ \big ( H - ( \lambda + i \varepsilon )Ê\big )^{-1} - \big ( H - ( \lambda - i \varepsilon )Ê\big )^{-1} \big ]  d \lambda ,
\end{align*}
for any closed interval $I \subset [ 0 , \infty )$. Taking the weak limit $\varepsilon \to 0^+$ in \eqref{eq:cauchy}, we obtain that
\begin{align}
R^4 \prod_{ j = 1 }^n ( R - \mu_j )^{\nu_j}Ê=& R^4 \prod_{ j = 1 }^n ( R - \mu_j )^{\nu_j} \Pi_{ \mathrm{pp} }( R ) + \tilde{E}_{ H }( [ 0 , \infty ) ) . \label{eq:cauchy-3}
\end{align}

Next, we prove that $\mathcal{K}$ (see \eqref{eq:def_K}) is contained in $\mathrm{Ker}( \tilde{E}_{ H }( [ 0 , \infty ) ) )$. Taking the Laplace transform of the resolvent $(H-i)^{-1}$ (see \eqref{eq:Laplace} in Appendix \ref{app:proj}) it follows from Proposition \ref{prop:spectr-sing} that, if $I$ is a closed interval not containing any spectral singularities of $H$, then
\begin{equation*}
\tilde{E}_{ H }( I ) = ( H - i )^{-4} \prod_{ j = 1 }^n \big ( ( H - i )^{-1} - \mu_j \big )^{\nu_j} E_{H}( I ).
\end{equation*}
Hence $\mathrm{Ker}( E_{ H }( I ) ) \subset \mathrm{Ker}( \tilde{E}_{ H }( I ) )$. For any $u \in \mathcal{K} = \bigcap_I \mathrm{Ker}( E_{ H }( I ) )$, where the intersection runs over all closed intervals $I \subset [ 0 , \infty ) \setminus \{ \lambda_1 , \dots , \lambda_n \}$, and any  $v \in \mathcal{H}$, $\tau > 0$ small enough, we thus have that
\begin{align*}
\big \langle v , \tilde{E}_{ H }( [ 0 , \infty ) ) u \big \rangle = \sum_{j=1}^n \big \langle v , \tilde{E}_{ H }( [ \lambda_j - \tau , \lambda_j + \tau ] ) u \big \rangle .
\end{align*}
Letting $\tau \to 0$, using \eqref{eq:unifbound_tildeE} and Lebesgue's dominated convergence theorem, we deduce that $ \langle v , \tilde{E}_{ H }( [ 0 , \infty ) ) u  \rangle = 0$, which in turn implies that $\tilde{E}_{ H }( [ 0 , \infty ) ) u =0$. It then follows from \eqref{eq:cauchy-3} that, for all $u \in \mathcal{K}$,
\begin{align}
R^4 \prod_{ j = 1 }^n ( R - \mu_j )^{\nu_j}Êu = R^4 \prod_{ j = 1 }^n ( R - \mu_j )^{\nu_j} \Pi_{ \mathrm{pp} }( R ) u. \label{eq:last}
\end{align}

Given Equation \eqref{eq:last}, we can conclude the proof of Theorem \ref{prop:proj-ran2} as follows. Let $u \in \mathcal{K}$. We want to prove that $u \in \mathcal{H}_{ \mathrm{b} }( H ) \oplus \mathcal{H}_{ \mathrm{p} }( H )$. If $\Pi_{ \mathrm{pp} }( R ) u \neq 0$, since $R^4 \prod_{ j = 1 }^n ( R - \mu_j )^{\nu_j}$ is invertible in $\mathrm{Ran}( \Pi_{ \mathrm{pp} }( R ) )$, Equation \eqref{eq:last} shows that $u \in \mathrm{Ran} ( \Pi_{ \mathrm{pp} }( R ) ) = \mathcal{H}_{ \mathrm{b} }( H ) \oplus \mathcal{H}_{ \mathrm{p} }( H )$ (see \eqref{Ran_PIpp}), as claimed. Therefore it remains to show that if $u \neq 0$ then $\Pi_{ \mathrm{pp} }( R ) u \neq 0$. Suppose that $\Pi_{ \mathrm{pp} }( R ) u =0$. Then \eqref{eq:last} implies that $R^4 \prod_{ j = 1 }^n ( R - \mu_j )^{\nu_j}Êu = 0$, which, since $R$ is invertible, yields $\prod_{ j = 1 }^n ( R - \mu_j )^{\nu_j}Êu = 0$. The identity
\begin{equation}
( R - \mu_1 )^{\nu_1} \prod_{ j = 2 }^n ( R - \mu_j )^{\nu_j}Êu = 0 , \label{eq:ksyt1}
\end{equation}
shows that, if $\mu_1 = 0$, then $\prod_{ j = 2 }^n ( R - \mu_j )^{\nu_j}Êu = 0$ (because $R$ is invertible). If $\mu_1 \neq 0$, expanding the product $( R - \mu_1 )^{\nu_1}$, we obtain that
\begin{equation*}
(- \mu_1)^{\nu_1}Ê\prod_{ j = 2 }^n ( R - \mu_j )^{\nu_j}Êu = \sum_{l=0}^{\nu_1-1} \binom{\nu_1}{l} ( - \mu_1 )^l R^{\nu_1-l} \prod_{ j = 2 }^n ( R - \mu_j )^{\nu_j}Êu.
\end{equation*}
This shows that $\prod_{ j = 2 }^n ( R - \mu_j )^{\nu_j}Êu \in \mathcal{D}( H )$. Composing the left side of \eqref{eq:ksyt1} by $H - i$, we arrive at
\begin{equation*}
( \mathrm{Id}Ê- \mu_1 (H - i ) ) ( R - \mu_1 )^{\nu_1-1} \prod_{ j = 2 }^n ( R - \mu_j )^{\nu_j}Êu = 0 ,
\end{equation*}
or equivalently, since $\mu_1^{-1}Ê= \lambda_1 - i$,
\begin{equation*}
( \lambda_1Ê- H ) ( R - \mu_1 )^{\nu_1-1} \prod_{ j = 2 }^n ( R - \mu_j )^{\nu_j}Êu = 0 .
\end{equation*}
Since $\lambda_1 \in [ 0 , \infty )$, we know that $\lambda_1$ is not an eigenvalue of $H$ by Lemma \ref{lm:Hb} and Hypothesis \ref{V-1}. Therefore
\begin{equation}
 ( R - \mu_1 )^{\nu_1-1} \prod_{ j = 2 }^n ( R - \mu_j )^{\nu_j}Êu = 0 . \label{eq:ksyt2}
\end{equation}
We have proven that \eqref{eq:ksyt1} implies \eqref{eq:ksyt2}. Proceeding by induction, we obtain that $u = 0$, which concludes our proof that $\mathcal{K} \subset \mathcal{H}_{ \mathrm{b} }( H ) \oplus \mathcal{H}_{ \mathrm{p} }( H )$.

By \eqref{eq:hzp1}, this implies that $\mathcal{H}_{ \mathrm{b} }( H ) \oplus \mathcal{H}_{ \mathrm{d} }( H ) \subset \mathcal{H}_{ \mathrm{b} }( H ) \oplus \mathcal{H}_{ \mathrm{p} }( H )$. Since $\mathcal{H}_{ \mathrm{b} }( H )$ and $\mathcal{H}_{ \mathrm{d} }( H )$ are orthogonal and since $\mathcal{H}_{Ê\mathrm{p} }( H ) \subset \mathcal{H}_{ \mathrm{d} }( H )$, this finally establishes that $\mathcal{H}_{Ê\mathrm{p} }( H ) = \mathcal{H}_{ \mathrm{d} }( H )$.
\end{proof}
It should be noted that Equation \eqref{eq:cauchy-3} implies the spectral decomposition formula \eqref{eq:spectral_decomp}. To see this, it suffices to compose both sides of \eqref{eq:cauchy-3} by $(H-i)^4=R^{-4}$ and invoke arguments of Appendix \ref{app:proj}. Since \eqref{eq:cauchy-3} suffices for the purpose of proving Theorem \ref{prop:proj-ran2}, we do not elaborate.

\subsection{Asymptotic completeness}\label{subsec:asymp}

We begin this section by proving Theorem \ref{thm:AC}. More precisely, assuming that $H$ does not have spectral singularities in $[0,\infty)$, we prove that the wave operators $W_-( H , H_0 )$ and $W_+( H^* , H_0 )$ are asymptotically complete.
\begin{theorem}\label{prop:AC}
Suppose that Hypotheses \ref{V-1}--\ref{V3} hold and that $H$ has no spectral singularities in $[0,\infty)$. Then the wave operators $W_- (H ,H_0)$ and $W_+( H^* , H_0 )$ are asymptotically complete. In particular, the restriction of $H$ to $( \mathcal{H}_{ \mathrm{b} }( H ) \oplus \mathcal{H}_{ \mathrm{p} }( H^* ) )^\perp$ is similar to $H_0$ and the restriction of $H^*$ to $( \mathcal{H}_{ \mathrm{b} }( H ) \oplus \mathcal{H}_{ \mathrm{p} }( H ) )^\perp$ is similar to $H_0$.
\end{theorem}
\begin{proof}
We want to prove that
\begin{equation*}
Ê\mathrm{Ran}( W_-( H , H_0 ) ) = \big ( \mathcal{H}_{ \mathrm{b} }( H ) \oplus \mathcal{H}_{ \mathrm{p} }( H^* ) \big )^\perp, \quad \mathrm{Ran}( W_+( H^* , H_0 ) ) = \big ( \mathcal{H}_{ \mathrm{b} }( H ) \oplus \mathcal{H}_{ \mathrm{p} }( H ) \big )^\perp.
\end{equation*}
By Theorem \ref{prop:proj-ran2} and Proposition \ref{prop:randense_W-}, we have that
\begin{equation*}
\overline{Ê\mathrm{Ran}( W_-( H , H_0 ) ) } = \big ( \mathcal{H}_{ \mathrm{b} }( H ) \oplus \mathcal{H}_{ \mathrm{p} }( H^* ) \big )^\perp, \quad \overline{\mathrm{Ran}( W_+( H^* , H_0 ) )} = \big ( \mathcal{H}_{ \mathrm{b} }( H ) \oplus \mathcal{H}_{ \mathrm{p} }( H ) \big )^\perp.
\end{equation*}
Therefore we must prove that 
\begin{equation*}
\big ( \mathcal{H}_{ \mathrm{b} }( H ) \oplus \mathcal{H}_{ \mathrm{p} }( H^* ) \big )^\perp \subsetÊ\mathrm{Ran}( W_-( H , H_0 ) ), \quad \big ( \mathcal{H}_{ \mathrm{b} }( H ) \oplus \mathcal{H}_{ \mathrm{p} }( H ) \big )^\perp \subset \mathrm{Ran}( W_+( H^* , H_0 ) ).
\end{equation*}
Since
\begin{equation*}
\mathrm{Ran}( W_-( H , H_0 ) ) = S( H ) \cap \mathcal{H}_{Ê\mathrm{ac}Ê}( H ), \quad \mathrm{Ran}( W_+( H^* , H_0 ) ) = S( H^* ) \cap \mathcal{H}_{Ê\mathrm{ac}Ê}( H )
\end{equation*}
by Theorem \ref{prop:range1}, and since $\mathcal{H}_{ \mathrm{b}Ê}( H )^\perp = \mathcal{H}_{Ê\mathrm{ac}Ê}( H )$ (see \eqref{eq:HacHbperp}), it suffices to verify that
\begin{equation}
\big ( \mathcal{H}_{ \mathrm{b} }( H ) \oplus \mathcal{H}_{ \mathrm{p} }( H^* ) \big )^\perp \subsetÊS( H ), \quad \big ( \mathcal{H}_{ \mathrm{b} }( H ) \oplus \mathcal{H}_{ \mathrm{p} }( H ) \big )^\perp \subset S( H^* ). \label{eq:inclS(H)}
\end{equation}

We prove the first inclusion in \eqref{eq:inclS(H)}. Recall that $S( H ) = \{ u \in \mathcal{H} , \int_0^\infty \| C e^{ i t H }Êu \|^2 dt < \infty \}$. Let $u \in ( \mathcal{H}_{ \mathrm{b} }( H ) \oplus \mathcal{H}_{ \mathrm{p} }( H^* ) )^\perp$. To prove that $u$ belongs to $S( H )$, we verify that
\begin{align}
\int_0^\infty \big \|ÊC e^{ i s ( H + i \varepsilon ) }Êu \big \|^2 ds &= \frac{1}{ 2 \pi } \int_{Ê\mathbb{R}Ê}Ê\big \|ÊC ( H - ( \lambda - i \varepsilon ) )^{-1}Êu \big \|^2 d \lambda , \label{eq:parseval_1}
\end{align}
for all $\varepsilon > 0$, and we show that the right side of \eqref{eq:parseval_1} is uniformly bounded in $\varepsilon > 0$, for $\varepsilon$ small enough. From Lebesgue's monotone convergence theorem we can then conclude that $\int_0^\infty \| C e^{ i t H }Êu \|^2 dt < \infty$.

We begin by justifying \eqref{eq:parseval_1}. Recall that $\Pi_{ \mathrm{pp} }$ denotes the sum of all Riesz projections associated to the isolated eigenvalues of $H$. Since $\mathrm{Ker} ( \Pi_{ \mathrm{pp} } ) = ( \mathcal{H}_{ \mathrm{b} }( H ) \oplus \mathcal{H}_{ \mathrm{p} }( H^* ) )^\perp$ (see \eqref{eq:rankerpipp}), the spectrum of the restriction of $H$ to the subspace $ ( \mathcal{H}_{ \mathrm{b} }( H ) \oplus \mathcal{H}_{ \mathrm{p} }( H^* )  )^\perp$ is contained in $\sigma_{ \mathrm{ess} } ( H ) = [ 0 , \infty )$. Moreover, we deduce from the resolvent equation
\begin{align*}
( H - z )^{-1} &= ( H_V - z )^{-1} + i ( H_V - z )^{-1} C^*C ( H_V - z )^{-1} \\
&\quad  -  ( H_V - z )^{-1} C^*C ( H - z )^{-1} C^*C ( H_V - z )^{-1} ,
\end{align*}
that, for all $z \in \{ z' \in \mathbb{C} , \mathrm{Im}( z' )< 0 \}$, 
\begin{equation}
\Big \|Ê( H - z )^{-1}Ê\big |_{ ( \mathcal{H}_{ \mathrm{b} }( H ) \oplus \mathcal{H}_{ \mathrm{p} }( H^* )  )^\perp } \Big \| \le \mathrm{c}  | \mathrm{Im}(z) |^{-2} , \label{eq:bound_resolv_restrict}
\end{equation}
for some positive constant $\mathrm{c}$. Here we used, in particular, that $H_V$ is self-adjoint and that $\| C ( H - z )^{-1} C^* \|$ is uniformly bounded in $z \in \{ z' \in \mathbb{C} , \mathrm{Re}( z' ) \ge 0 , - \varepsilon_0 < \mathrm{Im}( z' ) < 0 \}$, for some $\varepsilon_0 > 0$, as follows from the assumptions that $H$ has no spectral singularities in $[ 0 , \infty )$ and that \eqref{eq:sup>m} in Hypothesis \ref{V3} holds. The facts that the spectrum of the restriction of $H$ to the subspace $ ( \mathcal{H}_{ \mathrm{b} }( H ) \oplus \mathcal{H}_{ \mathrm{p} }( H^* )  )^\perp$ is contained in $\mathbb{R}$ and that \eqref{eq:bound_resolv_restrict} holds imply that there exists a positive constant $\mathrm{c}$ such that, for all $v \in ( \mathcal{H}_{ \mathrm{b} }( H ) \oplus \mathcal{H}_{ \mathrm{p} }( H^* )  )^\perp$ and $t \ge 0$, $\|Êe^{ i t H }Êv \|Ê\le \mathrm{c}( 1 + t^3 ) \| v \|$ (see \cite[Theorem 2.1]{EiZw07_01}). In particular, for all $\varepsilon > 0$, $s \mapsto \| e^{ - s \varepsilon } e^{ i s H }Êu \|^2$ is integrable on $[ 0 , \infty )$. By Parseval's theorem, this gives \eqref{eq:parseval_1}.

Next, we show that the right side of \eqref{eq:parseval_1} is uniformly bounded in $\varepsilon > 0$ small enough. By Hypothesis \ref{V-1},  $H_V$ has finitely many eigenvalues and all the eigenvalues of $H_V$ are negative. Let $J \subset ( - \infty , 0 )$ be a compact interval whose interior contains all the eigenvalues of $H_V$. We decompose the integral over $\mathbb{R}$ in the right side of \eqref{eq:parseval_1} into a sum of integrals over $J$ and $\mathbb{R} \setminus J$ and estimate each term separately. Since the spectrum of $H|_{  ( \mathcal{H}_{ \mathrm{b} }( H ) \oplus \mathcal{H}_{ \mathrm{p} }( H^* )  )^\perp }$ in a complex neighborhood of $J$ is empty, we have that
\begin{align}
\int_{ÊJÊ}Ê\big \|ÊC ( H - ( \lambda - i \varepsilon ) )^{-1}Êu \big \|^2 d \lambda \le \mathrm{c} \|Êu \|, \label{eq:abg1}
\end{align}
for some positive constant $\mathrm{c}$ independent of $\varepsilon$.

To estimate the integral over $\mathbb{R} \setminus J$, we use the resolvent equation, writing
\begin{align}
&\int_{Ê\mathbb{R} \setminus JÊ}Ê\big \|ÊC ( H - ( \lambda - i \varepsilon ) )^{-1}Êu \big \|^2 d \lambda \notag \\
& \le 2 \int_{Ê\mathbb{R} \setminus JÊ}Ê\big \|ÊC ( H_V - ( \lambda - i \varepsilon ) )^{-1}Êu \big \|^2 d \lambda \notag \\
&\quad + 2 \int_{Ê\mathbb{R} \setminus JÊ}Ê\big \|ÊC ( H - ( \lambda - i \varepsilon ) )^{-1}ÊC^* C ( H_V - ( \lambda - i \varepsilon ) )^{-1} u \big \|^2 d \lambda. \label{eq:abg2}
\end{align}
We decompose $u = u_{ \mathrm{pp}Ê} + u_{ \mathrm{ac}Ê}$ with $u_{Ê\mathrm{pp} } \in \mathcal{H}_{ \mathrm{pp} }( H_V )$ and $u_{Ê\mathrm{ac} } \in \mathcal{H}_{ \mathrm{ac}Ê}( H_V )$. Since $H$ has no spectral singularities in $[ 0 , \infty )$ and  \eqref{eq:sup>m} in Hypothesis \ref{V3} holds, there exists $\mathrm{k} > 0$ such that $\|ÊC ( H - ( \lambda - i \varepsilon ) )^{-1}ÊC^* \| \le \mathrm{k}$, uniformly in $\lambda \in \mathbb{R}$ and $\varepsilon > 0$ small enough. Therefore, since in addition the eigenvalues of $H_V$ are in the interior of $J$, we clearly have that
\begin{align}
& \int_{Ê\mathbb{R} \setminus JÊ}Ê\big \|ÊC ( H_V - ( \lambda - i \varepsilon ) )^{-1}Êu_{Ê\mathrm{pp} } \big \|^2 d \lambda \notag \\
& +  \int_{Ê\mathbb{R} \setminus JÊ}Ê\big \|ÊC ( H - ( \lambda - i \varepsilon ) )^{-1}ÊC^* C ( H_V - ( \lambda - i \varepsilon ) )^{-1} u_{Ê\mathrm{pp} } \big \|^2 d \lambda \le \mathrm{c} \| u_{Ê\mathrm{pp} } \|^2 , \label{eq:abg3}
\end{align}
for some positive constant $\mathrm{c}$ independent of $\varepsilon$. For the contribution given by $u_{Ê\mathrm{ac}Ê}$, we use Hypothesis \ref{V1}, which yields
\begin{align}
& \int_{Ê\mathbb{R} \setminus JÊ}Ê\big \|ÊC ( H_V - ( \lambda - i \varepsilon ) )^{-1}Êu_{Ê\mathrm{ac} } \big \|^2 d \lambda \notag \\
& + \int_{Ê\mathbb{R} \setminus JÊ}Ê\big \|ÊC ( H - ( \lambda - i \varepsilon ) )^{-1}ÊC^* C ( H_V - ( \lambda - i \varepsilon ) )^{-1} u_{Ê\mathrm{ac} } \big \|^2 d \lambda \le ( 1 + \mathrm{k} ) \mathrm{c}^2_V \| u_{Ê\mathrm{ac} } \|^2 , \label{eq:abg4}
\end{align}
where $\mathrm{c}_V$ is the constant in Hypothesis \ref{V1}. Equations \eqref{eq:abg2}, \eqref{eq:abg3} and \eqref{eq:abg4} show that
\begin{align*}
\int_{Ê\mathbb{R} \setminus JÊ}Ê\big \|ÊC ( H - ( \lambda - i \varepsilon ) )^{-1}Êu \big \|^2 d \lambda \le \mathrm{c} \|Êu \|, 
\end{align*}
for some positive constant $\mathrm{c}$ independent of $\varepsilon$. Together with \eqref{eq:abg1}, this gives
\begin{align*}
\int_{Ê\mathbb{R}Ê}Ê\big \|ÊC ( H - ( \lambda - i \varepsilon ) )^{-1}Êu \big \|^2 d \lambda \le \mathrm{c} \|Êu \| .
\end{align*}

Finally, the last equation combined with \eqref{eq:parseval_1} yields that, for all $u \in ( \mathcal{H}_{ \mathrm{b} }( H ) \oplus \mathcal{H}_{ \mathrm{p} }( H^* ) )^\perp$ and $\varepsilon > 0$ small enough, 
\begin{align*}
\int_0^\infty e^{ - \varepsilon s } \big \|ÊC e^{ i s  H }Êu \big \|^2 ds &\le \mathrm{c} \| u \|^2 ,
\end{align*}
for some positive constant $\mathrm{c}$ independent of $\varepsilon$. Lebesgue's monotone convergence theorem then proves that $u \in \mathcal{S}( H )$, which concludes the proof of the first inclusion in \eqref{eq:inclS(H)}. The proof of the second inclusion is analogous.

The facts that the restriction of $H$ to $( \mathcal{H}_{ \mathrm{b} }( H ) \oplus \mathcal{H}_{ \mathrm{p} }( H^* ) )^\perp$ and that the restriction of $H^*$ to $( \mathcal{H}_{ \mathrm{b} }( H ) \oplus \mathcal{H}_{ \mathrm{p} }( H ) )^\perp$ are similar to $H_0$ follow from the asymptotic completeness  of the wave operators (see Section \ref{subsec:compl}).
\end{proof}
We mention that an alternative proof of Theorem \ref{prop:AC} is based on the representation formula \eqref{eq:spectral_decomp_2}. More precisely, if Hypotheses \ref{V-1}--\ref{V3} hold and if $H$ does not have spectral singularities, i.e. $\nu_j = 0$ in \eqref{eq:cauchy-3}, $j = 1 , \dots , n$, then one can deduce from \eqref{eq:cauchy-3} that \eqref{eq:spectral_decomp_2} holds. Using that $(  \mathcal{H}_{ \mathrm{b} }( H ) \oplus \mathcal{H}_{ \mathrm{p} }( H^* ) )^\perp = \mathrm{Ker}( \Pi_{ \mathrm{pp} } )$ together with Equation \eqref{eq:spec-proj-exp} in Proposition \ref{prop:spectr-sing}, it is not difficult to verify that \eqref{eq:spectral_decomp_2} yields $(  \mathcal{H}_{ \mathrm{b} }( H ) \oplus \mathcal{H}_{ \mathrm{p} }( H^* ) )^\perp \subset S( H )$. Asymptotic completeness follows from this inclusion.

Theorem \ref{prop:AC} combined with Propositions  \ref{prop:existence_W+} and \ref{prop:invertscatt} has the following direct consequence, which proves Corollary \ref{cor:W+}.
\begin{corollary}
Suppose that Hypotheses \ref{V-1}--\ref{V3} hold and that $H$ has no spectral singularities in $[0,\infty)$. Then $W_+( H_0 , H ) : \mathcal{H} \to \mathcal{H}$ and $W_-( H_0 , H^* ): \mathcal{H} \to \mathcal{H}$ are surjective and
\begin{equation*}
\mathrm{Ker} ( W_+( H_0 , H ) ) = \mathcal{H}_{ \mathrm{b} }( H ) \oplus \mathcal{H}_{ \mathrm{p} }( H ), \quad \mathrm{Ker} ( W_-( H_0 , H^* ) ) = \mathcal{H}_{ \mathrm{b} }( H ) \oplus \mathcal{H}_{ \mathrm{p} }( H^* ). 
\end{equation*}
Moreover, $S( H, H_0 ): \mathcal{H} \to \mathcal{H}$ and $S( H^* , H_0 )$ are bijective.
\end{corollary}

Next, we prove Theorem \ref{thm:nonAC}. It justifies that the hypothesis used in Theorem \ref{prop:AC} according to which $H$ does not have any spectral singularities is essentially necessary for asymptotic completeness.
\begin{theorem}\label{prop:nonAC}
Suppose that Hypotheses \ref{V-1}--\ref{V3} hold. Assume that there exist an interval $J \subset [ 0 , \infty )$ and $u \in \mathcal{H}$ such that
\begin{align}
\lim_{ \varepsilon \downarrow 0 } \int_{ÊJÊ}Ê\big \|ÊC ( H - ( \lambda - i \varepsilon ) )^{-1}ÊC^* u \big \|^2 d \lambda = \infty. \label{eq:asspropnonAC}
\end{align}
Then $\mathrm{Ran}( W_-( H , H_0 ) )$ is not closed. In particular, $W_- (H ,H_0)$ is \emph{not} asymptotically complete. Likewise, if there exist an interval $J \subset [ 0 , \infty )$ and $u \in \mathcal{H}$ such that
\begin{align}
\lim_{ \varepsilon \downarrow 0 } \int_{ÊJÊ}Ê\big \|ÊC ( H^* - ( \lambda + i \varepsilon ) )^{-1}ÊC^* u \big \|^2 d \lambda = \infty , \label{eq:asspropnonAC_3}
\end{align}
then $\mathrm{Ran}( W_+( H^* , H_0 ) )$ is not closed and $W_+( H^* , H_0 )$ is not asymptotically complete.
\end{theorem}
\begin{proof}
We prove that $\mathrm{Ran} ( W_-( H , H_0 ) )$ is not closed. By Proposition \ref{prop:randense_W-}, it suffices to show that
\begin{align}
\mathrm{Ran}( W_-( H , H_0 ) ) \subsetneq \big ( \mathcal{H}_{ \mathrm{b} }( H ) \oplus \mathcal{H}_{ \mathrm{p} }( H^* ) \big )^\perp . \label{eq:strict_incl}
\end{align}
Recall that $\Pi_{ \mathrm{pp} }$ denote the sum of Riesz projections associated to all the isolated eigenvalues of $H$, and let $\overline{\Pi}_{ \mathrm{pp} } := \mathrm{Id}Ê- \Pi_{ \mathrm{pp} }$. Let $v := \overline{\Pi}_{ \mathrm{pp} } C^* u$. Then $v \in ( \mathcal{H}_{ \mathrm{b} }( H ) \oplus \mathcal{H}_{ \mathrm{p} }( H^* ) )^\perp$ (see \eqref{eq:rankerpipp}), and we claim that $v \notin \mathrm{Ran}( W_-( H , H_0 ) )$. Indeed, using Hypothesis \ref{V3} and applying \cite[Theorem 2.1]{EiZw07_01} as in the proof of Theorem \ref{prop:AC}, one verifies that, for all $\varepsilon > 0$, the map $s \mapsto \| e^{ - s \varepsilon } e^{ i s H }Êv \|^2$ is integrable on $[ 0 , \infty )$. Hence it follows from Parseval's theorem that
\begin{align}
\int_0^\infty \big \|ÊC e^{ i s ( H + i \varepsilon ) }Êv \big \|^2 ds &= \frac{1}{2\pi} \int_{Ê\mathbb{R}Ê}Ê\big \|ÊC ( H - ( \lambda - i \varepsilon ) )^{-1}Êv \big \|^2 d \lambda \notag \\
& \ge \frac{1}{2\pi} \int_{ÊJÊ}Ê\big \|ÊC ( H - ( \lambda - i \varepsilon ) )^{-1}Êv \big \|^2 d \lambda . \label{eq:slfg1}
\end{align}
The definition of $v$ implies that
\begin{align}
C ( H - ( \lambda - i \varepsilon ) )^{-1}Êv = C ( H - ( \lambda - i \varepsilon ) )^{-1}ÊC^* u - C ( H - ( \lambda - i \varepsilon ) )^{-1}Ê\Pi_{ \mathrm{pp} } C^* u. \label{eq:jskt1}
\end{align}
Since $\Pi_{ \mathrm{pp} }$ is a finite sum of Riesz projections corresponding to eigenvalues located in $\mathbb{C} \setminus [ 0 , \infty )$, we deduce that
\begin{align}
\lim_{ \varepsilon \downarrow 0 } \int_{ÊJÊ}Ê\big \|ÊC ( H - ( \lambda - i \varepsilon ) )^{-1}Ê\Pi_{ \mathrm{pp} } C^* u \big \|^2 d \lambda < \infty. \label{eq:jskt2}
\end{align}
Equations \eqref{eq:asspropnonAC}, \eqref{eq:jskt1} and \eqref{eq:jskt2} imply that
\begin{align*}
\lim_{ \varepsilon \downarrow 0 } \int_{ÊJÊ}Ê\big \|ÊC ( H - ( \lambda - i \varepsilon ) )^{-1}Êv \big \|^2 d \lambda = \infty .
\end{align*}
Hence \eqref{eq:slfg1} shows that $\int_0^\infty e^{ - \varepsilon s } \|ÊC e^{ i s  H  }Êv \|^2 ds \to \infty$, as $\varepsilon \to 0$, and therefore, by Lebesgue's monotone convergence theorem, $v \notin S( H ) = \{ w \in \mathcal{H} , \int_0^\infty  \|ÊC e^{ i s  H  }Êw \|^2 ds < \infty \}$. Since $\mathrm{Ran}( W_-( H , H_0 ) ) = S( H ) \cap \mathcal{H}_{ \mathrm{ac} }( H )$ by Theorem \ref{prop:range1}, we have thus proven that \eqref{eq:strict_incl} holds.

We can argue in the same way to show that $\mathrm{Ran}( W_+( H^* , H_0 ) ) \subsetneq ( \mathcal{H}_{ \mathrm{b} }( H ) \oplus \mathcal{H}_{ \mathrm{p} }( H ) )^\perp$ if \eqref{eq:asspropnonAC_3} holds. This concludes the proof.
\end{proof}

\subsection{Proof of Theorem \ref{thm:nonblowup}}\label{subsec:nonblowup}
As a consequence of Theorems \ref{prop:proj-ran2} and \ref{prop:AC}, we establish the following result, which, in particular, proves Theorem \ref{thm:nonblowup}.
\begin{theorem}\label{thm:nonblowup-2}
Suppose that Hypotheses \ref{V-1}--\ref{V3} hold and that $H$ has no spectral singularities in $[0,\infty)$. Then there exist $m_1>0$ and $m_2>0$ such that, for all $u \in \mathcal{H}_{ \mathrm{p} }( H^* )^\perp$,
\begin{equation}
m_1 \| u \| \le \big \| e^{ - i t H }Êu \big \|Ê\le m_2 \|Êu \|Ê, \quad t \in \mathbb{R}. \label{eq:nonblowup_1}
\end{equation}
Likewise, for all $u \in \mathcal{H}_{ \mathrm{p} }( H )^\perp$,
\begin{equation}
m_1 \| u \| \le \big \| e^{ i t H^* }Êu \big \|Ê\le m_2 \|Êu \|Ê, \quad t \in \mathbb{R}. \label{eq:nonblowup_2}
\end{equation}
\end{theorem}
\begin{proof}
We prove \eqref{eq:nonblowup_1}, the proof of \eqref{eq:nonblowup_2} is identical. From Theorem \ref{prop:AC} and the fact that $\mathcal{H}_{ \mathrm{b} }( H )$ and $\mathcal{H}_{ \mathrm{p} }( H^* )$ are orthogonal, we deduce that
\begin{equation*}
\mathcal{H}_{ \mathrm{p} }( H^* )^\perp = \mathrm{Ran}( W_-( H , H_0 ) ) \oplus \mathcal{H}_{ \mathrm{b} }( H ) ,
\end{equation*}
and we recall from \eqref{eq:HacHbperp} and \eqref{eq:1stincl} that $\mathrm{Ran}( W_-( H , H_0 ) )$ and $\mathcal{H}_{ \mathrm{b} }( H )$ are orthogonal.

Let $u = u_{ \mathrm{b} } + u_{ \mathrm{ac}Ê} \in \mathcal{H}_{ \mathrm{p} }( H^* )^\perp$ with $u_{ \mathrm{b} } \in \mathcal{H}_{ \mathrm{b} }( H )$ and $u_{ \mathrm{ac}Ê} \in \mathrm{Ran}( W_-( H , H_0 ) )$. Obviously, $\| e^{ - i t H }Êu_{ \mathrm{b} } \| = \| u_{ \mathrm{b} }\|$ for all $t \in \mathbb{R}$. Let $v \in \mathcal{H}$ be such that $u_{ \mathrm{ac} } = W_-( H , H_0 ) v$. By Proposition \ref{prop:existence_W-} and Theorem \ref{prop:AC}, $W_-( H , H_0 )$ is injective with closed range and therefore there exists $m_1 > 0$ such that $\| W_-( H , H_0 ) w \| \ge m_1 \|w \|$, for all $w \in \mathcal{H}$. Hence, using the intertwining property \eqref{eq:inter-1},
\begin{align*}
\big \| e^{ - i t H } u_{ \mathrm{ac} } \big \| &= \big \| e^{ - i t H }W_-( H , H_0 ) v \big \| \\
&= \big \|W_-( H , H_0 ) e^{ - i t H_0 } v \big \| \\
&\ge m_1 \big \|  e^{ - i t H_0 } v \big \| = m_1 \| v \| \ge m_1 \big \|  W_-( H , H_0 ) v \big \| = m_1 \| u_{ \mathrm{ac} } \| ,
\end{align*}
for all $t \in \mathbb{R}$, where in the second inequality we have used that $W_-( H , H_0 )$ is contractive. Together with the facts that $m_1 \le 1$ (because $W_-( H , H_0 )$ is contractive) and that $\| e^{ - i t H }Êu_{ \mathrm{b} } \| = \| u_{ \mathrm{b} }\|$, this proves the first inequality in \eqref{eq:nonblowup_1}.

To prove the second inequality in \eqref{eq:nonblowup_1}, we write 
\begin{align*}
\big \| e^{ - i t H } u_{ \mathrm{ac} } \big \| &= \big \|W_-( H , H_0 ) e^{ - i t H_0 } v \big \| \\
&\le \big \|  e^{ - i t H_0 } v \big \| = \| v \| \le m_1^{-1} \big \|  W_-( H , H_0 ) v \big \| = m_1^{-1} \| u_{ \mathrm{ac} } \| .
\end{align*}
Combined with the facts that $m_1 \le 1$ and that $\| e^{ - i t H }Êu_{ \mathrm{b} } \| = \| u_{ \mathrm{b} }\|$, this proves the second inequality in \eqref{eq:nonblowup_1} with $m_2 = m_1^{-1}$. Hence the theorem is proven.
\end{proof}

\subsection{Completeness of the local wave operators}\label{subsec:localwave}
We conclude this section by proving a result that completes the picture. We have seen that if $H$ has a spectral singularity (and if \eqref{eq:spectrsingenough} holds) then $W_-( H , H_0 )$ is not asymptotically complete. Nevertheless, in this case, we can establish that local wave operators (on intervals not containing any spectral singularities) are complete in an appropriate sense. 

We define the local wave operators on the interval $I$ by setting
\begin{equation*}
W_-( H , H_0 , I ) := \underset{t\to \infty }{\slim}  \, e^{ - i t H } e^{ i t H_0 } E_{ H_0 }Ê( I ) , \quad W_+( H^* , H_0 , I ) := \underset{t\to \infty }{\slim}  \, e^{ i t H^* } e^{ - i t H_0 }E_{ H_0 }Ê( I ) ,
\end{equation*}
where $E_{H_0}( I )$ denotes the usual spectral projection for the self-adjoint operator $H_0$. By Proposition \ref{prop:existence_W-}, it is clear that $W_-( H , H_0 , I)$ and $W_-( H^* , H_0 , I )$ exist and that their kernels equal $\mathrm{Ran}( E_{ H_0 }( I ) )^\perp$. If $H$ does not have spectral singularities in $I$, we can prove that the local wave operators are complete in the following sense:
\begin{theorem}\label{prop:loc-WO}
Suppose that Hypotheses \ref{V-1}, \ref{V0} and \ref{V1} hold. Let $I \subset [ 0 , \infty )$ be a compact interval containing no spectral singularities of $H$. Then the maps
\begin{align}
& W_-( H , H_0 , I ) : \mathrm{Ran}( E_{ H_0 }( I ) ) \to \mathrm{Ran} ( E_H ( I ) ) , \label{eq:loc-wave1} \\
& W_+( H^* , H_0 , I ) : \mathrm{Ran}( E_{ H_0 }( I ) ) \to \mathrm{Ran} ( E_{H^*} ( I ) ) , \label{eq:loc-wave2}
\end{align}
are bijective. 
\end{theorem}
\begin{remark}
Proceeding as in the proof of Theorem \ref{prop:range1}, it is not difficult to verify that the inverses of the maps \eqref{eq:loc-wave1}--\eqref{eq:loc-wave2} are given by
\begin{align*}
& W_-( H , H_0 , I )^{-1} = W_-( H_0 , H , I ) := \underset{t\to \infty }{\slim}  \, e^{ - i t H_0 } e^{ i t H } E_{ H }Ê( I ) ,\\
& W_+( H^* , H_0 , I )^{-1} = W_+( H_0 , H^* , I ) := \underset{t\to \infty }{\slim}  \, e^{ i t H_0 } e^{ i t H^* }E_{ H^* }Ê( I ) .
\end{align*}
\end{remark}
\begin{proof}[Proof of Theorem \ref{prop:loc-WO}]
By Proposition \ref{prop:existence_W-}, the restriction of $W_-( H , H_0 , I )$ to $\mathrm{Ran} ( E_{ H_0 }( I ) )$ is injective. Moreover, by the intertwining property \eqref{eq:inter-1}, combined with the fact that $E_H( I )$ is well-defined by Proposition \ref{prop:spectr-sing}, we have that
\begin{equation*}
W_-( H , H_0 , I ) = E_{H}( I ) W_- ( H , H_0 , I ) ,
\end{equation*}
which shows that $\mathrm{Ran}( W_-( H , H_0 , I ) ) \subset \mathrm{Ran}( E_H( I ) )$. To prove the converse inclusion, it suffices to use that $E_H( I )$ is a projection (see Proposition \ref{prop:spectr-sing}) and apply Theorem \ref{prop:proj-ran}. For if $v \in \mathrm{Ran}( E_H( I ) )$, there exists $u \in \mathcal{H}$ such that $v = W_- ( H , H_0 ) u$. Hence 
\begin{equation*}
v = E_H( I ) v = E_H ( I ) W_-( H , H_0 ) u = W_- ( H , H_0 ) E_{ÊH_0 }( I ) u ,
\end{equation*}
where we used again the intertwining property \eqref{eq:inter-1} in the last equality. This shows that 
\begin{equation*}
\mathrm{Ran}( E_H( I ) ) \subset \mathrm{Ran} ( W_- ( H , H_0 , I ) ) ,
\end{equation*}
and concludes the proof.
\end{proof}
In a stationary approach to scattering theory, the existence and completeness of local wave operators has been proven for various non-self-adjoint operators $H$ satisfying appropriate conditions related to ours; see \cite{Mo67_01,Mo68_01,Go70_01,Go71_01,Hu71_01}.

\section{Application to Schr{\"o}dinger operators}\label{sec:schr}

In this section we apply our abstract results to operators of the form \eqref{eq:intro1}. Therefore, throughout the section, we set
\begin{equation}
H = - \Delta + V( x ) - i W( x ) , \qquad H_V = - \Delta + V(x) , \qquad H_0 = - \Delta , \label{eq:Hschr}
\end{equation}
on $L^2( \mathbb{R}^3 )$, with $V,W: \mathbb{R}^3 \to \mathbb{R}$ and $W \ge 0$. Since $W \ge 0$ we can write $W = C^* C$ with $C = \sqrt{W}$ and hence $H$ is of the form \eqref{eq:exprH}. We do not aim at finding optimal conditions on the potentials $V$ and $W$ such that our abstract hypotheses are satisfied. Rather, we verify that our results can be applied under the assumption that $V$ and $W$ are bounded and compactly supported. Extensions of our results to more general potentials are left to future work. Thus, we assume that
\begin{equation}
V, W \in L^\infty_{ \mathrm{c} }( \mathbb{R}^3 ; \mathbb{R} ) := \{ u \in L^\infty( \mathbb{R}^3 ; \mathbb{R} ) , \, u \text{ is compactly supported} \} , \quad W \ge 0 . \label{eq:VW_Linftycomp}
\end{equation}
This implies that $V$ and $W$ are relatively compact with respect to $H_0$. In addition, we require that
\begin{equation}
0 \text{ is neither an eigenvalue nor a resonance of } H_V. \label{eq:0noteignorres}
\end{equation}
We say that $0$ is a resonance of $H_V$ if the equation $H_V u = 0$ has a solution $u \in H^2_{ \mathrm{loc} }( \mathbb{R}^3 ) \setminus L^2 ( \mathbb{R}^3 )$. Equivalently, $0$ is a resonance of $H_V$ if it is a pole of the meromorphic extension of the resolvent of $H_V$ in a sense described more precisely below (see the verification of Hypothesis \ref{V3}).

In what follows, we verify that, assuming \eqref{eq:VW_Linftycomp} and \eqref{eq:0noteignorres}, the operators $H_0$, $H_V$ and $H$ in \eqref{eq:Hschr} satisfy Hypotheses \ref{V-1}--\ref{V3}.

\vspace{0,2cm}

\noindent \textbf{Verification of Hypothesis \ref{V-1}.} By Fourier transformation, it is clear that the spectrum of $H_0 = - \Delta$ is purely absolutely continuous. The facts that, for $V$ bounded and compactly supported, the singular continuous spectrum of $H_V = - \Delta + V(x)$ is empty, that $H_V$ does not have positive eigenvalues, and that $H_V$ has only finitely many non-positive eigenvalues are well-known (see, e.g., Theorems XIII.6, XIII.21 and XIII.57 in \cite{ReedSimon4}). To make sure that Hypothesis \ref{V-1} is satisfied, it therefore suffices to assume, in addition to \eqref{eq:VW_Linftycomp}, that $0$ is not an eigenvalue of $H_V$. This is assumed in \eqref{eq:0noteignorres}.

\vspace{0,2cm}

\noindent \textbf{Verification of Hypothesis \ref{V2}.} Assuming \eqref{eq:VW_Linftycomp}, it is known that $H$ has only finitely many eigenvalues with finite algebraic multiplicities. See, e.g., \cite{FrLaSa16_01} and references therein. In particular, Hypothesis \ref{V2} is satisfied.

\vspace{0,2cm}

\noindent \textbf{Verification of Hypothesis \ref{V0}.} Assuming that $V$ is bounded and compactly supported, the existence and completeness of the wave operators $W_\pm( H_V , H_0 )$ and $W_{Ê\pm } ( H_0 , H_V )$ follow from well-known arguments. See, e.g., \cite{RS-III} or \cite{Ya92_01}.

\vspace{0,2cm}

\noindent \textbf{Verification of Hypothesis \ref{V1}.} To verify Hypothesis \ref{V1}, one can apply a result proven in \cite{BeMa92_01} (see also \cite{CoSa89_01}): Assuming \eqref{eq:0noteignorres} and that $V \in L^\infty_{ \mathrm{c} }( \mathbb{R}^3 ; \mathbb{R} )$, the inequality 
\begin{equation*}
\int_{ \mathbb{R}Ê} \big \| ( 1 + x^2 )^{-\frac{1+\varepsilon}{2}} e^{ - i t H_V } \Pi_{ \mathrm{ac} }( H_V ) u  \big \|^2 dt \le \mathrm{c}_\varepsilon^2Ê\| \Pi_{ \mathrm{ac} }( H_V ) u \|^2 
\end{equation*}
holds for all $\varepsilon > 0$, $u \in L^2( \mathbb{R}^3 )$ and some positive constant $\mathrm{c}_\varepsilon$. Since $C = \sqrt{W}$ is compactly supported, this implies that Hypothesis \ref{V1} is satisfied.

\vspace{0,2cm}

\noindent \textbf{Verification of Hypothesis \ref{V3}.} Hypothesis \ref{V3} is related to the theory of resonances for Schr{\"o}dinger operators. See, e.g., \cite{DyZw17_01}. Since $V$ and $W$ belong to $L^\infty_{ \mathrm{c} }( \mathbb{R}^3 )$, the map 
\begin{equation*}
\{ z \in \mathbb{C} , \mathrm{Im}( z ) > 0 \} \ni z \mapsto ( H - z^2 )^{-1} : L^2( \mathbb{R}^3 ) \to L^2( \mathbb{R}^3 ) 
\end{equation*}
is meromorphic and extends to a meromorphic map
\begin{equation}
\mathbb{C} \ni z \mapsto R(z^2) : L^2_{ \mathrm{c} }( \mathbb{R}^3 ) \to L^2_{ \mathrm{loc} }( \mathbb{R}^3 ) , \label{eq:mero_ext}
\end{equation}
where, we recall, $L^2_{ \mathrm{c}Ê}( \mathbb{R}^3 )= \{ u \in L^2( \mathbb{R}^3 ) , u \text{ is compactly supported} \}$ and $L^2_{Ê\mathrm{loc} }( \mathbb{R}^3 ) = \{ u : \mathbb{R}^3 \to \mathbb{C}, u \in L^2( K ) \text{ for all compact set } K \subset \mathbb{R}^3 \}$. Here, a map $\Omega \ni z \mapsto A(z) \in \mathcal{L}( E ; F )$, where $\Omega \subset \mathbb{C}$ is an open set and $E$, $F$ are Banach spaces, is called meromorphic in $\Omega$ if, for all $z_0 \in \Omega$, there exist a finite number of finite-rank operators $A_1, \dots A_n$, and a holomorphic family of operators $z \mapsto A_0(z)$, such that
\begin{equation*}
A(z) = A_0( z ) + \sum_{j=1}^n \frac{ A_j }{ ( z - z_0 )^j } ,
\end{equation*}
in a complex neighborhood of $z_0$. A map $z \mapsto A(z) : L^2_{ \mathrm{c} }( \mathbb{R}^3 ) \to L^2_{ \mathrm{loc} }( \mathbb{R}^3 )$ is called meromorphic if, for arbitrary bounded, compactly supported potentials $\rho_1$, $\rho_2$, the map $z \mapsto \rho_1 A( z ) \rho_2 : L^2( \mathbb{R}^3 ) \to L^2( \mathbb{R}^3 )$ is meromorphic. Resonances of $H$ are poles of the meromorphic extension \eqref{eq:mero_ext}. Hence, since $C = \sqrt{W}$ is compactly supported, 
\begin{equation}
z \mapsto C ( H - z^2 )^{-1} C^* : L^2( \mathbb{R}^3 ) \to L^2( \mathbb{R}^3 ) , \label{eq:mero_ext2}
\end{equation}
is meromorphic in $\{ z \in \mathbb{C} , \mathrm{Im}( z ) > 0 \}$ and extends to a meromorphic map in $\mathbb{C}$. Each spectral singularity $z_0^2 \in [ 0 , \infty )$ (with $z_0 \ge 0$), in the sense of Definition \ref{def:spec-sing} in Section \ref{subsec:hypoth}, corresponds to a resonance $-z_0 \in ( - \infty , 0 ]$, i.e., a pole of \eqref{eq:mero_ext2}. This implies that, for each spectral singularity $\lambda \in [ 0 , \infty )$, there exist an integer $\nu > 0$ and a compact interval $K_{\lambda}$, whose interior contains $\lambda$, such that the limit
\begin{equation*}
\lim_{ \varepsilon \downarrow 0 }  | \mu - \lambda |^{\nu} \big \|ÊC \big ( H - ( \mu - i \varepsilon ) \big )^{-1} C^* \big \| 
\end{equation*}
exists uniformly in $\mu \in K_{\lambda}$ in the norm topology of $\mathcal{L}( L^2( \mathbb{R}^3 ) )$.

Moreover, it is well-known that
\begin{equation*}
\big \| \rho_1 ( - \Delta - \mu - i \varepsilon )^{-1} \rho_2 \big \| = \mathcal{O} ( \mu^{-\frac12} ) , \quad \mu \to \infty ,
\end{equation*}
uniformly in $\varepsilon > 0$, for arbitrary bounded, compactly supported potentials $\rho_1$, $\rho_2$; see, e.g., \cite{DyZw17_01}. Since $V$ and $W$ are compactly supported, it then easily follows from a Neumann series expansion that there exists $m > 0$ such that
\begin{equation*}
\sup_{Ê\mu \ge m, \, \varepsilon > 0 } \big \| C \big ( H - ( \mu - i \varepsilon ) \big )^{-1} C^* \big \| < \infty.
\end{equation*}
In particular $H$ has no resonances in $( - \infty , - m^{1/2} ]$. Since, in addition, it is known that there are only finitely many resonances in $[ - m^{1/2} , 0 ]$, for any $m > 0$ (see \cite{DyZw17_01}), we can conclude that $H$ has finitely many spectral singularities, and hence Hypothesis \ref{V3} is indeed satisfied.

\vspace{0,2cm}

\noindent \textbf{Verification of \eqref{eq:spectrsingenough}.} Finally, we verify that, if $H$ has a spectral singularity, then \eqref{eq:spectrsingenough} is necessarily satisfied for operators $H$ of the form \eqref{eq:Hschr} satisfying \eqref{eq:VW_Linftycomp}. Indeed, suppose that $z_0^2 \in [ 0 , \infty )$ (with $z_0 \ge 0$) is a spectral singularity of $H$. Then $-z_0$ is a resonance of $H$, in the sense specified above, so that there exists a complex neighborhood $\Omega_{z_0}$ of $z_0$ such that
\begin{equation*}
C ( H - z^2 )^{-1} C^* = A_0( z ) + \sum_{j=1}^n \frac{ A_j }{ ( z + z_0 )^j } ,
\end{equation*}
for all $z \in \Omega_{z_0}$, $\mathrm{Im}(z) > 0$, where $z \mapsto A_0( z )$ is holomorphic in $\Omega_{z_0}$ and $A_j$, $j = 1 , \dots , n $ are finite-rank operators, with $A_n \neq 0$. Picking $z_{\lambda,\varepsilon} \in \Omega_{z_0}$, $\mathrm{Im}( z_{\lambda,\varepsilon} ) > 0$, such that $z_{\lambda,\varepsilon}^2 = \lambda - i \varepsilon$, for $\lambda$ in a real neighborhood $J_{z_0}$ of $z_0^2$ and $\varepsilon > 0$ small enough, we deduce that
\begin{equation*}
C ( H - ( \lambda - i \varepsilon ) )^{-1} C^* = A_0( z_{\lambda,\varepsilon} ) + \sum_{j=1}^n \frac{ A_j }{ ( z_{\lambda,\varepsilon} + z_0 )^j } .
\end{equation*}
For $u \in \mathcal{H}$ such that $A_n u \neq 0$, it is then not difficult to deduce that
\begin{align*}
\lim_{ \varepsilon \downarrow 0 } \int_{ÊJ_{z_0}Ê}Ê\big \|ÊC ( H - ( \lambda - i \varepsilon ) )^{-1}ÊC^* u \big \|^2 d \lambda = \infty ,
\end{align*}
which proves that \eqref{eq:spectrsingenough} holds.

\vspace{0,2cm}

Recall that the subspaces $\mathcal{H}_{ \mathrm{b} }( H )$, $\mathcal{H}_{ \mathrm{p} }( H )$, $\mathcal{H}_{ \mathrm{d} }( H )$ and $\mathcal{H}_{ \mathrm{p} }( H^* )$ are defined in Section \ref{subsec:subspaces}. For dissipative Schr{\"o}dinger operators \eqref{eq:Hschr}, it follows from Lemma \ref{lm:Hb} and the unique continuation principle (see, e.g., \cite[Theorem XIII.63]{ReedSimon4}) that, if $W( x ) > 0$ on some non-trivial open set, then $\mathcal{H}_{ \mathrm{b} }( H ) = \{ 0 \}$, i.e., $H$ does not have real eigenvalues.

Applying the results of Section \ref{subsec:main_res}, we obtain the following result.
\begin{theorem}
Let $H = - \Delta + V( x ) - i W( x )$ on $L^2( \mathbb{R}^3 )$ with $W \ge 0$, $W( x ) > 0$ on some non-trivial open set and $V,W \in L^\infty_{ \mathrm{c} }( \mathbb{R}^3 ; \mathbb{R} )$. Suppose that $0$ is neither an eigenvalue nor a resonance of $H_V = - \Delta + V(x)$. Then
\begin{equation*}
\mathcal{H}_{ \mathrm{p} }( H ) = \mathcal{H}_{ \mathrm{d} }( H ) .
\end{equation*}
Moreover, the wave operator $W_-( H , H_0 ) = \slim_{ t \to \infty } e^{ - i t H }Êe^{ i t H_0 }$, with $H_0 = - \Delta$, is asymptotically complete in the sense that
\begin{equation*}
\mathrm{Ran}( W_-( H , H_0 ) ) = \mathcal{H}_{ \mathrm{p} }( H^* )^\perp 
\end{equation*}
\emph{if and only if} $H$ does not have real resonances. In this case, the restriction of $H$ to $\mathcal{H}_{ \mathrm{p} }( H^* )^\perp $ is similar to $H_0$ and there exist $m_1>0$ and $m_2>0$ such that, for all $u \in \mathcal{H}_{ \mathrm{p} }( H^* )^\perp$,
\begin{equation*}
m_1 \| u \| \le \big \| e^{ - i t H }Êu \big \|Ê\le m_2 \|Êu \|Ê, \quad t \in \mathbb{R}.
\end{equation*}
\end{theorem}
It is proven in \cite{Wang2} that $0$ cannot be a resonance of $H$. Moreover, since $H$ is dissipative, it does not have positive resonances. It is likely that, generically, $H$ does not have real resonances, implying that the wave operator $W_-( H , H_0 )$ is generically asymptotically complete. However, for any $z_0>0$, it is not difficult to construct smooth compactly supported potentials $V$ and $W$ such that $-z_0$ is a resonance of $H = - \Delta + V(x) - i W(x)$ (see \cite{Wang1}). Our results clearly underline the importance of real resonances in the scattering theory of dissipative Schr{\"o}dinger operators.

\section{Scattering theory for Lindblad master equations}\label{sec:Lindblad}

In this section we outline some consequences of our results for the scattering theory of Lindblad master equations. We refer the reader to \cite{Davies2} and \cite{FaFaFrSc17_01} for more details and references on this subject.

Let $\mathcal{H}$ be a complex separable Hilbert space. On the space of trace-class operators $\mathcal{J}_1( \mathcal{H} )$, we consider a Lindbladian of the form
\begin{equation*}
\mathcal{L}( \rho ) = H \rho - \rho H^* + i \sum_{j \in \mathbb{N} } W_j \rho W_j^*, \qquad H = H_V - \frac{i}{2} \sum_{j \in \mathbb{N} } W^*_j W_j ,
\end{equation*}
for all $\rho \in \mathcal{J}_1( \mathcal{H} )$ where, as above, $H_V = H_0 + V$ is a self-adjoint operator on $\mathcal{H}$, with $H_0$ self-adjoint and $V$ symmetric and relatively compact with respect to $H_0$. We suppose that, for all $j \in \mathbb{N}$, $W_j \in \mathcal{L}( \mathcal{H} )$, and  $\sum_{ j \in \mathbb{N} } W_j^* W_j \in \mathcal{L}( \mathcal{H} )$. Since $\frac12 \sum_{ j \in \mathbb{N} } W_j^* W_j$ is non-negative, its square root, denoted by $C \in \mathcal{L}( \mathcal{H} )$, satisfies
\begin{equation*}
C^* C := \frac12 \sum_{ j \in \mathbb{N} } W_j^* W_j .
\end{equation*}
Hence, in particular, $H$ is a dissipative operator on $\mathcal{H}$ of the form \eqref{eq:exprH}. We suppose that $C$ is relatively compact with respect to $H_0$.

The domain of the unbounded operator $\mathcal{L}$ is defined by
  \begin{align*}
  \mathcal{D}(\mathcal{L})= \big \{ & \rho \in \mathcal{J}_{1}(\mathcal{H}) , \, \rho ( \mathcal{D}( H_0 ) ) \subset \mathcal{D}( H_0 )  \text{ and } \\
  &H_0 \rho - \rho H_0, \text{ defined on } \mathcal{D}( H_0 ), \text{ extends to an element of } \mathcal{J}_{1}(\mathcal{H}) \big \}.
  \end{align*}
It is known (see \cite[Theorem 5.2]{Daviesbook}) that $\mathcal{L}$ is the generator of a quantum dynamical semigroup $\{ e^{ - i t \mathcal{L} } \}_{ t \ge 0 }$, i.e., a strongly continuous one-parameter semigroup on $\mathcal{J}_1( \mathcal{H} )$ such that, for all $t \ge 0$, $e^{ - i t \mathcal{L} }$ preserves the trace and is a completely positive operator. As mentioned in the introduction, if one considers a quantum particle interacting with a dynamical target, takes the trace over the degrees of freedom of the target and studies the reduced effective evolution of the particle, then, in the kinetic limit, the dynamics of the particle is given by a quantum dynamical semigroup of the form $\{ e^{ - i t \mathcal{L} } \}_{ t \ge 0 }$. The free dynamics of the particle is supposed to be given by the group of isometries $\{ e^{ - i t \mathcal{L}_0 } \}_{ t \in \mathbb{R} }$, with generator
\begin{equation*}
\mathcal{L}_0( \rho ) = H_0 \rho - \rho H_0 ,
\end{equation*}
for all $\rho \in \mathcal{D}( \mathcal{L}_0 ) \subset \mathcal{J}_1( \mathcal{H} )$. The domain, $\mathcal{D}( \mathcal{L}_0 )$, of $\mathcal{L}_0$ coincides with $\mathcal{D}( \mathcal{L} )$.

Let $\Pi_{ \mathrm{pp} }^\perp : \mathcal{H} \to \mathcal{H}$ denote the orthogonal projection onto $(\mathcal{H}_{ \mathrm{b} }( H ) \oplus \mathcal{H}_{ \mathrm{p} }( H ))^\perp$. A modified wave operator $\tilde{\Omega}_+( \mathcal{L}_0 , \mathcal{L} )$ is defined by
\begin{equation*}
\tilde{\Omega}_+( \mathcal{L}_0 , \mathcal{L} ) := \underset{t\to +\infty }{\slim} \: e^{it \mathcal{L}_0}  \big ( \Pi_{ \mathrm{pp} }^\perp e^{-it\mathcal{L}}  ( \cdot ) \Pi_{ \mathrm{pp} }^\perp \big ) .
\end{equation*}
In \cite{Davies2}, the projection onto $(\mathcal{H}_{ \mathrm{b} }( H ) \oplus \mathcal{H}_{ \mathrm{d} }( H ))^\perp$ is considered instead of $\Pi_{ \mathrm{pp} }^\perp$; but we know from Theorem \ref{thm:Hp-Hd} that these two projections coincide under our assumptions. It is proven in \cite[Theorem 4]{Davies2} (see also \cite{FaFaFrSc17_01}) that the asymptotic completeness of the wave operator $W_- (H , H_0)$ implies the existence of $\tilde{\Omega}_+( \mathcal{L}_0 , \mathcal{L} )$.  Therefore, as a consequence of Theorems \ref{thm:Hp-Hd} and  \ref{thm:AC}, we obtain the following result.
\begin{theorem}\label{thm:lindblad}
Suppose that Hypotheses \ref{V-1}--\ref{V3} hold and that $H$ has no spectral singularities in $[0,\infty)$. Then $\tilde{\Omega}_+( \mathcal{L}_0 , \mathcal{L} )$ exists on $\mathcal{J}_1( \mathcal{H} )$.
\end{theorem}
For all  $\rho \in \mathcal{J}_1( \mathcal{H} )$ with $\rho \ge 0$ and $\mathrm{tr}( \rho ) = 1$, the number $\mathrm{tr}( \tilde{\Omega}_+( \mathcal{L}_0 , \mathcal{L} ) \rho ) \in [ 0 , 1 ]$ is interpreted as the probability that the particle, initially in the state $\rho$, eventually escapes from the target. This quantity is therefore well-defined under our assumptions.

\appendix

\section{Proof of Lemmas \ref{lm:Hb} and \ref{lm:real-generalized}}\label{app:spectrum}

In this section we prove Lemmas \ref{lm:Hb} and \ref{lm:real-generalized} stated in Section \ref{subsec:spectrum}.
\begin{proof}[Proof of Lemma \ref{lm:Hb}]
We prove that 
\begin{equation}
\mathcal{H}_{ \mathrm{b} }( H ) = \mathrm{Span} \big \{ u \in \mathcal{D}( H ) , \, \exists \lambda \in \mathbb{R} , \, H u = \lambda u \big \} \subset \mathcal{H}_{ \mathrm{pp} }( H_V ) \cap \mathrm{Ker}( C ). \label{eq:Span}
\end{equation}
Let $u \in \mathcal{D}( H )$ be such that $H u = \lambda u$ with $\lambda \in \mathbb{R}$. Then
\begin{equation*}
\lambda \| u \|^2 = \langle u , H u \rangle = \langle u , H_V u \rangle - i \| C u \|^2 .
\end{equation*}
Identifying the imaginary parts, this implies that $u=0$ or $u \in \mathrm{Ker}( C )$. If $u \in \mathrm{Ker}( C )$ then $\lambda u = H u = H_V u$ and therefore $u \in \mathcal{H}_{ \mathrm{pp} }( H_V )$. Since $\mathcal{H}_{ \mathrm{pp} }( H_V )$ and $\mathrm{Ker}( C )$ are vector spaces, this establishes \eqref{eq:Span}. 

Likewise, if $\lambda u = H u$ and $u \in \mathrm{Ker}( C )$ then $\lambda u = H^* u$. This yields $\mathcal{H}_{ \mathrm{b} }( H ) = \mathcal{H}_{ \mathrm{b} }( H^* )$.
\end{proof}
\begin{proof}[Proof of Lemma \ref{lm:real-generalized}]
Let $u$ be as in the statement of the lemma. Using that $\lambda \in \mathbb{R}$, we write, for $t \ge 0$,
\begin{align*}
 \|Êe^{ - i t H }Êu \| = \| e^{ - i t ( H - \lambda ) } u \| = \Big \|Ê\sum_{j = 0}^{k-1} \frac{ (-it)^j }{ j! } ( H - \lambda )^j u \Big \|.
\end{align*}
If $k > 1$, this implies that
\begin{align*}
\| e^{ - i t H }Êu \|Ê\ge \frac{Êt^{k-1} }{ (k-1)! } \big \|Ê( H - \lambda )^{k-1} u \big \| - \mathcal{O}( t^{k-2} ) , \quad t \to \infty ,
\end{align*}
which is in contradiction with the fact that $e^{ - i t H }$ is contractive.
\end{proof}

\section{Existence and properties of the wave operators}\label{app:existencewave}

In this section we prove Propositions \ref{prop:existence_W-}--\ref{prop:invertscatt} on the existence and properties of the wave operators. We begin by proving Proposition \ref{prop:existence_W-}.

\begin{proof}[Proof of Proposition \ref{prop:existence_W-}]
We establish the proposition for $W_-( H , H_0 )$, the proof is the same for $W_+( H^* , H_0 )$. Note that most of the arguments employed here are already present in \cite{FaFaFrSc17_01}.

We establish the existence of $W_-( H , H_0 )$. By Hypothesis \ref{V0}, we know that $W_-( H_V , H_0 ) = \slim_{ t \to \infty } e^{ - i t H_V} e^{ i t H_0 }$ exists. Hence, since $e^{ - i t H }Êe^{ i t H_0 }$ is uniformly bounded in $t \ge 0$ (because $e^{-itH}$ is contractive and $e^{ i t H_0}$ is unitary),
\begin{equation*}
e^{ - i t H }Êe^{ i t H_0 } u = e^{ - i t H } e^{ i t H_V } W_-( H_V , H_0 ) u + o(1) , \quad t \to \infty ,
\end{equation*}
for all $u \in \mathcal{H}$. Moreover, $W_-( H_V , H_0 )$ maps $\mathcal{H}$ to $\mathrm{Ran} ( \Pi_{ \mathrm{ac} }( H_V ) )$ by Hypothesis \ref{V0}. Therefore, to prove that $W_-( H , H_0 )$ exists, it suffices to verify that 
\begin{equation}
W_-( H , H_V ) u = \underset{t\to  \infty }{\lim} e^{ - i t H } e^{ i t H_V } u \label{eq:exis_limit_WHHV}
\end{equation}
exists for all $u \in \mathrm{Ran}( \Pi_{ \mathrm{ac} }( H_V ) )$. To this end we use Cook's argument: We write, 
\begin{equation}\label{eq:cook}
e^{ - i t H } e^{ i t H_V } u = u -  \int_0^t e^{Ê- i s H }ÊC^* C e^{ i s H_V }Êu ds.
\end{equation}
For $0 < t_1 < t_2 < \infty$, we have that
\begin{align*}
 \Big \| \int_{t_1}^{t_2} e^{Ê- i s H }ÊC^* C e^{Êi s H_V }Êu ds \Big \| 
&\le \sup_{Êv \in \mathcal{H} , \| v \|Ê= 1 } \int_{t_1}^{t_2} \big | \big \langle C e^{ i s H^* }Êv , C e^{Êi s H_V }Êu \big \rangle \big | ds \\
&\le \sup_{Êv \in \mathcal{H} , \| v \|Ê= 1 } \Big ( \int_{t_1}^{t_2} \big \| C e^{ i s H^* }Êv \|^2 ds \Big )^{\frac12} \Big ( \int_{t_1}^{t_2} \| C e^{ i s H_V }Êu \|^2 ds \Big )^{ \frac12 } \\
& \le \frac12 \Big ( \int_{t_1}^{t_2} \| C e^{ i s H_V }Êu \|^2 ds \Big )^{ \frac12 } ,
\end{align*}
where we used \eqref{eq:alpha2} in the last inequality. Since, for all $u \in \mathrm{Ran}( \Pi_{ \mathrm{ac} }( H_V ) )$,  $s \mapsto \| C e^{ i s H_V }Êu \|^2$ is integrable on $\mathbb{R}$ by Hypothesis \ref{V1}, this implies that $\big ( \int_0^{t_n} e^{Êi s H_0 }ÊC^* C e^{Ê- i s H }Êu \, ds \big )_{ n \in \mathbb{N} }$ is a Cauchy sequence, for any sequence $(t_n)$ with $t_n \to \infty$. Hence the limit in \eqref{eq:exis_limit_WHHV} exists, and therefore $W_-( H , H_0 )$ exists.

Next, we prove that $W_-( H , H_0 )$ is injective. Since 
\begin{equation*}
W_- ( H , H_0 ) = W_- ( H , H_V ) W_-( H_V , H_0 )
\end{equation*}
from the above, and since $W_-( H_V , H_0 )$ is bijective from $\mathcal{H}$ to $\mathrm{Ran}( \Pi_{ \mathrm{ac}Ê}( H_V ) )$ by Hypothesis \ref{V0}, it suffices to prove that $W_- ( H , H_V )$ restricted to $\mathrm{Ran}( \Pi_{ \mathrm{ac}Ê}( H_V ) )$ is injective. We claim that
\begin{equation}\label{eq:limit_injective}
\lim_{t \to \infty} \big \| W_-( H , H_V ) e^{ i t H_V } u \big \| = \| u \| , 
\end{equation}
for all $u \in \mathrm{Ran}( \Pi_{ \mathrm{ac}Ê}( H_V ) )$. Indeed, from \eqref{eq:cook}, we obtain that
\begin{equation*}
W_-( H , H_V ) u = u - \int_0^\infty e^{Ê- i s H }ÊC^* C e^{ i s H_V }Êu ds ,
\end{equation*}
for all $u \in \mathrm{Ran}( \Pi_{ \mathrm{ac}Ê}( H_V ) )$. Applying this equality to $u = e^{ i t H_V } v$, with $v \in \mathrm{Ran}( \Pi_{ \mathrm{ac}Ê}( H_V ) )$, and changing variables, this gives
\begin{align*}
ÊW_-( H , H_V ) e^{ i t H_V } v = Êe^{ i t H_V } v - \int_t^\infty e^{ - i ( s - t ) H } C^* C e^{ i s H_V } v ds .
\end{align*}
Proceeding as above, we have that
\begin{align*}
\Big \|Ê\int_t^\infty e^{ - i ( s - t ) H } C^* C e^{ i s H_V } v ds \Big \|Ê&\le \sup_{ w \in \mathcal{H} , \| w \|Ê= 1Ê} \Big (Ê\int_0^\infty \big \| C e^{ i s H^* } w \big \|^2 ds \Big )^{\frac12} \Big ( \int_t^\infty \big \| C e^{ i s H_V } v \big \|^2 ds \Big )^{\frac12}   \\
&\le \frac12 \Big ( \int_t^\infty \big \| C e^{ i s H_V } v \big \|^2 ds \Big )^{\frac12} \to 0 , \quad t \to \infty ,
\end{align*}
where we used \eqref{eq:alpha2} in the last inequality, and Hypothesis \ref{V1} to justify that the limit vanishes. This proves \eqref{eq:limit_injective}. To conclude that $W_-( H , H_V )$ is injective, let $u \in \mathrm{Ran}( \Pi_{ \mathrm{ac}Ê}( H_V ) )$ be such that $W_-( H , H_V ) u = 0$. Then the usual intertwining property gives
\begin{equation*}
e^{ i t H } W_- ( H , H_V ) u = W_-( H , H_V ) e^{ i t H_V } u = 0 ,
\end{equation*}
for all $t \ge 0$. Letting $t \to \infty$ shows that $u = 0$ by \eqref{eq:limit_injective}.

The fact that $W_-( H , H_0 )$ is a contraction is a direct consequence of the contractivity of $\{ e^{ - i t H }Ê\}_{ t \ge 0 }$ and unitarity of $\{ e^{ - i t H_0 }Ê\}_{ t \in \mathbb{R} }$.

Finally, the intertwining properties \eqref{eq:inter-1}--\eqref{eq:inter-2} follow from standard arguments (see, e.g., \cite{RS-III}).
\end{proof}
Before proving Proposition \ref{prop:randense_W-}, we establish Proposition \ref{prop:existence_W+}.
\begin{proof}[Proof of Proposition \ref{prop:existence_W+}]
We establish the existence of $W_+ ( H_0 , H )$. The proof of the existence of $W_-( H_0 , H^* )$ is identical. Recall that $\Pi_{Ê\mathrm{pp} }( H_V )$ and $\Pi_{ \mathrm{ac} } ( H )$ denote the orthogonal projections onto $\mathcal{H}_{\mathrm{pp}}( H_V )$ and $\mathcal{H}_{ \mathrm{ac}Ê}( H )$, respectively. Since, according to Hypothesis \ref{V-1}, $\Pi_{Ê\mathrm{pp} }( H_V )$ is compact, we know that $\Pi_{ \mathrm{pp} }( H_V ) e^{ - i t H } \Pi_{ \mathrm{ac} } ( H ) \to 0$ strongly, as $t \to \infty$, by Lemma \ref{lm:compwn}. Therefore is suffices to prove the existence of $ \slim e^{ i t H_0 } \Pi_{ \mathrm{ac} }( H_V ) e^{ - i t H } \Pi_{ \mathrm{ac} } ( H ) $, as $t \to \infty$.
We write
\begin{equation}
e^{ i t H_0 } \Pi_{ \mathrm{ac} }( H_V ) e^{ - i t H } \Pi_{ \mathrm{ac} } ( H ) = e^{ i t H_0 } e^{ - i t H_V } \Pi_{ \mathrm{ac} }( H_V ) e^{ i t H_V } e^{ - i t H } \Pi_{ \mathrm{ac} } ( H ). \label{eq:fjdkf2}
\end{equation}
Applying Cook's argument exactly as in the proof of Proposition \ref{prop:existence_W-}, and using Hypothesis \ref{V1}, one verifies that
\begin{equation*}
\underset{t\to \infty }{\slim} \, \Pi_{ \mathrm{ac} }( H_V ) e^{ i t H_V }  e^{ - i t H } \Pi_{ \mathrm{ac} } ( H ) =:  W_+ ( H_V , H ) 
\end{equation*}
exists. Together with \eqref{eq:fjdkf2} and Hypothesis \ref{V0}, this shows that $\slim e^{ i t H_0 } \Pi_{ \mathrm{ac} }( H_V ) e^{ - i t H }$ exists, as $t \to \infty$, and that
\begin{equation*}
W_+( H_0 , H ) = \underset{t\to \infty }{\slim} \, e^{ i t H_0 } \Pi_{ \mathrm{ac} }( H_V ) e^{ - i t H } \Pi_{ \mathrm{ac} } ( H ) = W_+( H_0 , H_V ) W_+( H_V , H ).
\end{equation*}

The facts that $W_+( H_0 , H )$ and $W_-( H_0 , H^* )$ are contractions follow from the contractivity of $\{ e^{ - i t H }Ê\}_{ t \ge 0 }$, $\{ e^{ i t H^* }Ê\}_{ t \ge 0 }$ and unitarity of $\{ e^{ - i t H_0 }Ê\}_{ t \in \mathbb{R} }$.

To verify that the range of $W_+( H_0 , H )$ is dense in $\mathcal{H}$, we observe that, by \eqref{eq:HacHbperp}--\eqref{eq:1stincl},
\begin{equation*}
\mathrm{Ran}( W_+( H^* , H_0 ) ) \subset \mathcal{H}_{ \mathrm{ac} }( H^* ) = \mathcal{H}_{ \mathrm{ac} }( H ) .
\end{equation*}
This yields $W_+( H^* , H_0 ) = \Pi_{ \mathrm{ac} }( H ) W_+( H^* , H_0 )$, from which one easily deduces that
\begin{equation}\label{eq:adj-waves}
W_+( H_0 , H )^* = W_+( H^* , H_0 ).
\end{equation}
Since $W_+( H^* , H_0 )$ is injective by Proposition \ref{prop:existence_W-}, this shows that $\mathrm{Ran}( W_+( H_0 , H ) )$ is dense in $\mathcal{H}$. Likewise, $\mathrm{Ran}( W_-( H_0 , H^* ) )$ is dense in $\mathcal{H}$.

Next, we prove \eqref{eq:ker_W+}. The definitions of $W_+( H_0 , H )$ and $\mathcal{H}_{\mathrm{d}}( H )$ (see \eqref{eq:defHd}) imply that
\begin{equation*}
\mathrm{Ker} ( W_+( H_0 , H ) ) = \mathcal{H}_{ \mathrm{ac} }( H )^\perp \oplus \big ( \mathcal{H}_{ \mathrm{ac} }( H ) \cap \mathcal{H}_{\mathrm{d}}( H ) \big ) .
\end{equation*}
Since $\mathcal{H}_{ \mathrm{ac} }( H )^\perp = \mathcal{H}_{ \mathrm{b} }( H )$ (see \eqref{eq:HacHbperp}), and since it is easy to verify that $\mathcal{H}_{\mathrm{d}}( H ) \subset \mathcal{H}_{ \mathrm{b} }( H )^\perp$, the latter equation gives
\begin{equation*}
\mathrm{Ker} ( W_+( H_0 , H ) ) = \mathcal{H}_{ \mathrm{b} }( H ) \oplus \mathcal{H}_{\mathrm{d}}( H ).
\end{equation*}
The same arguments apply to $W_-( H_0 , H^* )$ instead of $W_+( H_0 , H )$.

Finally, the intertwining properties \eqref{eq:inter-3}--\eqref{eq:inter-4} follow from standard arguments (see, e.g., \cite{RS-III}).
\end{proof}
Now we prove Proposition \ref{prop:randense_W-}.

\begin{proof}[Proof of Proposition \ref{prop:randense_W-}]
To prove Proposition \ref{prop:randense_W-}, it suffices to use \eqref{eq:adj-waves}, which implies that
\begin{equation*}
\overline{\mathrm{Ran} ( W_+( H^* , H_0 ) )} = \mathrm{Ker}Ê( W_+( H_0 , H ) )^\perp = \big ( \mathcal{H}_{ \mathrm{b} }( H ) \oplus \mathcal{H}_{\mathrm{d}}( H ) \big )^\perp ,
\end{equation*}
where the last equality follows from Proposition \ref{prop:existence_W+}. The same arguments apply to $W_-( H , H_0 )$ instead of $W_+( H^* , H_0 )$. 
\end{proof}
Finally we prove Propositions \ref{prop:existencescatt} and \ref{prop:invertscatt}. Recall that the scattering operators $S( H , H_0 )$ and $S( H^* , H_0 )$ are defined in \eqref{eq:defscattop}. 
\begin{proof}[Proof of Proposition \ref{prop:existencescatt}]
Existence and contractivity of $S( H , H_0 )$ and $S( H^* , H_0 )$ are obvious consequences of Propositions \ref{prop:existence_W-} and \ref{prop:existence_W+}. The relation $S( H , H_0 )^* = S( H^* , H_0 )$ follows directly from the definitions involved. 
\end{proof}
\begin{proof}[Proof of Proposition \ref{prop:invertscatt}]
We prove that $(i) \Rightarrow (ii)$. Since $S(H , H_0 )$ is bijective, there exists $m > 0$ such that, for all $u \in \mathcal{H}$, 
\begin{equation*}
\|ÊS( H , H_0 ) u \| = \| W_+( H_0 , H ) W_- ( H , H_0 ) u \| \ge m \| u \|.
\end{equation*}
Since $W_+( H_0 , H )$ is a contraction, this implies that $\| W_-( H , H_0 ) u \|Ê\ge m \| u \|$, and therefore $W_-( H, H_0 )$ has closed range. By Proposition \ref{prop:randense_W-}, this yields
\begin{equation*}
Ê\mathrm{Ran}( W_-( H , H_0 ) ) = \big ( \mathcal{H}_{ \mathrm{b} }( H ) \oplus \mathcal{H}_{ \mathrm{d} }( H^* ) \big )^\perp.
\end{equation*}
In the same way the bijectivity of $S( H_0 , H^* )$ implies that $\mathrm{Ran}( W_+( H^* , H_0 ) ) = \big ( \mathcal{H}_{ \mathrm{b} }( H ) \oplus \mathcal{H}_{ \mathrm{d} }( H ) \big )^\perp$.

Next, we prove that $(ii) \Rightarrow (i)$. Proposition \ref{prop:existence_W-} shows that $W_-( H , H_0 )$ and $W_+( H^* , H_0 )$ are injective with closed ranges, so that there exists $m > 0$ such that, for all $u \in \mathcal{H}$, $\| W_-( H , H_0 ) u \|Ê\ge m \| u \|$ and $\| W_+( H^* , H_0 ) u \|Ê\ge m \| u \|$. Since
\begin{align*}
\| S( H , H_0 ) u \| &= \lim_{Êt \to \infty } \|Êe^{ i t H_0 }ÊW_-( H , H_0 ) e^{ - i t H_0 }Êu \| \\
&\ge \lim_{Êt \to \infty } m \| e^{ - i t H_0 }Êu \| = m \|Êu \| ,
\end{align*}
we deduce that $S( H , H_0 )$ is also injective with closed range. Here we have used the intertwining property \eqref{eq:inter-1} and unitarity of $e^{ - i t H }$. By the same argument, this shows that  $S( H^* , H_0 )$ is also injective with closed range. Since $S( H , H_0 )^* = S( H^* , H_0 )$ according to Proposition \ref{prop:existencescatt}, this proves that $S( H , H_0 )$ and $S( H^* , H_0 )$ are bijective.
\end{proof}

\section{Spectral projections for non self-adjoint operators}\label{app:proj}

In this section we establish the properties of the spectral projection $E_H( I )$ (see \eqref{eq:spec-proj}) stated in Proposition \ref{prop:spectr-sing}. Next, we prove \eqref{eq:weak-Gamma}.

\begin{proof}[Proof of Proposition \ref{prop:spectr-sing}]

Let $I \subset [ 0 , \infty )$ be a closed interval not containing any spectral singularities of $H$. We first prove that the weak limit defining $E_H( I )$ exists in $\mathcal{L}( \mathcal{H} )$. By the uniform boundedness principle, it suffices to prove that
\begin{equation*}
E_H( I ) := \underset{ \varepsilon \downarrow 0 }{\lim} \int_I \big \langle u , \big ( ( H - ( \lambda + i \varepsilon ) )^{-1} - ( H - ( \lambda - i \varepsilon ) )^{-1} \big ) v \big \rangle d \lambda 
\end{equation*}
exists for all $u ,v \in \mathcal{H}$.  Using the resolvent equation, we have that
\begin{align}
( H - ( \lambda \pm i \varepsilon ) )^{-1} &= ( H_V - ( \lambda \pm i \varepsilon ) )^{-1} + i ( H - ( \lambda \pm i \varepsilon ) )^{-1} C^*C ( H_V - ( \lambda \pm i \varepsilon ) )^{-1} \label{eq:resolv1} \\
&= ( H_V - ( \lambda \pm i \varepsilon ) )^{-1} + i ( H_V - ( \lambda \pm i \varepsilon ) )^{-1} C^*C ( H_V - ( \lambda \pm i \varepsilon ) )^{-1} \notag \\
& - ( H_V - ( \lambda \pm i \varepsilon ) )^{-1} C^*C ( H - ( \lambda \pm i \varepsilon ) )^{-1} C^*C ( H_V - ( \lambda \pm i \varepsilon ) )^{-1} . \label{eq:resolv2}
\end{align}
Since $H_V$ is a self-adjoint operator and $I \subset [ 0 , \infty )$ does not contain any eigenvalues of $H_V$ by Hypothesis \ref{V-1}, Stone's formula implies that
\begin{equation*}
\lim_{ \varepsilon \downarrow 0 } Ê\frac{1}{ 2 i \pi } \int_I \big \langle u , \big ( ( H_V - ( \lambda + i \varepsilon ) )^{-1} - ( H_V - ( \lambda - i \varepsilon ) )^{-1} \big ) v \big \rangle d \lambda = \langle u , E_{H_V}( I ) v \rangle ,
\end{equation*}
exists for all $u, v \in \mathcal{H}$, where $E_{H_V}( I )$ denotes the usual spectral projection for $H_V$.

Next, we show that 
\begin{align*}
& \int_I \big \langle u , \big ( ( H - ( \lambda + i \varepsilon ) )^{-1} - ( H_V - ( \lambda + i \varepsilon ) )^{-1} \big ) v \big \rangle d \lambda \\
& =  i \int_I \big \langle C ( H^* - ( \lambda - i \varepsilon ) )^{-1} u , C ( H_V - ( \lambda + i \varepsilon ) )^{-1} v \big \rangle d \lambda
\end{align*}
converges, as $\varepsilon \to 0^+$. Here we have used \eqref{eq:resolv1}.  By  \eqref{eq:alpha2} and Parseval's theorem, we have that
\begin{equation*}
\int_{ \mathbb{R} } \big \|ÊC ( H^* - ( \lambda - i \varepsilon ) )^{-1} u \big \|^2 d \lambda = 2 \pi \int_0^\infty e^{ - 2 \varepsilon t }Ê\big \|ÊC e^{ i t H^* }Êu \big \|^2 \le \pi \|Êu \|^2 ,
\end{equation*}
for all $\varepsilon > 0$ and $u \in \mathcal{H}$. Therefore Lebesgue's monotone convergence theorem implies that $\lambda \mapsto C ( H^* - ( \lambda - i \varepsilon ) )^{-1} u$ converges in $L^2 ( \mathbb{R} ; \mathcal{H} )$, as $\varepsilon \to 0^+$, to a limit denoted by $\lambda \mapsto C ( H^* - ( \lambda - i 0^+ ) )^{-1} u$. Since $H_V$ has only finitely many eigenvalues, all negative, according to Hypothesis \ref{V-1}, and since $I \subset [ 0 , \infty )$, it follows that, for all $v \in \mathcal{H}$, $\lambda \mapsto C ( H_V - ( \lambda + i \varepsilon ) )^{-1} \Pi_{ \mathrm{pp} }( H_V ) v$ converges to $\lambda \mapsto C ( H_V - \lambda )^{-1} \Pi_{ \mathrm{pp} }( H_V ) v$ in $L^2 ( I ; \mathcal{H} )$, as $\varepsilon \to 0^+$. Moreover, using Hypothesis \ref{V1}, we deduce from the same argument as above that $\lambda \mapsto C ( H_V - ( \lambda + i \varepsilon ) )^{-1} \Pi_{Ê\mathrm{ac} }( H_V ) v$ converges to $\lambda \mapsto C ( H_V - ( \lambda + i 0^+ ) )^{-1} \Pi_{Ê\mathrm{ac} }( H_V ) v$ in $L^2( \mathbb{R} ; \mathcal{H} )$, as $\varepsilon \to 0^+$. Hence we have that
\begin{align}
& \lim_{ \varepsilon \downarrow 0 } \int_I \big \langle u , \big ( ( H - ( \lambda + i \varepsilon ) )^{-1} - ( H_V - ( \lambda + i \varepsilon ) )^{-1} \big ) v \big \rangle d \lambda \notag \\
& = i \int_I \big \langle C ( H^* - ( \lambda - i 0^+ ) )^{-1} u , C ( H_V - ( \lambda + i 0^+ ) )^{-1} v \big \rangle d \lambda. \label{eq:aabb1}
\end{align}

It remains to verify that 
\begin{align*}
& \int_I \big \langle u , \big ( ( H - ( \lambda - i \varepsilon ) )^{-1} - ( H_V - ( \lambda - i \varepsilon ) )^{-1} \big ) v \big \rangle d \lambda \\
& = i \int_I \big \langle u , ( H_V - ( \lambda - i \varepsilon ) )^{-1} C^*C ( H_V - ( \lambda - i \varepsilon ) )^{-1} v \big \rangle d \lambda \\
& \quad - \int_I \big \langle u , ( H_V - ( \lambda - i \varepsilon ) )^{-1} C^* C ( H - ( \lambda - i \varepsilon ) )^{-1} C^*C ( H_V - ( \lambda - i \varepsilon ) )^{-1} v \big \rangle d \lambda
\end{align*}
converges, as $\varepsilon \to 0^+$. Here we have used \eqref{eq:resolv2}. The existence of the limit as $\varepsilon$ goes to $0$ follows in the same way as in \eqref{eq:aabb1}, using in addition that 
\begin{align}
\sup_{Ê\lambda \in I }Ê\lim_{ \varepsilon \downarrow 0 } \big \| C ( H - ( \lambda - i \varepsilon ) )^{-1} C^* \big \| < \infty , \label{eq:sup-proj}
\end{align}
since $H$ does not have spectral singularities in $I$. Hence 
\begin{align*}
& \lim_{Ê\varepsilon \downarrow 0 }  \int_I \big \langle u , \big ( ( H - ( \lambda - i \varepsilon ) )^{-1} - ( H_V - ( \lambda - i \varepsilon ) )^{-1} \big ) v \big \rangle d \lambda \\
& = i \int_I \big \langle C ( H_V - ( \lambda + i 0^+ ) )^{-1} u , C ( H_V - ( \lambda - i 0^+ ) )^{-1} v \big \rangle d \lambda \\
& \quad -  \int_I \big \langle C ( H_V - ( \lambda + i 0^+ ) )^{-1} u , C ( H - ( \lambda - i 0^+ ) )^{-1} C^*C ( H_V - ( \lambda - i 0^+ ) )^{-1} v \big \rangle d \lambda .
\end{align*}

Summing up, we have proven that
\begin{equation*}
\lim_{ \varepsilon \downarrow 0 } Ê\frac{1}{ 2 i \pi } \int_I \big \langle u , \big ( ( H - ( \lambda + i \varepsilon ) )^{-1} - ( H - ( \lambda - i \varepsilon ) )^{-1} \big ) v \big \rangle d \lambda , 
\end{equation*}
exists for all $u,v \in \mathcal{H}$. This proves the existence of $E_H( I ) \in \mathcal{L}( \mathcal{H} )$ as the weak limit of the right side of \eqref{eq:spec-proj}. The fact that the weak limit in the right side of \eqref{eq:spec-proj-adj} exists in $\mathcal{L}( \mathcal{H} )$ follows in the same way. The equality $E_H( I )^* = E_{H^*}( I )$ is a direct consequence of the definitions.

Next, we prove \eqref{eq:spec-proj-exp}. Note that the arguments above show that the weak limit in the right side of \eqref{eq:spec-proj-exp} exists in $\mathcal{L}( \mathcal{H} )$. We begin by proving \eqref{eq:spec-proj-exp} for a compact interval $I$. Let
\begin{equation*}
f( t ) := \underset{ \varepsilon \downarrow 0 }{\wlim} \frac{1}{ 2 i \pi } \int_I e^{ i t ( \lambda - H ) } \big ( ( H - ( \lambda + i \varepsilon ) )^{-1} - ( H - ( \lambda - i \varepsilon ) )^{-1} \big ) d \lambda .
\end{equation*}
Obviously, $f(0) = E_H( I )$, and therefore, to prove \eqref{eq:spec-proj-exp}, it suffices to verify that the derivative of $f$ vanishes. We compute
\begin{align}
 \partial_t f(t) &= \underset{ \varepsilon \downarrow 0 }{\wlim} \frac{1}{ 2 \pi } \int_I e^{ i t ( \lambda - H ) } ( \lambda - H ) \big ( ( H - ( \lambda + i \varepsilon ) )^{-1} - ( H - ( \lambda - i \varepsilon ) )^{-1} \big ) d \lambda \notag \\
& = \underset{ \varepsilon \downarrow 0 }{\wlim} \frac{\varepsilon}{ 2 i \pi } \int_I e^{ i t ( \lambda - H ) } \big ( ( H - ( \lambda + i \varepsilon ) )^{-1} +  ( H - ( \lambda - i \varepsilon ) )^{-1} \big ) d \lambda. \label{eq:partialf}
\end{align}
We claim that there exists a positive constant $\mathrm{c}_I$, depending on the interval $I$, such that, for all $u \in \mathcal{H}$,
\begin{equation}\label{eq:mdy1}
\int_I \big \| \big ( H - ( \lambda \pm i \varepsilon ) \big )^{-1} u \big \| d \lambda \le \mathrm{c}_I \varepsilon^{-\frac12} \| u \|.
\end{equation}
Indeed, by the Cauchy-Schwarz inequality,
\begin{align*}
 \int_I \big \| \big ( H - ( \lambda \pm i \varepsilon ) \big )^{-1} u \big \| d \lambda & \le |I|^{\frac12} \Big ( \int_I \big \| \big ( H - ( \lambda \pm i \varepsilon ) \big )^{-1} u \big \|^2 d \lambda \Big )^{\frac12} \\
& = |I|^{\frac12} \Big ( \int_I \big \langle u , \big ( H^* - ( \lambda \mp i \varepsilon ) \big )^{-1} \big ( H - ( \lambda \pm i \varepsilon ) \big )^{-1} u \big \rangle d \lambda \Big )^{\frac12}.
\end{align*}
The resolvent equation gives
\begin{align*}
& \big ( H^* - ( \lambda \mp i \varepsilon ) \big )^{-1} - \big ( H - ( \lambda \pm i \varepsilon ) \big )^{-1} \\
& = - 2i \big ( H^* - ( \lambda \mp i \varepsilon ) \big )^{-1} C^* C \big ( H - ( \lambda \pm i \varepsilon ) \big )^{-1} \mp 2 i \varepsilon \big ( H^* - ( \lambda \mp i \varepsilon ) \big )^{-1} \big ( H - ( \lambda \pm i \varepsilon ) \big )^{-1} ,
\end{align*}
and therefore,
\begin{align*}
& \int_I \big \langle u , \big ( H^* - ( \lambda \mp i \varepsilon ) \big )^{-1} \big ( H - ( \lambda \pm i \varepsilon ) \big )^{-1} u \big \rangle d \lambda \\
&  = \pm \frac{ i }{ 2 \varepsilon } \int_I \big \langle u , \big ( H^* - ( \lambda \mp i \varepsilon ) \big )^{-1}  - \big ( H - ( \lambda \pm i \varepsilon ) \big )^{-1} u \big \rangle d \lambda \mp \frac{1}{ \varepsilon } \int_I \big \| CÊ\big ( H - ( \lambda \pm i \varepsilon ) \big )^{-1} u \big \|^2 d \lambda .
\end{align*}
Using the resolvent equations \eqref{eq:resolv1}--\eqref{eq:resolv2} and the same arguments as above, it is not difficult to verify that 
\begin{equation*}
\int_I \big | \big \langle u , \big ( H^* - ( \lambda \mp i \varepsilon ) \big )^{-1}  - \big ( H - ( \lambda \pm i \varepsilon ) \big )^{-1} u \big \rangle \big |  d \lambda \le \mathrm{c} \| u \|^2 ,
\end{equation*}
for some positive constant $\mathrm{c}$. Likewise, the resolvent equation \eqref{eq:resolv1} together with Hypothesis \ref{V1} and the fact that $H$ does not have spectral singularities in $I$ implies that 
\begin{equation*}
\int_I \big \| CÊ\big ( H - ( \lambda \pm i \varepsilon ) \big )^{-1} u \big \|^2 \le \mathrm{c} \| u \|^2.
\end{equation*}
This proves \eqref{eq:mdy1} and hence, by \eqref{eq:partialf}, that $\partial_t f( t ) = 0$. This establishes \eqref{eq:spec-proj-exp} for any compact interval $I$ not containing any spectral singularities. If $I$ is unbounded, it suffices to approximate $I$ by a sequence of compact intervals $I_n \subset I$ and use that both sides of \eqref{eq:spec-proj-exp}, with $I$ replaced by $I_n$, define operators in $\mathcal{L}( H )$ uniformly bounded in $n \in \mathbb{N}$ as follows from the arguments above.

It remains to prove \eqref{eq:proj-inters-int}. We write
\begin{align*}
E_H( I_1 ) E_H( I_2 ) = \underset{ \varepsilon \downarrow 0 }{\wlim} \frac{1}{ 2 i \pi } \int_{I_1} \big ( ( H - ( \lambda + i \varepsilon ) )^{-1} - ( H - ( \lambda - i \varepsilon ) )^{-1} \big ) E_H( I_2 ) d \lambda.
\end{align*}
Using the Laplace transform
\begin{align}
\big (  H - ( \lambda \pm i \varepsilon ) \big )^{-1} = i \int_0^\infty e^{ - \varepsilon t } e^{\pm i t \lambda } e^{ \mp i t H }Êdt , \label{eq:Laplace}
\end{align}
which holds in the strong sense on $\mathrm{Ran} ( E_H (I_2) )$ (as follows from \eqref{eq:spec-proj-exp}), we obtain from \eqref{eq:spec-proj-exp} together with Fubini's theorem that
\begin{align*}
E_H( I_1 ) E_H( I_2 ) &= \underset{ \varepsilon \downarrow 0 }{\wlim} \frac{1}{ 2 i \pi } \int_{I_2} \mathds{1}_{ I_1 }( \lambda' ) \big ( ( H - ( \lambda' + i \varepsilon ) )^{-1} - ( H - ( \lambda' - i \varepsilon ) )^{-1} \big ) d \lambda' \\
& = \underset{ \varepsilon \downarrow 0 }{\wlim} \frac{1}{ 2 i \pi } \int_{I_1 \cap I_2} \big ( ( H - ( \lambda' + i \varepsilon ) )^{-1} - ( H - ( \lambda' - i \varepsilon ) )^{-1} \big ) d \lambda' = E_H (I_1 \cap I_2 ).
\end{align*}
This concludes the proof.
\end{proof}
Next, we prove \eqref{eq:weak-Gamma}.  Recall that we want to compute the weak limit
\begin{align*}
 \underset{\varepsilon \downarrow 0}{\wlim} \frac{1}{2i\pi} \oint_{\Gamma_{\varepsilon}} \mu^4 \prod_{ j = 1 }^n ( \mu - \mu_j )^{\nu_j} ( R - \mu Ê)^{-1}  d \mu ,
\end{align*}
where $R = ( H - i )^{-1}$, $\mu_j = ( \lambda_j - i )^{-1}$, $\{ \lambda_j \}_{ j = 1 }^n$ are the spectral singularities of $H$, and $\Gamma_{\varepsilon} = \Gamma_{1,\varepsilon} \cup \Gamma_{2,\varepsilon} \cup \Gamma_{3,\varepsilon} \cup \Gamma_{4,\varepsilon}$ is the curve oriented counterclockwise defined in \eqref{eq:def-Gammaeps}.
\begin{proof}[Proof of \eqref{eq:weak-Gamma}]
We decompose
\begin{align*}
 \underset{\varepsilon \downarrow 0}{\wlim} \oint_{\Gamma_{\varepsilon}} \mu^4 \prod_{ j = 1 }^n ( \mu - \mu_j )^{\nu_j} ( R - \mu Ê)^{-1}  d \mu = \sum_{k=1}^4  \underset{\varepsilon \downarrow 0}{\wlim} \oint_{\Gamma_{k,\varepsilon}} \mu^4 \prod_{ j = 1 }^n ( \mu - \mu_j )^{\nu_j} ( R - \mu Ê)^{-1}  d \mu ,
\end{align*}
and compute each limit separately. We will show that the integrals over $\Gamma_{2,\varepsilon}$ and $\Gamma_{4,\varepsilon}$ vanish in the limit $\varepsilon \to 0^+$, while the sum of the integrals over $\Gamma_{1,\varepsilon}$ and $\Gamma_{3,\varepsilon}$ converges to the expected limit. We begin by computing the limit of the integral over $\Gamma_{3,\varepsilon}$, which is easier than that over $\Gamma_{1,\varepsilon}$.

\vspace{0,2cm}

\noindent \textbf{Limit of the integral over $\Gamma_{3,\varepsilon}$.}
The integral over $\Gamma_{3,\varepsilon}$ is given by
\begin{align}
& \int_{\Gamma_{3,\varepsilon}} \mu^4 \prod_{ j = 1 }^n ( \mu - \mu_j )^{\nu_j+1} ( R - \mu Ê)^{-1} d \mu \notag \\
& = \int_{ - \frac12 \mathrm{e}_0 }^{ \gamma_3( \varepsilon ) } ( \lambda - i + i \varepsilon )^{-6} \prod_{ j = 1 }^n \big ( ( \lambda - i + i \varepsilon )^{-1} - \mu_j \big )^{\nu_j} \big ( R - ( \lambda - i + i \varepsilon )^{-1}Ê\big )^{-1} d \lambda \notag\\
& =- \int_{ - \frac12 \mathrm{e}_0 }^{ \gamma_3( \varepsilon ) } ( \lambda - i + i \varepsilon )^{-4} \prod_{ j = 1 }^n \big ( ( \lambda - i + i \varepsilon )^{-1} - \mu_j \big )^{\nu_j} \big ( H - ( \lambda + i \varepsilon ) Ê\big )^{-1}  d \lambda \notag \\
& \quad - \int_{ - \frac12 \mathrm{e}_0 }^{ \gamma_3( \varepsilon ) } ( \lambda - i + i \varepsilon )^{-5} \prod_{ j = 1 }^n \big ( ( \lambda - i + i \varepsilon )^{-1} - \mu_j \big )^{\nu_j}  d \lambda , \label{eq:rjc}
\end{align}
where we used \eqref{eq:resolv_R*} in the second equality. We show that the integral over $[ 0 , \gamma_3( \varepsilon ) ]$ can be replaced by the integral over $[ 0 , \infty )$ up to a term that vanishes, as $\varepsilon \to 0^+$. Using that $ \lambda \mapsto ( \lambda - i + i \varepsilon )^{-1} - \mu_j$ is bounded on $[ - \frac12 \mathrm{e}_0 , \infty )$ and that $| \lambda - i + i \varepsilon |^{-1} \le 2 (\lambda^2 + 1)^{-1/2}$, for $\varepsilon > 0$ small enough, one easily verifies that
\begin{align}
& \int_{ - \frac12 \mathrm{e}_0 }^{\gamma_3( \varepsilon )} ( \lambda - i + i \varepsilon )^{-5} \prod_{ j = 1 }^n \big ( ( \lambda - i + i \varepsilon )^{-1} - \mu_j \big )^{\nu_j}  d \lambda \notag \\
&= \int_{ - \frac12 \mathrm{e}_0 }^\infty ( \lambda - i )^{-5} \prod_{ j = 1 }^n \big ( ( \lambda - i )^{-1} - \mu_j \big )^{\nu_j}  d \lambda + \mathcal{O}( \varepsilon ). \label{eq:ttii1}
\end{align}
To treat the first term in the right side of \eqref{eq:rjc}, we observe that $\| ( H - ( \lambda + i \varepsilon ) )^{-1} \| \le \varepsilon^{-1}$ since $- i H$ generates a contraction semigroup. Using that, for $\lambda \in [ \gamma_3( \varepsilon ) , \infty )$ and $\varepsilon > 0$ small enough, we have that $| ( \lambda - i + i \varepsilon )^{-4} | \le 2 \varepsilon (\lambda^2+1)^{-3/2}$, we deduce that
\begin{align}
& \int_{ - \frac12 \mathrm{e}_0 }^{ \gamma_3( \varepsilon ) } ( \lambda - i + i \varepsilon )^{-4} \prod_{ j = 1 }^n \big ( ( \lambda - i + i \varepsilon )^{-1} - \mu_j \big )^{\nu_j} \big ( H - ( \lambda + i \varepsilon ) Ê\big )^{-1}  d \lambda \notag \\
& = \int_{ - \frac12 \mathrm{e}_0 }^{ \infty } ( \lambda - i + i \varepsilon )^{-4} \prod_{ j = 1 }^n \big ( ( \lambda - i + i \varepsilon )^{-1} - \mu_j \big )^{\nu_j} \big ( H - ( \lambda + i \varepsilon ) Ê\big )^{-1}  d \lambda + \mathcal{O} ( \varepsilon ) , \label{eq:lse1}
\end{align}
where $\mathcal{O} ( \varepsilon )$ stands for a bounded operator whose norm is of order $\mathcal{O}( \varepsilon )$, as $\varepsilon \to 0^+$. Next, we claim that
\begin{equation}
\int_{  - \frac12 \mathrm{e}_0 }^\infty | \lambda - i + i \varepsilon |^{-2} \big | \big \langle u , \big ( H - ( \lambda + i \varepsilon ) \big )^{-1} v \rangle \big | d \lambda \le \mathrm{c} \varepsilon^{-\frac12} \| u \|Ê\|Êv \| , \label{eq:abf1}
\end{equation}
for some positive constant $\mathrm{c}$ and for all $u,v \in \mathcal{H}$. Indeed, by the Cauchy-Schwartz inequality, it suffices to establish that
\begin{equation}
\int_{  - \frac12 \mathrm{e}_0 }^\infty \big \|Ê\big ( H - ( \lambda + i \varepsilon ) \big )^{-1} v \big \|^2 d \lambda \le \mathrm{c}^2 \varepsilon^{-1}Ê\|Êv \|^2. \label{eq:abf2}
\end{equation}
Using the resolvent equation, we compute
\begin{align}
& \big \|Ê\big ( H - ( \lambda + i \varepsilon ) \big )^{-1} v \big \|^2 \notag \\
& = \big \langle v , \big ( H^* - ( \lambda - i \varepsilon ) \big )^{-1}Ê\big ( H - ( \lambda + i \varepsilon ) \big )^{-1} v \big \rangle \notag \\
& = ( 2 i \varepsilon )^{-1} \big\langle v , \big[\big ( H - ( \lambda + i \varepsilon ) \big )^{-1} - \big ( H^* - ( \lambda - i \varepsilon ) \big )^{-1} \big ] v \big \rangle -  \varepsilon^{-1} \big \|ÊC \big ( H^* - ( \lambda - i \varepsilon ) \big )^{-1} v \big \|^2 \notag \\
& = ( 2 i \varepsilon )^{-1} \big\langle v , \big[\big ( H_V - ( \lambda + i \varepsilon ) \big )^{-1} - \big ( H_V - ( \lambda - i \varepsilon ) \big )^{-1} \big ] v \big \rangle \notag \\
&\quad + ( 2 \varepsilon )^{-1} \mathrm{Re} \big ( \big\langle C \big ( H^* - ( \lambda - i \varepsilon ) \big )^{-1} v , C \big ( H_V - ( \lambda - i \varepsilon ) \big )^{-1} v \big \rangle \big )\notag \\
& \quad - \varepsilon^{-1} \big \|ÊC \big ( H^* - ( \lambda - i \varepsilon ) \big )^{-1} v \big \|^2 . \label{eq:abf3}
\end{align}
Since the eigenvalues of $H_V$ belong to $( - \infty , - \mathrm{e}_0 ]$, decomposing $v = v_{ \mathrm{pp} } + v_{ \mathrm{ac}Ê}$ with $v_{ \mathrm{pp}Ê} \in \mathcal{H}_{ \mathrm{pp} }( H_V )$ and $v_{ \mathrm{ac}Ê} \in \mathcal{H}_{ \mathrm{ac}Ê}( H_V )$, we obtain from Hypothesis \ref{V1} that
\begin{equation}
\int_{  - \frac12 \mathrm{e}_0 }^\infty \big \| C \big ( H_V - ( \lambda \pm i \varepsilon ) \big )^{-1} v \big \|^2 d \lambda \le \mathrm{c}^2 \| v \|^2 , \label{eq:ghz1}
\end{equation}
for some positive constant $\mathrm{c}$. Moreover, by \eqref{eq:alpha2} and Parseval's theorem, we also have that
\begin{equation*}
\int_{  - \frac12 \mathrm{e}_0 }^\infty \big \| C \big ( H^* - ( \lambda - i \varepsilon ) \big )^{-1} v \big \|^2 d \lambda \le \pi \| v \|^2.
\end{equation*}
The last two inequalities together with \eqref{eq:abf3} prove \eqref{eq:abf2}, and therefore \eqref{eq:abf1}. Using \eqref{eq:abf1}, we then deduce from \eqref{eq:lse1} that
\begin{align}
& \int_{ - \frac12 \mathrm{e}_0 }^{ \gamma_3( \varepsilon ) } ( \lambda - i + i \varepsilon )^{-4} \prod_{ j = 1 }^n \big ( ( \lambda - i + i \varepsilon )^{-1} - \mu_j \big )^{\nu_j} \big \langle u , \big ( H - ( \lambda + i \varepsilon ) Ê\big )^{-1} v \big \rangle d \lambda \notag \\
& = \int_{ - \frac12 \mathrm{e}_0 }^{ \infty } ( \lambda - i )^{-4} \prod_{ j = 1 }^n \big ( ( \lambda - i )^{-1} - \mu_j \big )^{\nu_j} \big \langle u , \big ( H - ( \lambda + i \varepsilon ) Ê\big )^{-1} v \big \rangle  d \lambda + \mathcal{O} ( \varepsilon^{\frac12} ) \|u \| \| v \| ,
\end{align}
for all $u , v \in \mathcal{H}$. Summing up, we have proven that
\begin{align}
& \int_{\Gamma_{3,\varepsilon}} \mu^4 \prod_{ j = 1 }^n ( \mu - \mu_j )^{\nu_j} \big \langle u , ( R - \mu Ê)^{-1} v \big \rangle d \mu \notag \\
& = - \int_{ - \frac12 \mathrm{e}_0 }^\infty ( \lambda - i )^{-5} \prod_{ j = 1 }^n \big ( ( \lambda - i )^{-1} - \mu_j \big )^{\nu_j}  d \lambda \langle u , v \rangle \notag \\
& \quad - \int_{ - \frac12 \mathrm{e}_0 }^{ \infty } ( \lambda - i )^{-4} \prod_{ j = 1 }^n \big ( ( \lambda - i )^{-1} - \mu_j \big )^{\nu_j} \big \langle u ,  \big ( H - ( \lambda + i \varepsilon ) Ê\big )^{-1} v \big \rangle d \lambda \notag \\
& \quad + \mathcal{O} ( \varepsilon^{ \frac12 } ) \| u \|Ê\|Êv \|. \label{eq:contrib1}
\end{align}

\vspace{0,2cm}

\noindent \textbf{Limit of the integral over $\Gamma_{1,\varepsilon}$.} To compute the limit of the integral over $\Gamma_{1,\varepsilon}$, we modify the argument above as follows. In the same way as for the integral over $\Gamma_{3 , \varepsilon }$ (see \eqref{eq:rjc} and \eqref{eq:ttii1}), we obtain that
\begin{align}
& \int_{\Gamma_{1,\varepsilon}} \mu^4 \prod_{ j = 1 }^n ( \mu - \mu_j )^{\nu_j} ( R - \mu Ê)^{-1} d \mu \notag \\
& =  \int_{ - \frac12 \mathrm{e}_0 }^{ \gamma_1( \varepsilon ) } ( \lambda - i - i \varepsilon )^{-4} \prod_{ j = 1 }^n \big ( ( \lambda - i - i \varepsilon )^{-1} - \mu_j \big )^{\nu_j} \big ( H - ( \lambda - i \varepsilon ) Ê\big )^{-1}  d \lambda \notag \\
& \quad + \int_{ - \frac12 \mathrm{e}_0 }^{ \infty } ( \lambda - i  )^{-5} \prod_{ j = 1 }^n \big ( ( \lambda - i )^{-1} - \mu_j \big )^{\nu_j}  d \lambda + \mathcal{O} (\varepsilon) . \label{eq:rjc4}
\end{align}
To estimate the norm of the operator $( H - ( \lambda - i \varepsilon ) )^{-1}$ for $\lambda$ large, we use twice the resolvent equation:
\begin{align}
& ( H - ( \lambda - i \varepsilon ) )^{-1} = ( H_V - ( \lambda - i \varepsilon ) )^{-1} + i ( H_V - ( \lambda - i \varepsilon ) )^{-1} C^*C ( H_V - ( \lambda - i \varepsilon ) )^{-1} \notag \\
& \quad - ( H_V - ( \lambda - i \varepsilon ) )^{-1} C^* C ( H - ( \lambda - i \varepsilon ) )^{-1} C^* C ( H_V - ( \lambda - i \varepsilon ) )^{-1} . \label{eq:twice-resolv}
\end{align}
By \eqref{eq:sup>m} in Hypothesis \ref{V3} and the fact that $H_V$ is self-adjoint, this implies that
\begin{equation}
\sup_{Ê\lambda \ge m } \big \| ( H - ( \lambda - i \varepsilon ) )^{-1} \big \| \le \mathrm{c} \varepsilon^{-2} , \label{eq:sup>m_2}
\end{equation}
for some positive constant $\mathrm{c}$. Hence, using that, for $\lambda \in [ \gamma_1( \varepsilon ) , \infty )$ and $\varepsilon > 0$ small enough, we have that $| ( \lambda - i - i \varepsilon )^{-4} | \le \varepsilon^2 ( \lambda^2 + 1)^{-1}$, we deduce that
\begin{align}
& \int_{ - \frac12 \mathrm{e}_0 }^{ \gamma_1( \varepsilon ) } ( \lambda - i - i \varepsilon )^{-4} \prod_{ j = 1 }^n \big ( ( \lambda - i - i \varepsilon )^{-1} - \mu_j \big )^{\nu_j} \big ( H - ( \lambda - i \varepsilon ) Ê\big )^{-1}  d \lambda \notag \\
& = \int_{ - \frac12 \mathrm{e}_0 }^{ \infty } ( \lambda - i - i \varepsilon )^{-4} \prod_{ j = 1 }^n \big ( ( \lambda - i - i \varepsilon )^{-1} - \mu_j \big )^{\nu_j} \big ( H - ( \lambda - i \varepsilon ) Ê\big )^{-1}  d \lambda + \mathcal{O} ( \varepsilon ) . \label{eq:lse4}
\end{align}
We introduce \eqref{eq:twice-resolv} into the right side of this equation, thus obtaining a sum of $3$ terms. For the first term, we observe that
\begin{equation*}
\int_{ - \frac12 \mathrm{e}_0 }^\infty | \lambda - i - i \varepsilon |^{-2} \big | \big \langle u ,Ê( H_V - ( \lambda - i \varepsilon ) )^{-1} v \rangle \big | d \lambda \le \mathrm{c} \varepsilon^{-\frac12} \| u \|Ê\|Êv \| ,
\end{equation*}
for all $u,v \in \mathcal{H}$ (by the Cauchy-Schwartz inequality and the same argument we used to establish \eqref{eq:abf2}). This yields
\begin{align}
& \int_{ - \frac12 \mathrm{e}_0 }^{ \infty } ( \lambda - i - i \varepsilon )^{-4} \prod_{ j = 1 }^n \big ( ( \lambda - i - i \varepsilon )^{-1} - \mu_j \big )^{\nu_j} \big \langle u , \big ( H_V - ( \lambda - i \varepsilon ) Ê\big )^{-1} v \big \rangle d \lambda \notag \\
& = \int_{ - \frac12 \mathrm{e}_0 }^{ \infty } ( \lambda - i )^{-4} \prod_{ j = 1 }^n \big ( ( \lambda - i )^{-1} - \mu_j \big )^{\nu_j} \big \langle u , \big ( H_V - ( \lambda - i \varepsilon ) Ê\big )^{-1} v \big \rangle d \lambda \notag \\
&\quad + \mathcal{O}( \varepsilon^{\frac12} ) \|Êu \|Ê\|Êv \|. \label{eq:lse7}
\end{align}
For the second term coming from the introduction of \eqref{eq:twice-resolv} into \eqref{eq:lse4}, we use \eqref{eq:ghz1}, which yields
\begin{align}
& \int_{ - \frac12 \mathrm{e}_0 }^{ \infty } ( \lambda - i - i \varepsilon )^{-4} \prod_{ j = 1 }^n \big ( ( \lambda - i - i \varepsilon )^{-1} - \mu_j \big )^{\nu_j} \notag \\
& \phantom{ \int_{ - \frac12 \mathrm{e}_0 }^{ \infty } }  \big \langle C \big ( H_V - ( \lambda + i \varepsilon ) Ê\big )^{-1}  u , C \big ( H_V - ( \lambda - i \varepsilon ) Ê\big )^{-1} v \big \rangle d \lambda \notag \\
& = \int_{ - \frac12 \mathrm{e}_0 }^{ \infty } ( \lambda - i )^{-4} \prod_{ j = 1 }^n \big ( ( \lambda - i )^{-1} - \mu_j \big )^{\nu_j} \big \langle C \big ( H_V - ( \lambda + i \varepsilon ) Ê\big )^{-1} u , \big ( H_V - ( \lambda - i \varepsilon ) Ê\big )^{-1} v \big \rangle d \lambda \notag \\
& \quad + \mathcal{O}( \varepsilon ) \|Êu \|Ê\|Êv \|. \label{eq:lse8}
\end{align}
For the last term coming from the introduction of \eqref{eq:twice-resolv} into \eqref{eq:lse4}, we note that
\begin{equation*}
( \lambda - i - i \varepsilon )^{-1} -  \mu_j = ( \lambda - i - i \varepsilon )^{-1} Ê- ( \lambda_j - i )^{-1} = ( \lambda - i - i \varepsilon )^{-1}  ( \lambda_j - i )^{-1} ( \lambda - i \varepsilon - \lambda_j ) .
\end{equation*}
Moreover, it follows from Hypothesis \ref{V3} that the map
\begin{equation*}
z \mapsto \prod_{ j = 1 }^n \frac{ ( z - \lambda_j )^{\nu_j} }{ ( z - i )^{\nu_j} } C (  H - z )^{-1} C^* \in \mathcal{L}( \mathcal{H} ) ,
\end{equation*}
is uniformly bounded in the region $\{ z \in \mathbb{C}, \mathrm{Re}(z) \ge - \frac12 \mathrm{e}_0 , - \varepsilon_0 < \mathrm{Im}(z) < 0 \}$, for $\varepsilon_0 > 0$ small enough. Combining this with \eqref{eq:ghz1} we obtain that
\begin{align}
& \int_{ - \frac12 \mathrm{e}_0 }^{ \infty } ( \lambda - i - i \varepsilon )^{-4} \prod_{ j = 1 }^n \big ( ( \lambda - i - i \varepsilon )^{-1} - \mu_j \big )^{\nu_j} \notag \\
& \qquad \big \langle C \big ( H_V - ( \lambda + i \varepsilon ) Ê\big )^{-1}  u , C \big ( H - ( \lambda - i \varepsilon ) \big )^{-1} C^* C \big ( H_V - ( \lambda - i \varepsilon ) Ê\big )^{-1} v \big \rangle d \lambda \notag \\
& = \int_0^{ \infty } ( \lambda - i )^{-4} \prod_{ j = 1 }^n \big ( ( \lambda - i )^{-1} - \mu_j \big )^{\nu_j} \notag \\
& \qquad \big \langle C \big ( H_V - ( \lambda + i \varepsilon ) Ê\big )^{-1}  u , C \big ( H - ( \lambda - i \varepsilon ) \big )^{-1} C^* C \big ( H_V - ( \lambda - i \varepsilon ) Ê\big )^{-1} v \big \rangle d \lambda  \phantom{\int} \notag \\
&\quad  + \mathcal{O}( \varepsilon ) \|Êu \|Ê\|Êv \|. \label{eq:lse9}
\end{align}
Equations \eqref{eq:rjc4}, \eqref{eq:twice-resolv}, \eqref{eq:lse4}, \eqref{eq:lse7}, \eqref{eq:lse8} and \eqref{eq:lse9} imply that
\begin{align}
& \int_{\Gamma_{1,\varepsilon}} \mu^4 \prod_{ j = 1 }^n ( \mu - \mu_j )^{\nu_j} \big \langle u , ( R - \mu Ê)^{-1} v \big \rangle d \mu \notag \\
& =  \int_{ - \frac12 \mathrm{e}_0 }^\infty ( \lambda - i )^{-5} \prod_{ j = 1 }^n \big ( ( \lambda - i )^{-1} - \mu_j \big )^{\nu_j}  d \lambda \langle u , v \rangle \notag \\
& \quad + \int_{ - \frac12 \mathrm{e}_0 }^{ \infty } ( \lambda - i )^{-4} \prod_{ j = 1 }^n \big ( ( \lambda - i )^{-1} - \mu_j \big )^{\nu_j} \big \langle u , \big ( H - ( \lambda - i \varepsilon ) Ê\big )^{-1} v \big \rangle  d \lambda  \notag \\
& \quad  + \mathcal{O} ( \varepsilon^{\frac12} ) \| u \|Ê\|Êv \| , \label{eq:contrib3}
\end{align}
for all $u,v \in \mathcal{H}$. 

 \vspace{0,2cm}

\noindent \textbf{Limit of the integral over $\Gamma_{2,\varepsilon}$.} Using \eqref{eq:resolv_R*}, we see that the integral over $\Gamma_{2,\varepsilon}$ is given by
\begin{align}
& \int_{\Gamma_{2,\varepsilon}} \mu^4 \prod_{ j = 1 }^n ( \mu - \mu_j )^{\nu_j} ( R - \mu Ê)^{-1} d \mu  = - \int_{ \theta_1( \varepsilon ) }^{ \theta_3( \varepsilon ) } \varepsilon^3 e^{ 3 i \thetaÊ} \prod_{ j = 1 }^n ( \varepsilon e^{ i \theta } - \mu_j )^{\nu_j} \varepsilon i e^{ i \thetaÊ} d \theta \notag \\
& - \int_{ \theta_1( \varepsilon ) }^{ \theta_3( \varepsilon ) } \varepsilon^2 e^{ 2 i \thetaÊ} \prod_{ j = 1 }^n ( \varepsilon e^{ i \theta } - \mu_j )^{\nu_j} \big ( H - ( \varepsilon^{-1} e^{ - i \theta } + i ) \big )^{-1} \varepsilon i e^{ i \thetaÊ} d \theta. \label{eq:contrib2}
\end{align}
Obviously the first term in the right side of \eqref{eq:contrib2} vanishes, as $\varepsilon \to 0$. To treat the second term, we decompose the integral into $3$ integrals, say over $[ \theta_1( \varepsilon ) , \theta_2 ]$, $[ \theta_2 , \theta_4 ]$ and $[ \theta_4 ,  \theta_3( \varepsilon ) ]$. The parameters $\theta_2$ and $\theta_4$ are chosen such that 
\begin{align*}
& \mathrm{Re} \big ( \varepsilon^{-1} e^{ - i \theta } + i \big ) \ge m , \quad \text{ for } \theta \in [ \theta_1( \varepsilon ) , \theta_2 ] , \\
& \mathrm{dist} \big ( \big ( \varepsilon^{-1} e^{ - i \theta } + i \big ) , \sigma( H_V ) \big ) > \| C^* C \| , \quad \text{ for } \theta \in [ \theta_2 , \theta_4 ] , \\
& \mathrm{Re} \big ( \varepsilon^{-1} e^{ - i \theta } + i \big ) \ge m , \quad \text{ for } \theta \in [ \theta_4 , \theta_3( \varepsilon ) ] .
\end{align*}
See Figure \ref{fig3}.
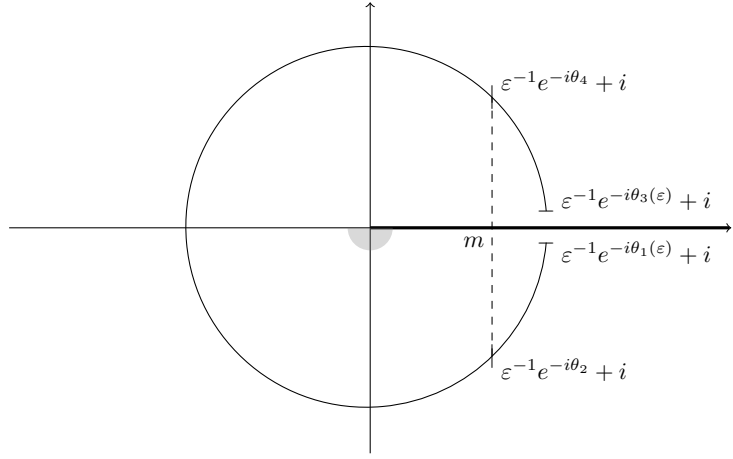
\begin{figure}[H] 
\begin{center}
\begin{tikzpicture}[scale=0.6, every node/.style={scale=0.9}]

   \draw[->](-4,2) -- (12,2);      
      \draw[->](4,-3) -- (4,7);      
  \draw[-, very thick] (4,2)--(11.95,2);
  
      \draw[-, dashed](6.7,-0.85) -- (6.7,4.9);     
      \draw (6.7,-0.85) node {\small $|$};
      \draw (6.7,4.9) node {\small $|$};
  
        \draw (7.9,2.37) arc (5:355:4cm);
      \draw (7.9,2.37) node {\small $-$};        
      \draw (7.9,1.67) node {\small $-$};        
              
         \draw (6.3,2) node[below] {\small $m$};

         \draw (9.9,1) node[above] {\small $\varepsilon^{-1} e^{ - i \theta_1( \varepsilon ) } + i $};    
         \draw (9.9,3.1) node[below] {\small $\varepsilon^{-1} e^{ - i \theta_3( \varepsilon ) } + i$};    
         \draw (8.3,-0.7) node[below] {\small $\varepsilon^{-1} e^{ - i \theta_2 } + i$};    
         \draw (8.3,4.8) node[above] {\small $\varepsilon^{-1} e^{ - i \theta_4 } + i$};    
                  
         \fill[gray, opacity=0.3] (3.5,2) arc (180:360:0.5cm);

        \end{tikzpicture}
\caption{ \footnotesize  \textbf{The set $\big \{ \varepsilon^{-1} e^{ - i \theta } + i , \theta \in [ \theta_1( \varepsilon ) , \theta_3( \varepsilon ) ] \big \}$.} The eigenvalues of $H$ and $H_V$ are contained in the grey semi-disc. The parameter $m$ satisfies \eqref{eq:sup_epsilon'}. }\label{fig3}
\end{center}
\end{figure}

Here $m>0$ satisfies
\begin{equation}\label{eq:sup_epsilon'}
\sup_{Ê\lambda \ge m } \big \| ( H - ( \lambda \pm i \varepsilon' ) )^{-1} \big \| \le \mathrm{c}_{ \varepsilon_0 } (\varepsilon')^{-2} ,
\end{equation}
for any $\varepsilon_0>0$ and $\varepsilon' \in ( 0 , \varepsilon_0 ]$, where $\mathrm{c}_{ \varepsilon_0 }$ is a positive constant depending on $\varepsilon_0$, according to Hypothesis \ref{V3} (see \eqref{eq:sup>m_2} for $( H - ( \lambda - i \varepsilon' ) )^{-1}$ and use that $- i H$ generates a semigroup of contractions for $( H - ( \lambda + i \varepsilon' ) )^{-1}$). For $\varepsilon > 0$ small enough, it is not difficult to verifiy that such choices of $\theta_2$ and $\theta_4$ are possible.

For $\theta \in [ \theta_1( \varepsilon ) , \theta_2 ]$, we have that $\mathrm{Im} ( \varepsilon^{-1} e^{ - i \theta } + i ) \le - \varepsilon $ (note that $\varepsilon^{-1} e^{ - i \theta_1( \varepsilon ) } + i = \gamma_1( \varepsilon ) - i \varepsilon$) and $\mathrm{Re} \big ( \varepsilon^{-1} e^{ - i \theta } + i \big ) \ge m$. Hence, by \eqref{eq:sup_epsilon'},
\begin{align}
 \int_{ \theta_1( \varepsilon ) }^{ \theta_2 } \varepsilon^2 e^{ 2 i \thetaÊ} \prod_{ j = 1 }^n ( \varepsilon e^{ i \theta } - \mu_j )^{\nu_j} \big ( H - ( \varepsilon^{-1} e^{ - i \theta } + i ) \big )^{-1} \varepsilon i e^{ i \theta } d \theta = \mathcal{O}( \varepsilon ), \quad \varepsilon \to 0^+. \label{eq:slcv1}
\end{align}
Similarly, for $ \theta \in [ \theta_4 , \theta_2( \varepsilon ) ]$, we have that $\mathrm{Im} ( \varepsilon^{-1} e^{ - i \theta } + i ) \ge \varepsilon $ and $\mathrm{Re} \big ( \varepsilon^{-1} e^{ - i \theta } + i \big ) \ge m$. Therefore \eqref{eq:sup_epsilon'} yields
\begin{align}
 \int_{ \theta_4 }^{ \theta_3( \varepsilon ) } \varepsilon^2 e^{ 2 i \thetaÊ} \prod_{ j = 1 }^n ( \varepsilon e^{ i \theta } - \mu_j )^{\nu_j} \big ( H - ( \varepsilon^{-1} e^{ - i \theta } + i ) \big )^{-1} \varepsilon i e^{ i \theta } d \theta = \mathcal{O}( \varepsilon ), \quad \varepsilon \to 0^+. \label{eq:slcv2}
\end{align}
For $\theta \in [ \theta_2 , \theta_4 ]$, we have $ \mathrm{dist}  ( ( \varepsilon^{-1} e^{ - i \theta } + i ) , \sigma( H_V ) ) > \| C^* C \|$ and we observe that
\begin{equation*}
\big \| ( H- z )^{-1} \big \|Ê\le 2 \, \mathrm{dist}( z , H_V )^{-1} ,
\end{equation*}
for any $z \in \mathbb{C}$ such that $\mathrm{dist}( z , \sigma ( H_V ) ) > \| C^* C \|$, as follows from the resolvent equation
\begin{equation*}
( H - z )^{-1}Ê= ( H_V - z )^{-1}Ê\big [ \mathrm{Id} + i C^* C ( H_V - z )^{-1} \big ]^{-1}.
\end{equation*}
This yields
\begin{align}
 \int_{ \theta_2 }^{ \theta_4 } \varepsilon^2 e^{ 2 i \thetaÊ} \prod_{ j = 1 }^n ( \varepsilon e^{ i \theta } - \mu_j )^{\nu_j} \big ( H - ( \varepsilon^{-1} e^{ - i \theta } + i ) \big )^{-1} \varepsilon i e^{ i \theta } d \theta = \mathcal{O}( \varepsilon^3 ), \quad \varepsilon \to 0^+. \label{eq:slcv3}
\end{align}
From \eqref{eq:contrib2}, \eqref{eq:slcv1}, \eqref{eq:slcv2} and \eqref{eq:slcv3}, we conclude that
\begin{align}
\underset{\varepsilon \downarrow 0}{\wlim} \int_{\Gamma_{2,\varepsilon}} \mu^4 \prod_{ j = 1 }^n ( \mu - \mu_j )^{\nu_j} ( R - \mu Ê)^{-1} d \mu =0. \label{eq:contrib2-conc}
\end{align}

\vspace{0,2cm}

\noindent \textbf{Limit of the integral over $\Gamma_{4,\varepsilon}$.} It follows from \eqref{eq:resolv_R*} that  the integral over $\Gamma_{4,\varepsilon}$ is given by
\begin{align}
& \int_{\Gamma_{4,\varepsilon}} \mu^4 \prod_{ j = 1 }^n ( \mu - \mu_j )^{\nu_j} ( R - \mu Ê)^{-1} d \mu \notag \\
& = \int_{ - \varepsilon }^{ \varepsilon } ( - \frac12 \mathrm{e}_0 - i + i x )^{-4} \prod_{ j = 1 }^n \big ( ( - \frac12 \mathrm{e}_0 - i + i x )^{-1} - \mu_j \big )^{\nu_j} \big ( H - ( - \frac12 \mathrm{e}_0 + i x ) Ê\big )^{-1}  d x \notag \\
& \quad + \int_{ - \varepsilon }^{ \varepsilon } ( - \frac12 \mathrm{e}_0 - i + i x )^{-5} \prod_{ j = 1 }^n \big ( ( - \frac12 \mathrm{e}_0 - i + i x )^{-1} - \mu_j \big )^{\nu_j}  d x . \label{eq:contrib4}
\end{align}
Obviously, the second term in the right side of \eqref{eq:contrib4} is of order $\mathcal{O}( \varepsilon )$, as $\varepsilon \to 0^+$. Moreover, for any $x \in [ - \varepsilon_0 , \varepsilon_0 ]$, where $\varepsilon_0$ is fixed sufficiently small (depending only on $H$), we have that $\mathrm{dist}( - \frac12 \mathrm{e}_0 + i x ; \sigma (H) ) \ge \varepsilon_0 /2$. This implies that
\begin{equation*}
\sup_{ x \in [ - \varepsilon_0 , \varepsilon_0 ] } \big \|Ê\big ( H - ( - \frac12 \mathrm{e}_0 + i x ) Ê\big )^{-1} \big \| \le \mathrm{c} ,
\end{equation*}
where $\mathrm{c}$ is a positive constant depending only on $H$. Hence the first term in the right side of \eqref{eq:contrib4} is also of order $\mathcal{O}( \varepsilon )$, as $\varepsilon \to 0^+$. Therefore,
\begin{align}
\underset{\varepsilon \downarrow 0}{\wlim} \int_{\Gamma_{4,\varepsilon}} \mu^4 \prod_{ j = 1 }^n ( \mu - \mu_j )^{\nu_j} ( R - \mu Ê)^{-1} d \mu =0. \label{eq:contrib4-conc}
\end{align}
\vspace{0,2cm}

\noindent \textbf{Conclusion.} Putting together \eqref{eq:contrib1} and \eqref{eq:contrib3} gives
\begin{align}
& \int_{\Gamma_{1,\varepsilon}} \mu^4 \prod_{ j = 1 }^n ( \mu - \mu_j )^{\nu_j} \big \langle u , ( R - \mu Ê)^{-1} v \big \rangle d \mu + \int_{\Gamma_{3,\varepsilon}} \mu^4 \prod_{ j = 1 }^n ( \mu - \mu_j )^{\nu_j} \big \langle u , ( R - \mu Ê)^{-1} v \big \rangle d \mu \notag \\
& = \int_{ - \frac12 \mathrm{e}_0 }^{ \infty } ( \lambda - i )^{-4} \prod_{ j = 1 }^n \big ( ( \lambda - i )^{-1} - \mu_j \big )^{\nu_j} \big \langle u , \big [ \big ( H - ( \lambda - i \varepsilon )Ê\big )^{-1} - \big ( H - ( \lambda + i \varepsilon )Ê\big )^{-1} \big ] v \big \rangle  d \lambda \notag \\
& \quad + \mathcal{O} ( \varepsilon^{\frac12} ) \| u \|Ê\|Êv \| , \label{eq:contrib1-3}
\end{align}
for all $u , v \in \mathcal{H}$. Arguing exactly as in the proof of Proposition \ref{prop:spectr-sing}, using, instead of \eqref{eq:sup-proj}, that 
\begin{equation*}
\sup_{Ê\lambda \in [ - \frac12 \mathrm{e}_0 ,\infty) }Ê\lim_{ \varepsilon \downarrow 0 } \prod_{ j = 1 }^n \big ( ( \lambda - i )^{-1} - \mu_j \big )^{\nu_j} \big \|ÊC \big ( H - ( \lambda - i \varepsilon ) \big )^{-1} C^* \big \|Ê,
\end{equation*}
by Hypothesis \ref{V3}, 
we then deduce that
\begin{align}
& \underset{\varepsilon \downarrow 0}{\wlim} \Big ( \int_{\Gamma_{1,\varepsilon}} \mu^4 \prod_{ j = 1 }^n ( \mu - \mu_j )^{\nu_j} ( R - \mu Ê)^{-1} d \mu + \int_{\Gamma_{3,\varepsilon}} \mu^4 \prod_{ j = 1 }^n ( \mu - \mu_j )^{\nu_j} ( R - \mu Ê)^{-1}  d \mu \Big ) \notag \\
& = \underset{\varepsilon \downarrow 0}{\wlim} \int_{ - \frac12 \mathrm{e}_0 }^{ \infty } ( \lambda - i )^{-4} \prod_{ j = 1 }^n \big ( ( \lambda - i )^{-1} - \mu_j \big )^{\nu_j} \notag \\
& \phantom{ = \underset{\varepsilon \downarrow 0}{\wlim} \int_{ - \frac12 \mathrm{e}_0 }^{ \infty } }  \big [ \big ( H - ( \lambda - i \varepsilon )Ê\big )^{-1} - \big ( H - ( \lambda + i \varepsilon )Ê\big )^{-1} \big ]  d \lambda  \label{eq:conclusion-lim}
\end{align}
exists in $\mathcal{L}( \mathcal{H} )$. Moreover, the integral over $[ - \frac12 \mathrm{e}_0 , \infty )$ can be replaced by the integral over $[ 0 , \infty )$ since, for $\lambda \in [ - \frac12 \mathrm{e}_0 , 0 )$, we have that $( H - ( \lambda \pm i 0^+ ) )^{-1} = ( H - \lambda )^{-1}$. Using in addition \eqref{eq:contrib2-conc} and \eqref{eq:contrib4-conc}, we obtain \eqref{eq:weak-Gamma}. The fact that \eqref{eq:unifbound_tildeE} holds is a consequence of the arguments above. This concludes the proof.
\end{proof}

\end{document}